\documentclass[a4paper,onecolumn,11pt,accepted=2023-05-11]{quantumarticle}
\pdfoutput=1
\usepackage[utf8]{inputenc}
\usepackage[english]{babel}
\usepackage[T1]{fontenc}
\usepackage{amsmath}

\usepackage{tikz}
\usepackage{lipsum}

\usepackage{fullpage}

\usepackage[unicode=true]{hyperref}
\usepackage{amsmath}
\usepackage{cite}
\usepackage{amsfonts}
\usepackage{amssymb}
\usepackage{amsthm}
\usepackage{authblk}

\usepackage{mypack} 
\usepackage{adjustbox}
\usepackage{soul}
\usepackage{comment}
\usepackage{bbold}
\usepackage{tikz}
\usetikzlibrary{decorations.pathreplacing,angles,quotes}
\usepackage{bbold}

\usepackage{diagbox}

\usetikzlibrary{matrix,arrows,decorations.pathmorphing}
\usepackage{tikz-cd}
\usetikzlibrary{arrows} 

\usetikzlibrary{decorations.markings}
 \usepackage{caption}
\usepackage{subcaption}
 
 \usepackage{pdfpages}

\usepackage{centernot}
\usepackage{mathtools}
\usepackage{stmaryrd}

\definecolor{bluegray}{rgb}{0.4, 0.6, 0.8}
\definecolor{turquoise}{rgb}{0.2, 0.7, 0.6}
\newcommand\co[1]{{\color{black}#1}}  
\newcommand\rev[1]{{\color{black}#1}} 

\usepackage{soul}  
  
\title{Simplicial quantum contextuality}

\author{Cihan Okay}
\email{cihan.okay@bilkent.edu.tr}
\affiliation{Department of Mathematics, Bilkent University, Ankara, Turkey}
\orcid{0000-0001-8097-5227}

\author{Aziz Kharoof}
\email{aziz.kharoof@bilkent.edu.tr} 
\orcid{0000-0002-4010-6526}

\author{Selman Ipek}
\email{selman.ipek@bilkent.edu.tr} 
\orcid{0000-0002-4475-4221}

\begin{document}
  \maketitle


  


\begin{abstract}
We introduce a new framework for contextuality based on simplicial sets, combinatorial models of topological spaces that play a prominent role in modern homotopy theory. Our approach extends measurement scenarios to consist of spaces (rather than sets) of measurements and outcomes, and thereby generalizes nonsignaling distributions to simplicial distributions, which are distributions on spaces modeled by simplicial sets. Using this formalism we present a topologically inspired new proof of Fine's theorem for characterizing noncontextuality in Bell scenarios. Strong contextuality is generalized suitably for simplicial distributions, allowing us to define cohomological witnesses that extend the earlier topological constructions restricted to algebraic relations among quantum observables to the level of probability distributions. Foundational theorems of quantum theory such as the Gleason's theorem and Kochen--Specker theorem can be expressed naturally within this new language.
\end{abstract}

\tableofcontents

\section{Introduction}

It has long been recognized that quantum mechanics stands in stark contrast to the classical physics that preceded it. Early pioneers, such as Schrodinger \cite{schrodinger35}, as well as Einstein, Podolsky, and Rosen (EPR) \cite{epr35}, were quick to recognize \emph{nonlocality} as a novel feature of quantum mechanics, a notion which was later refined by the seminal work of Bell \cite{bell64}. In tandem to this Kochen and Specker (KS) \cite{KS67}, and independently Bell \cite{bell66},  
demonstrated no-go results concerning the possibility of hidden variable models with pre-existing outcome assignments. 
In recent years these ostensibly different phenomena have been given a unified account, being recognized as facets of a single fundamental feature of quantum theory known simply as contextuality.  

Along side its important role in the foundations of quantum mechanics, contextuality has found many applications in quantum information theory and has been shown to be responsible for many information-processing advantages. See \cite{mermin1993hidden,budroni2021quantum,brunner2014bell,amaral2019resource} and references therein for both of these aspects. There are several different approaches to capture the essence of this important quantum phenomenon: operation-theoretic \cite{Spekkens}, sheaf-theoretic \cite{abramsky2011sheaf} and (hyper)graph-theoretic \cite{cabello2010non,Acin}. 
Among these the sheaf-theoretic approach has yielded a successful framework for a systematic study of contextuality for nonsignaling distributions arising from commuting quantum observables, not necessarily acting on distinct tensor factors; allowing for a comprehensive treatment of both quantum nonlocality as well as  contextuality more broadly. 
More recently, a topological approach was introduced in \cite{Coho} that is well-suited for
studying contextuality as a computational resource in  
measurement-based quantum computation (MBQC)  \cite{raussendorf2016cohomological}. Building off of these efforts, here we introduce a new framework that 
\co{extends}
both the   topological and sheaf-theoretic approaches and is based on the theory of simplicial sets, which are combinatorial models of topological spaces generalizing simplicial complexes. For previous work where topology and sheaf-theory come in contact, see e.g., \cite{beer2018contextuality,
cunha,giovanni2019a,sivert,
montanhano2021contextuality}.

Studies of contextuality often adopt an operational stance wherein physical systems are treated as ``boxes" that take inputs and produce outputs. A measurement scenario then specifies a finite set of inputs and outputs corresponding to choices of measurement and their potential outcomes. Drawing on the quantum analogy, not all measurements can be simultaneously performed, rather only certain subsets of measurements called {\it contexts}.
The statistics of a long run of measurements can then be summarized by a collection of probability distributions that satisfy a compatibility condition known as the {\it nonsignaling condition}. Examples of contextual nonsignaling distributions were first observed in studies of quantum nonlocality through so-called Bell scenarios, most famously by Bell \cite{bell64} and by many authors since then; see e.g., \cite{pr94,werner2001all,barrett2005nonlocal}; also see \cite{brunner2014bell}. For more general types of scenarios where contexts are given by commuting observables, not necessarily local, the sheaf-theoretic formalism \cite{abramsky2011sheaf} provides a rigorous approach to study contextuality. On the other hand, certain types of contextuality proofs, mainly the state-independent ones such as the Mermin square scenario \cite{mermin1993hidden}, can be described using topological methods \cite{Coho} based on chain complexes and group cohomology. In the latter framework measurement contexts are organized into a topological space and a cohomological obstruction defined on this space detects state-independent contextuality. Although this approach cannot be directly applied to the study of nonsignaling distributions, certain types of state-dependent contextuality proofs that are useful as a computational resource in MBQC can be successfully analyzed.

Motivated by the topological approach we introduce a new framework, the {\it simplicial approach to contextuality}, in which we define contextuality for scenarios consisting of {\it spaces of measurements} and {\it spaces of outcomes}. Upgrading the set of measurements and outcomes to a space allows for more general types of distributions. We refer to those new types of distributions as {\it simplicial distributions}. Ordinary nonsignaling distributions defined on   
\emph{sets} of measurements and outcomes  
are referred to as {\it discrete scenarios} and can be studied as a special case, but with extra freedom provided by topology.  
A key conceptual innovation on this front is to model a measurement and its corresponding outcome as a topological  simplex which comes with an intrinsic spatial dimension. 
Such simplicies can then be organized into a combinatorial space which forms a simplicial set.
For example, the Clauser, Horne, Shimony, Holt (CHSH) scenario \cite{chsh69}, a well-known bipartite Bell scenario, can be treated in a couple of different ways. 
The canonical realization of the CHSH scenario as a discrete scenario gives a one-dimensional space (boundary of a square where measurements label the vertices), whereas a certain two-dimensional realization (a square whose edges are labeled by the measurements) proves to be useful for a new topological proof of Fine's theorem \cite{fine1982hidden,fine1982joint} \rev{for characterizing noncontextuality}. Alternatively, another realization  of the same scenario (as a punctured torus where measurements label the edges) gives a characterization of contextuality in terms of an extension condition to the whole torus. 
In this paper we focus on realizations  where measurements and outcomes label one-dimensional simplices (edges) increasing the geometric dimension by one as compared to the canonical realization in the discrete case.

The main benefit of mixing  spaces with probabilities is the topological intuition that allows us to decompose such distributions to distributions on simpler pieces of the underlying space. Contextual properties of those simpler pieces can be analyzed separately and then ``glued back" to determine the contextuality of the original distribution.  
\rev{In this paper our goal is to analyze the well-known scenarios and provide new insights using our simplicial formalism. However, our framework goes beyond the theory of nonsignaling distributions; see Example \ref{ex:small-contextual}. A systematic study of such new scenarios will appear elsewhere.}
Our main contributions in this paper can be summarized as follows:
\begin{itemize}
\item We introduce a new class of objects called simplicial distributions (Definition \ref{def:simp-dist}) that generalize simplicial sets to probabilistic ones and 
\co{extends}
the theory of nonsignaling distributions.

\item Contextuality for simplicial distributions (Definition \ref{def:simp-contextuality}) is introduced and analyzed for well-known scenarios. \rev{This notion of contextuality generalizes the  sheaf-theoretic formulation (Theorem \ref{thm:ComparisonToSheaf}).}

\item A new proof of Fine's theorem \cite{fine1982hidden}\cite{fine1982joint} for the CHSH scenario based on topological arguments is presented (Theorem \ref{thm:Fines-theorem}).

\item We  characterize contextuality in terms of the failure of extending distributions (Proposition \ref{pro:ext-noncontextual}) and apply this to the CHSH scenario (Corollary \ref{cor:extension-torus}).

\item The cohomological framework of \cite{Coho} is extended to simplicial distributions, and in particular to nonsignaling distributions, by constructing cohomology classes that detect strong contextuality (Corollary \ref{cor:Coho-StrongContex}).


\item Quantum measurements are generalized to simplicial quantum measurements, where outcomes are represented by a space instead of a set. We then pass to simplicial distributions via the Born rule and define contextuality for quantum states with respect to a simplicial quantum measurement (Definition \ref{def:quantum-contextuality}).


\item Finally two foundational results from quantum theory, Gleason's theorem \cite{gleason1975measures} and Kochen--Specker theorem \cite{KS67}, are presented in our simplicial framework (Theorem \ref{thm:Gleason} and Corollary \ref{thm:KS}) demonstrating the 
\co{expressive}
 power of the simplicial language.
\end{itemize} 


The rest of the paper is organized as follows: In Section \ref{sec:MotIdea} we provide an informal introduction to our framework, building basic intuition about 
simplicial distributions from a topological viewpoint. %
%
In Section \ref{sec:DistSpace} we provide a formal foundation for our framework, defining abstract simplicial sets before modeling measurements and outcomes as spaces within this setting. %
%
A generalization (Lemma \ref{lem:gluing-lemma}) of Fine's ansatz \cite{fine1982joint,halliwell2019fine} is provided in Section \ref{sec:Glue-Extend} which is later utilized in our topological proof of Fine's theorem. %
In Section \ref{sec:StrongContCoho} we formulate strong contextuality for simplicial distributions and describe its relation to cohomology. 
Section \ref{sec:QuantumMeasSpace} introduces simplicial quantum measurements and applies the cohomological framework of the previous section to state-independent contextuality proofs that were previously analyzed in \cite{Coho}. Some foundational results in quantum theory are described in Section \ref{sec:FoundThms} before we conclude the main text. 
In the appendix, Section \ref{sec:app-SimplicialIdentities} provides background information on simplicial identities and Section \ref{sec:discrete-scenarios}  describes how to embed the sheaf-theoretic formalism of \cite{abramsky2011sheaf} into the simplicial framework.

\section{Motivation and the idea}\label{sec:MotIdea}

In this section we describe our simplicial framework in an informal way.
Starting from the topological approach to the state-independent Mermin square we build step-by-step towards probabilistic scenarios, such as the CHSH scenario.%
%

\subsection{Topological proofs of contextuality}
\label{sec:TopProofCont}

In the topological framework developed in \cite{Coho} a state-independent proof of quantum contextuality is detected by a cohomology class. 
For instance, the Mermin square scenario can be illustrated by a torus together with a certain triangulation as depicted in Fig.~(\ref{fig:Mermin-sq-beta}). 
\begin{figure}[h!]
\centering
\begin{subfigure}{.49\textwidth}
  \centering
  \includegraphics[width=.6\linewidth]{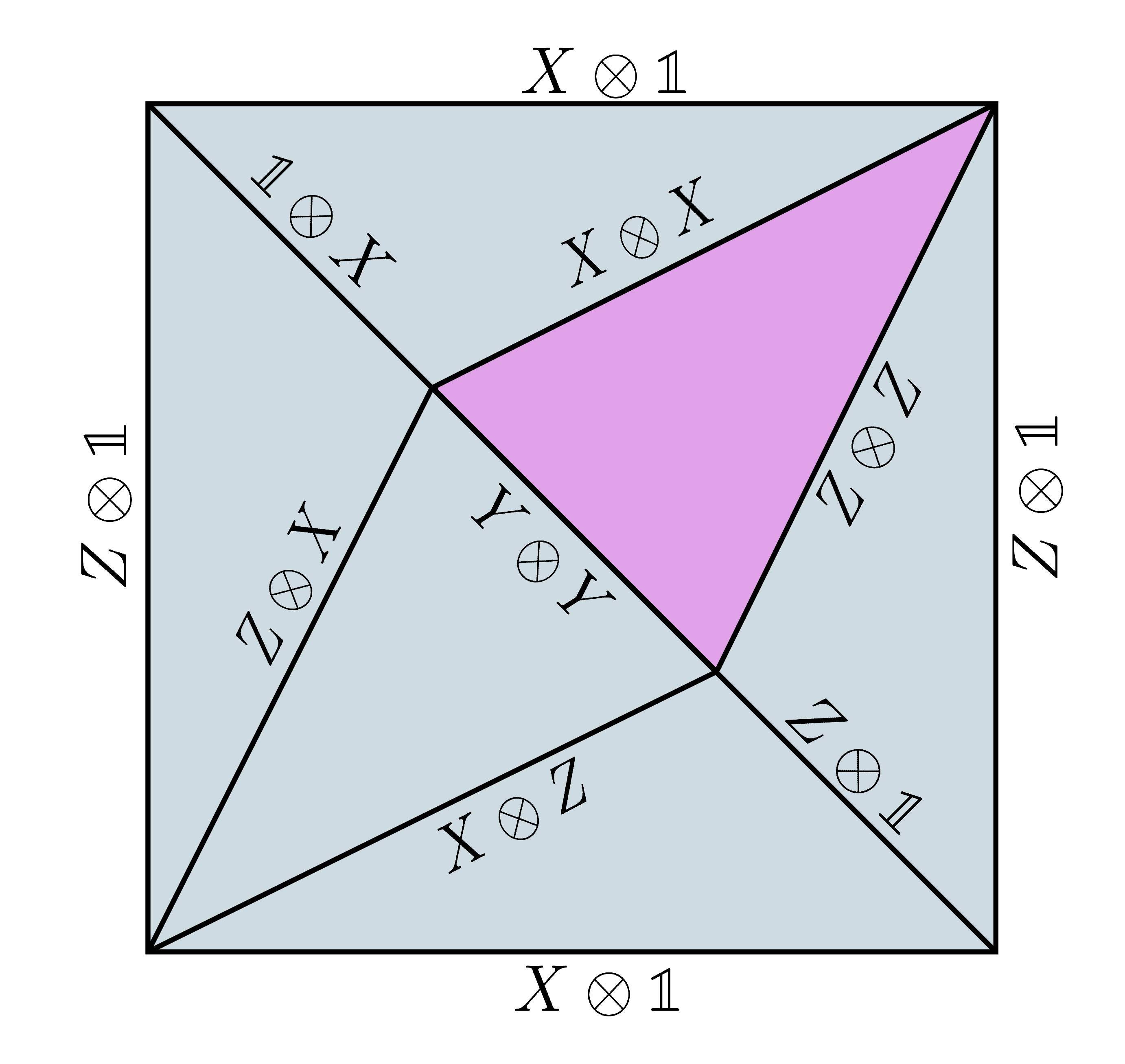}
  \caption{}
  \label{fig:Mermin-sq-beta}
\end{subfigure}%
\begin{subfigure}{.49\textwidth}
  \centering
  \includegraphics[width=.6\linewidth]{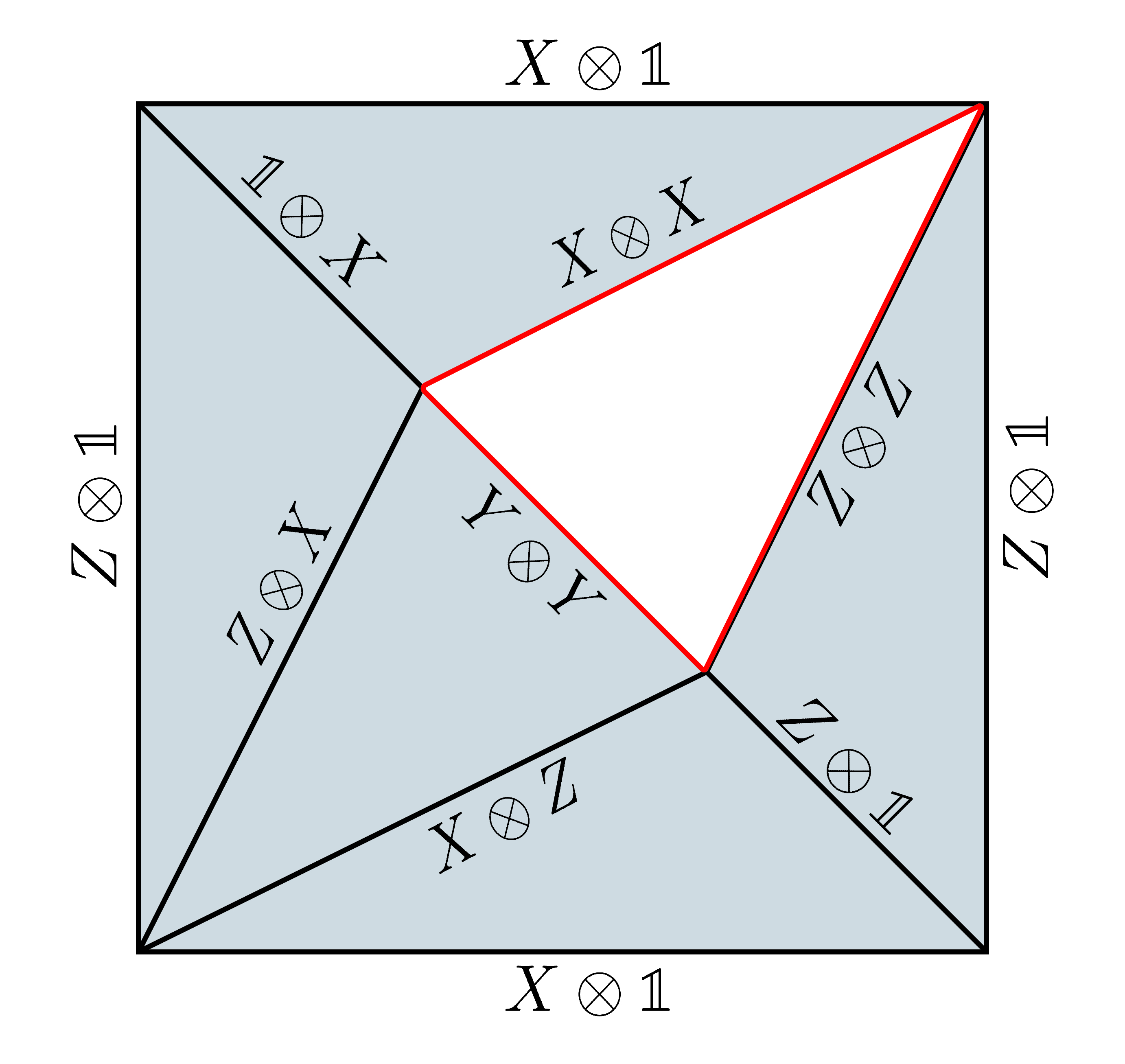}
  \caption{}
  \label{fig:Mermin-sq-relative}
\end{subfigure}
\caption{ (a) Torus representing the state-independent Mermin square. (b) Punctured torus representing the state-dependent Mermin square. The triangle with $\beta=1$ is removed.
}
\label{fig}
\end{figure}
The observables labeling the edges of the triangles, say denoted by $A,B,C$, satisfy   
\begin{equation}\label{eq:beta}
ABC = (-1)^{\beta(A,B)} \one. 
\end{equation}
The function $\beta$ which determines the sign gives rise to  a cohomology class $[\beta]$ that lives on the underlying space  constructed by gluing the triangles.
The fact that there is no way to assign eigenvalues $\pm 1$ to each observable consistent with the algebraic relations given by Eq.~(\ref{eq:beta}) translates into the statement that the cohomology class $[\beta]$ is nonzero.
This observation can be generalized to provide a topological basis for state-independent contextuality proofs using the language of chain complexes from algebraic topology.

In this topological formalism  it is also possible to discuss certain types of state-dependent proofs of contextuality. 
For example, removing the context on which $\beta$ is nonzero results in a scenario that is contextual with respect to the Bell state $|\Psi\rangle=(|00\rangle + |11\rangle)/\sqrt{2}$, the simultaneous eigenstate of the observables $X\otimes X$, $Y\otimes Y$ and $Z\otimes Z$ with the associated eigenvalues $1$, $-1$ and $1$; respectively. In this case one can modify the cocycle with respect to these eigenvalues to introduce a new cocycle $\beta_\Psi$. 
Topologically the underlying space $X$ is a punctured torus, that is a torus minus a point. Now the associated cohomology class $[\beta_\Psi]$, which lives on the space $X$, is a relative cohomology class defined with respect to the boundary $\partial X$. 
Another well-known example is the  Mermin star scenario (without the nonlocal context), which is contextual with respect to the Greenberger-Horne-Zeilinger (GHZ) state $|\Psi\rangle = (|000\rangle+|111\rangle)/\sqrt{2}$. These types of state-dependent contextuality proofs can be formalized using the language of relative chain complexes \cite{Coho,Fraction2018}.

\subsection{Including probability distributions}
\label{sec:IncProDist}

Although the formalism described in \cite{Coho} can deal with certain types of state-dependent contextuality proofs, neither the quantum state nor the measurement statistics enter the picture directly. Our approach 
\co{fills} this gap.

Before going into the details of the formalism we will illustrate the basic idea in the special case of a bipartite scenario. Consider two parties, Alice and Bob, each performing a quantum measurement described by the observables $A$ and $B$ acting on a Hilbert space $\hH$ with eigenvalues $(-1)^a$ where $a\in \ZZ_2=\set{0,1}$. 
\begin{figure}[h!]
\begin{center}
\includegraphics[width=.3\linewidth]{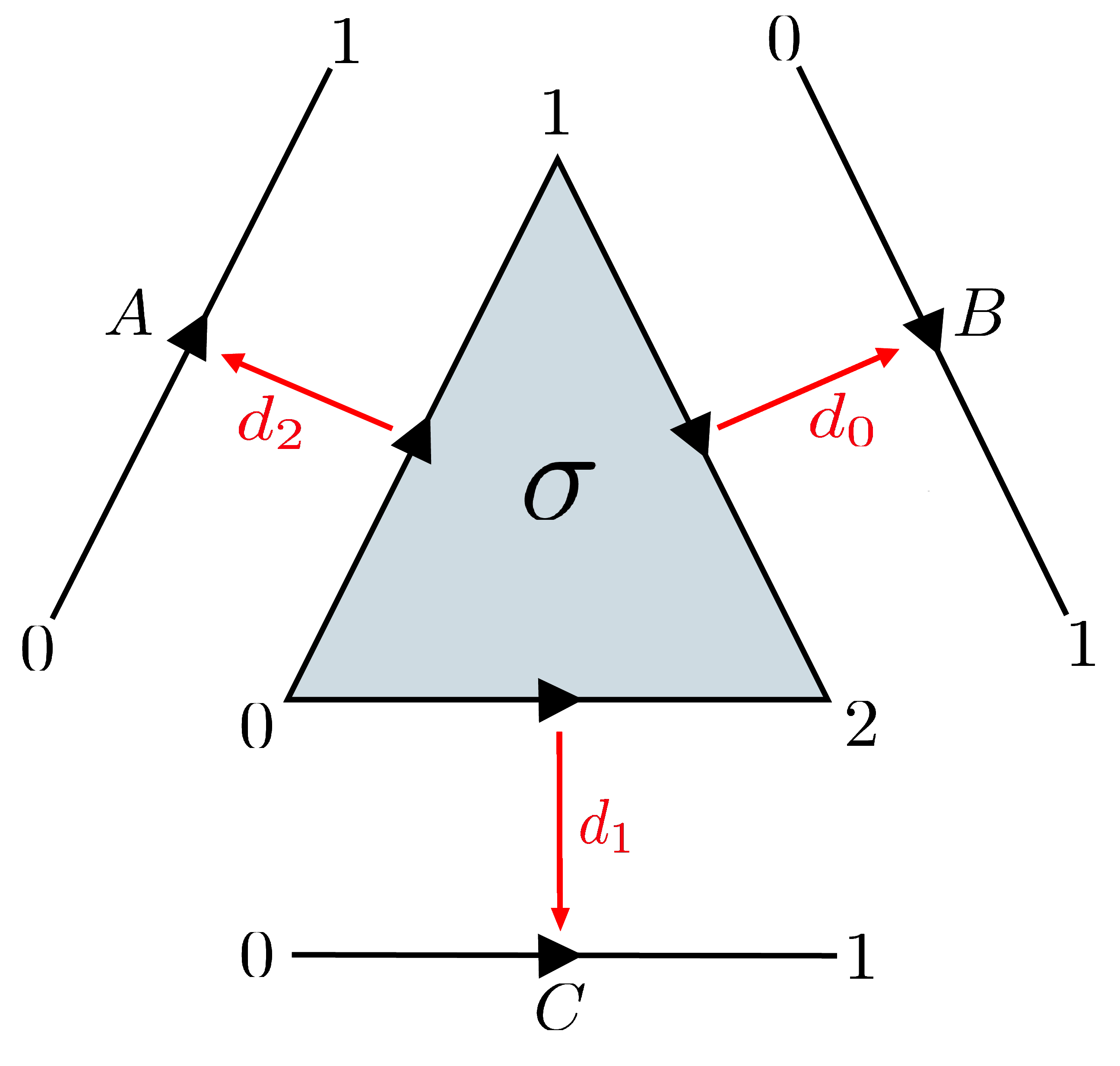} 
\end{center} 
\caption{A triangle whose edges are labeled by commuting observables. 
}
\label{fig:Triangle-faces}
\end{figure} 
At this point we are not assuming that the observables are local. 
We only assume that the measurements can be performed simultaneously, i.e. the observables commute.
Following the topological approach described in Section \ref{sec:TopProofCont} these observables can be organized into a triangle $\sigma$, or $2$-simplex as in Fig.~(\ref{fig:Triangle-faces}). 
The faces of the triangle are therefore denoted by  $d_0\sigma$ labeled by $B$, $d_1\sigma$ labeled by $C=AB$ and $d_2\sigma$ labeled by $A$. These are the edges, i.e. the $1$-simplices, constituting the boundary of the triangle as in Fig.~(\ref{fig:Triangle-faces}).
To clarify our topological conventions two remarks are in order:
(1) Notice that this approach models an individual measurement as a one-dimensional object (i.e., $1$-simplex) while a context consisting of more than one measurement spans at least a two-dimensional space (i.e., $2$-simplex). This increases the working dimension of our spaces, as compared to the sheaf \cite{abramsky2011sheaf} and (hyper)graph theoretic \cite{cabello2010non}\cite{Acin} approaches, thus facilitating a richer connection to topology.
(2) In general, for a 
\co{pair}
 of commuting observables $A$ and $B$\co{, and a complex valued function $F$ of two variables,} quantum mechanics asserts that one can always find a third commuting observable $C = F(A,B)$. Following \cite{Coho}, which itself drew from the state-independent contextuality proofs of Mermin \cite{mermin1993hidden}, we focus our attention on those compatible observables that are formed by product relations $AB$. In the context of MBQC such observables are more aptly called \emph{inferables}; see e.g., \cite{raussendorf2016cohomological}.

A pair of observables defines a projective measurement. 
Let $\Pi_{AB}^{ab}$ denote the projector onto the simultaneous eigenspace of $A$ and $B$ with eigenvalues $(-1)^a$ and $(-1)^b$; respectively. Marginalizing over each outcome gives the projectors for the observables $A$, $B$ and $C$:
\begin{equation}\label{eq:marginal-proj}
\Pi_{A}^a = \sum_{b} \Pi_{AB}^{ab}, \;\;\;
\Pi_{B}^b =  \sum_{a} \Pi_{AB}^{ab}, \;\;\;
 \Pi_{C}^c =  \sum_{a+b=c} \Pi_{AB}^{ab}
\end{equation} 
where $\Pi_A^a$ projects onto the $(-1)^a$ eigenspace, and similarly for the other observables. 
We can organize this information in a topological way by interpreting outcomes as geometric objects. A pair of outcomes $(a,b)$ represents a triangle and to each such triangle we associate the projector $\Pi_{AB}^{ab}$. 
Assembling the projectors $\Pi_{AB}^{ab}$ into a projective measurement $\Pi_{AB}:\ZZ_2^2 \to \Proj(\hH)$ and $\Pi_A^a$ (similarly $\Pi_B^b$ and $\Pi_C^c$) into a projective measurement $\Pi_{A}:\ZZ_2 \to \Proj(\hH)$ Eq.~(\ref{eq:marginal-proj}) can be written as face relation  in  Fig.~(\ref{fig:Triangle-proj-faces}):
$$
d_i \Pi_{AB} = \left\lbrace
\begin{array}{ll}
\Pi_B & i=0 \\
\Pi_C & i=1 \\
\Pi_A & i=2.
\end{array}
\right.
$$
\begin{figure}[h!]
\centering
\begin{subfigure}{.49\textwidth}
  \centering
  \includegraphics[width=.7\linewidth]{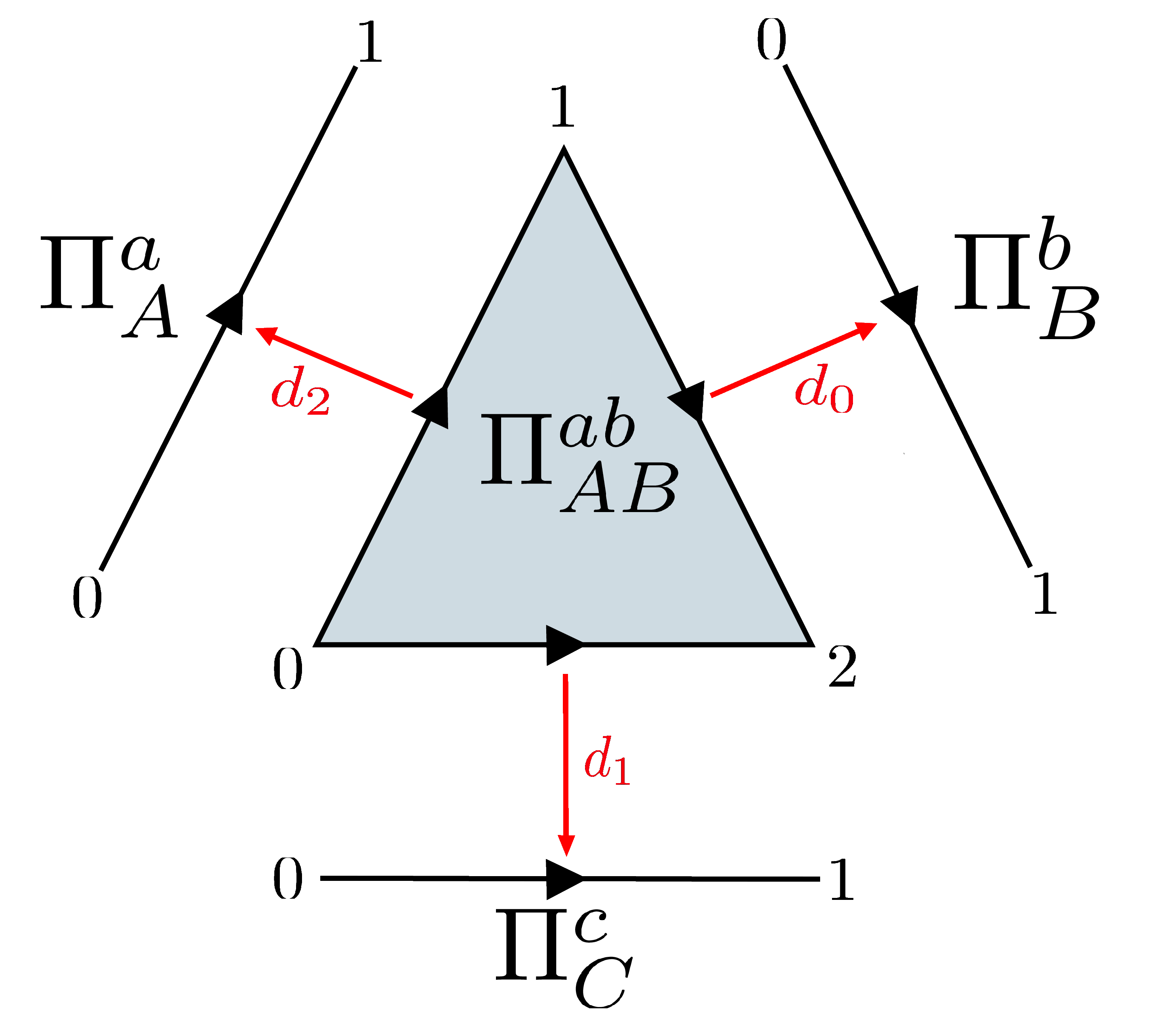}
  \caption{}
  \label{fig:Triangle-proj-faces}
\end{subfigure}%
\begin{subfigure}{.49\textwidth}
  \centering
  \includegraphics[width=.7\linewidth]{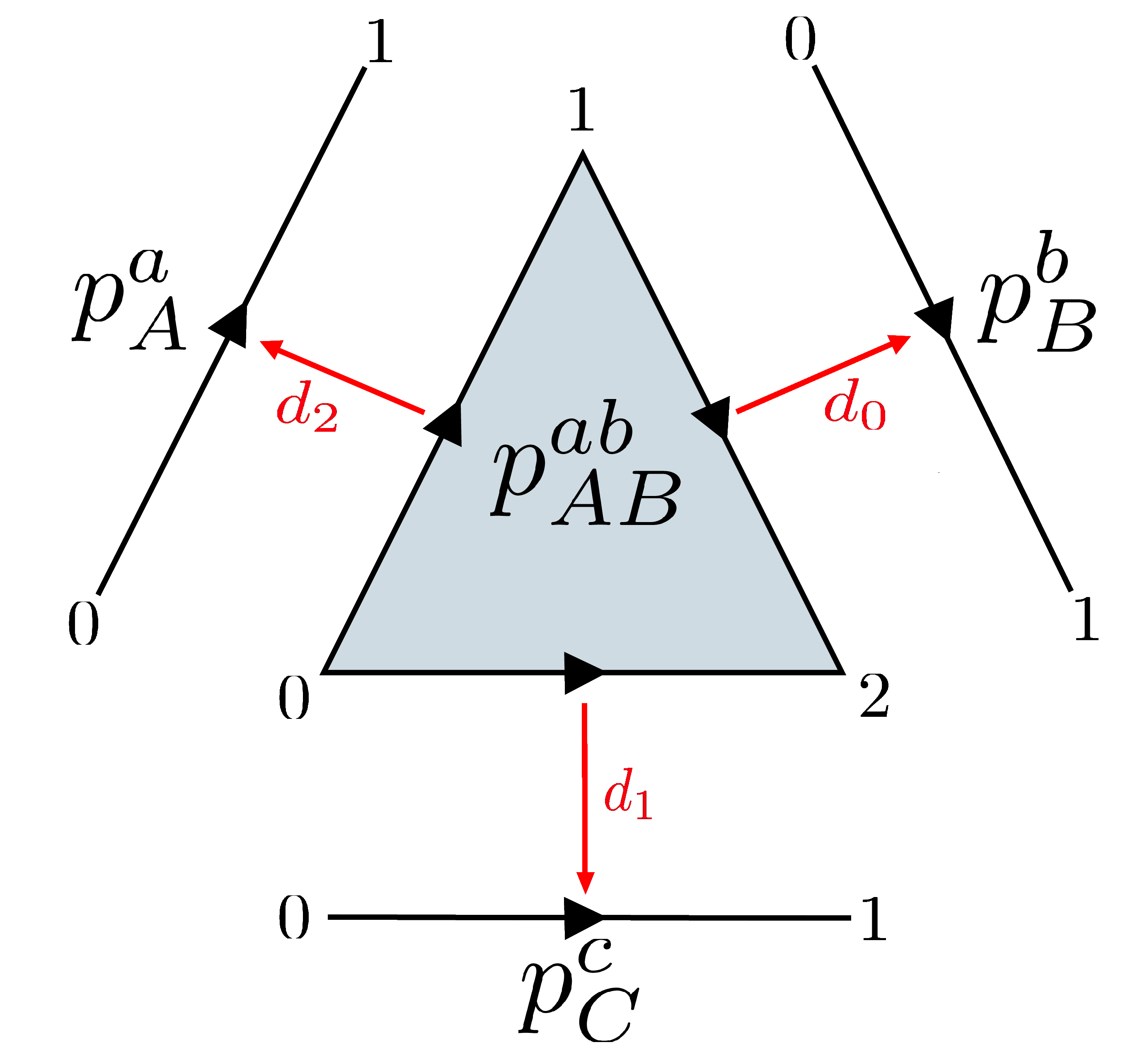}
  \caption{}
  \label{fig:Triangle-prob-faces}
\end{subfigure}
\caption{(a) The pair $(A,B)$ of observables is represented by a projective measurement on the collection of triangles labeled by the pairs $(a,b)\in \ZZ_2^2$ of outcomes. 
The outcome associated to the $d_1$ face is given by $c=a+b$. (b) The  probability distribution on the collection of triangles is given by the Born rule.
}
\label{fig}
\end{figure}

To pass to probabilities we pick a quantum state $\rho$ and apply the Born rule. 
 The probability  
 $p_{AB}^{ab} = \Tr(\rho \Pi_{AB}^{ab})$  is assigned to the triangle labeled by the pair $(a,b)$ of outcomes. 
 More precisely,  we regard $p_{AB}^{ab}$ as a probability distribution on the collection of triangles labeled by the outcomes $(a,b)$. 
We have the following marginals
\begin{equation}\label{eq:marginal-prop}
p_{A}^a = \sum_{b} p_{AB}^{ab}, \;\;\;
p_{B}^b =  \sum_{a} p_{AB}^{ab}, \;\;\;
p_{C}^c =  \sum_{a+b=c} p_{AB}^{ab}
\end{equation}   
where $p_A^a = \Tr(\rho \Pi_A^a)$, and similarly for the other observables.  
Face relations in  Fig.~(\ref{fig:Triangle-prob-faces})
encode these marginals 
$$
d_i p_{AB} = \left\lbrace
\begin{array}{ll}
p_B & i=0 \\
p_C & i=1 \\
p_A & i=2
\end{array}
\right.
$$
where $p_{AB}:\ZZ_2^2 \to \nnegR$ is the probability distribution associated to $p_{AB}^{ab}$ and $p_{A}:\ZZ_2 \to \nnegR$ (similarly $p_B$ and $p_C$) is the probability distribution associated to $p_A^a$.

\subsection{The CHSH scenario}\label{sec:chsh-scenario}
 
A bipartite Bell scenario consists of two parties, Alice and Bob, where Alice performs a measurement $x$ with outcome $a$ and Bob performs a measurement $y$ with outcome $b$. The joint statistics observed by Alice and Bob can be summarized by a conditional probability distribution $p_{xy}^{ab}$.  When there are multiple measurement choices for Alice and Bob this results in a collection of such probabilities that are subject to compatibility requirements called nonsignaling conditions; see e.g., \cite{masanes2006ns}.
For instance, in a CHSH scenario where Alice can perform the measurements $x_0$ and $x_1$ and Bob can perform $y_0$ and $y_1$, obtaining probabilities $p_{x_{i}y_{j}}^{ab}$ for each pairing $\{x_{i},y_{j}\}$, the corresponding nonsignaling conditions are given by
\begin{equation}\label{eq:nonsignaling-cond}
p_{x_i}^a= \sum_{b} p_{x_iy_0}^{ab} = \sum_{b} p_{x_iy_1}^{ab}  \; \text{ and }\; p_{y_j}^b=\sum_{a} p_{x_0y_j}^{ab} = \sum_{a} p_{x_1y_j}^{ab}~.
\end{equation}
These expressions ensure that Alice's choice of measurement does not affect the outcome of Bob's measurement, and \emph{vice versa}. If Alice and Bob are far apart (i.e., they are causally disconnected), then the nonsignaling conditions prohibit superluminal communication between them, in accordance with the theory of special relativity.


%
%

In our approach the nonsignaling conditions are depicted quite naturally. Using the topological representation of Section \ref{sec:IncProDist} each $p_{x_iy_j}^{ab}$ can be placed on a triangle and the nonsignaling conditions can be encoded as the matching of the face relations at the intersections of the triangles. 
To see this we start with the simpler case of
a single measurement per party.
We would like to represent the probabilities $p_{xy}^{ab}$ in a topological fashion as in the previous section. First we associate a triple of compatible measurements with the faces of the triangle. More specifically, to the faces $d_{2}\sigma$ and $d_{0}\sigma$ we assign the measurements $x$ and $y$, respectively. To complete the picture, however, we must also associate a measurement with the $d_{1}\sigma$ face compatible with $x$ and $y$.
%
\begin{figure}[h!]
\centering
\begin{subfigure}{.49\textwidth}
  \centering
  \includegraphics[width=.7\linewidth]{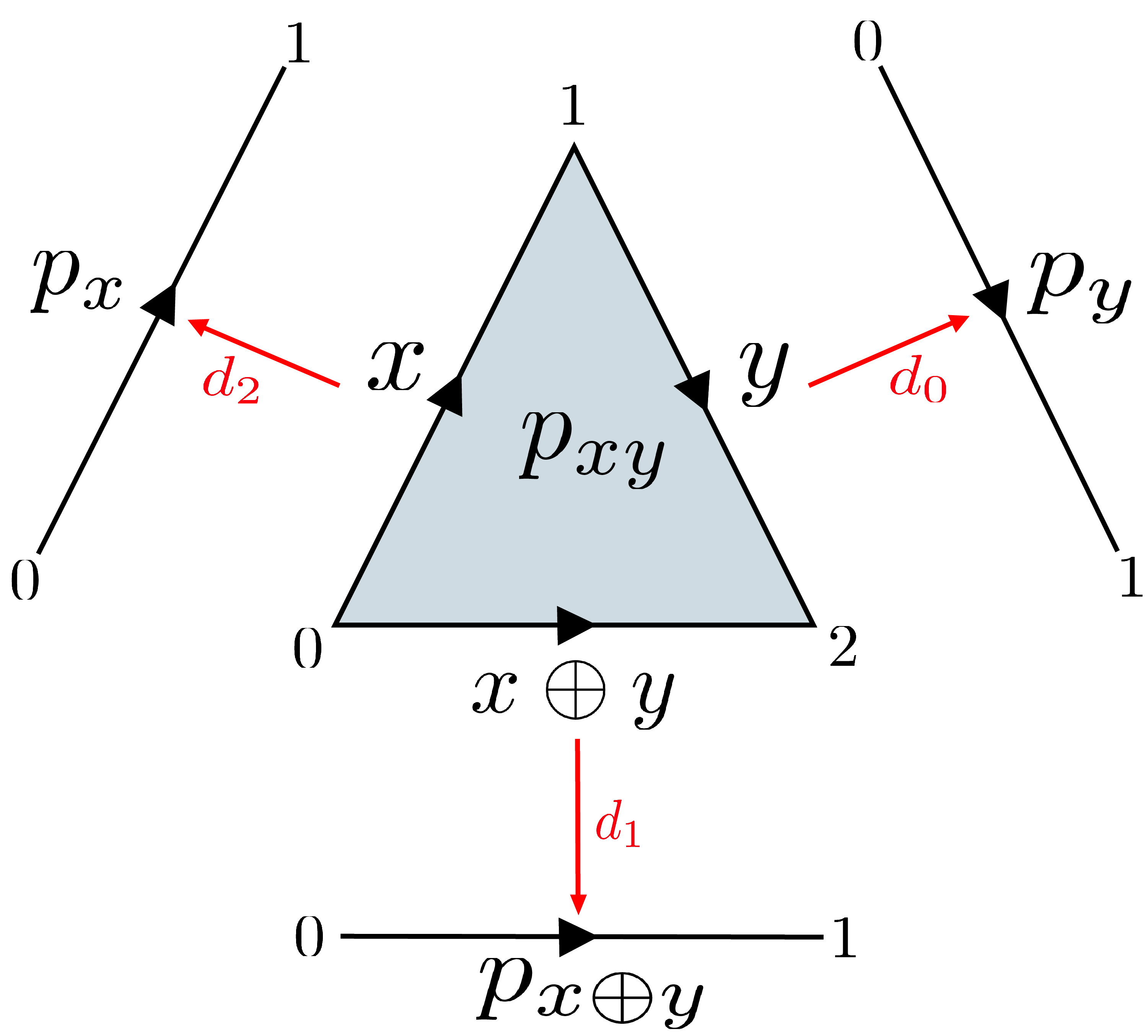}
  \caption{}
  \label{fig:Triangle-meas-faces-a}
\end{subfigure}%
\begin{subfigure}{.49\textwidth}
  \centering
  \includegraphics[width=.85\linewidth]{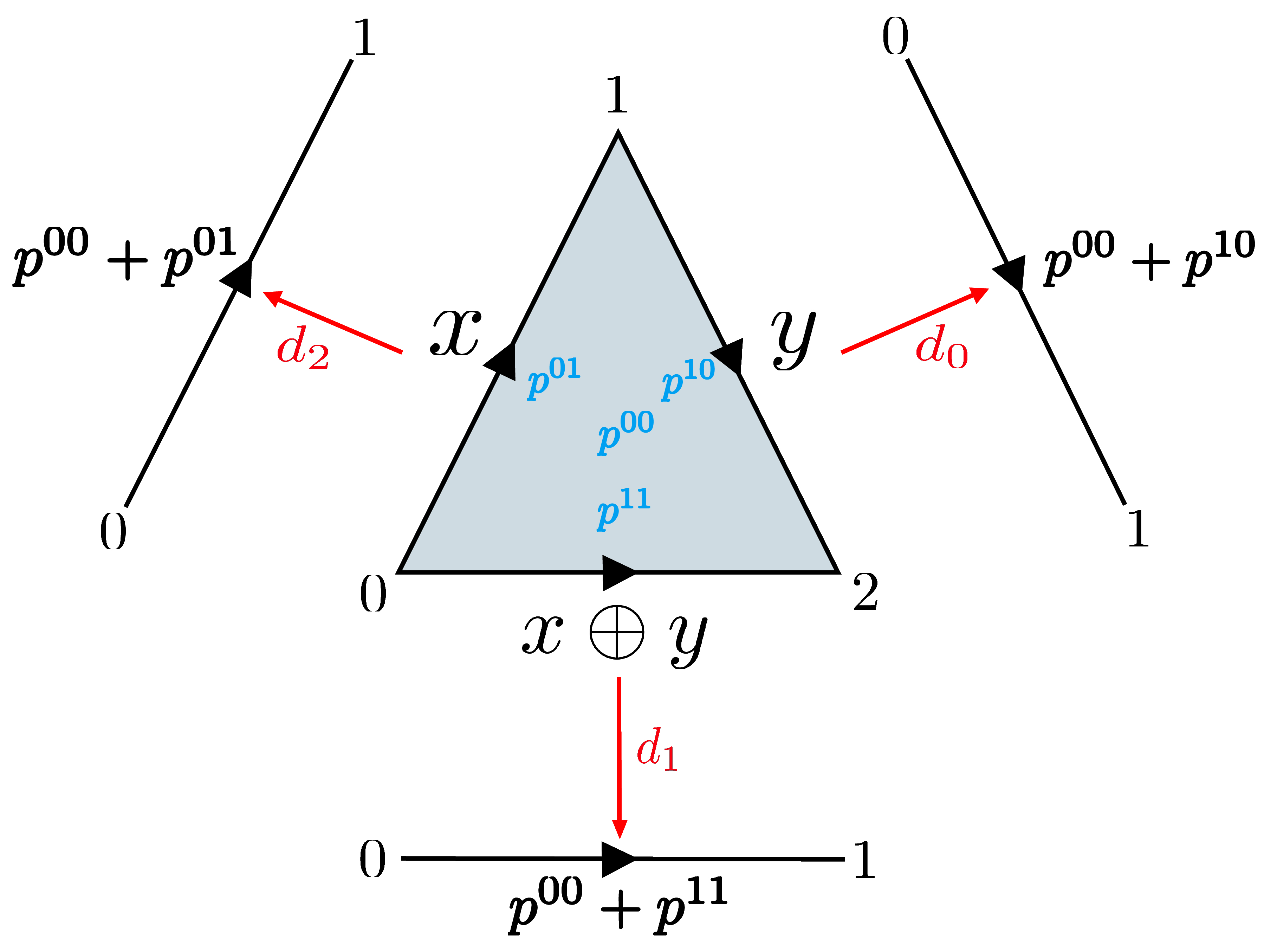}
  \caption{}
  \label{fig:Triangle-meas-faces-b}
\end{subfigure}
\caption{(a) Face relations of a distribution on a triangle. (b) A convenient way to represent the probabilities $p^{ab}$ on the triangle. \rev{Marginals on each edge can be obtained by summing $p^{00}$ with the probability close to that edge: $p_x^0=p^{00}+p^{01}$, $p_y^0=p^{00}+p^{10}$ and $p_{x\oplus y}^0=p^{00}+p^{11}$.}
}
\label{fig}
\end{figure}
A measurement 
is said to be \emph{compatible} with both $x$ and $y$ if its outcome
can be obtained by a classical post-processing of the outcomes of $x$ and $y$ individually.
This definition of compatibility draws on the operational definition of joint measurability introduced in \cite{krishna2017deriving,liang2011specker}. Similar notions of compatibility have also been expressed in \cite{raussendorf2016cohomological}.  
For our purposes we take $x\oplus y$ to be the XOR of $x$ and $y$ so that $x\oplus y$ produces the outcome $c=a+b \mod 2$ when the outcome of $x$ happens to be $a$ and the outcome of $y$ happens to be $b$.  Such a choice is well-motivated on physical grounds as the XOR produces precisely that statistic which is relevant for many studies of contextuality; e.g., appearing through correlation functions \cite{brunner2014bell,cleve2004nonlocal}, linear side-processing in MBQC \cite{raussendorf2016cohomological}, as well as in so-called all-versus-nothing arguments \cite{abramsky2017avn}. Later we will see that the XOR outcomes can be modeled by 
the nerve space (Definition \ref{def:NZd}). 

Probability distributions associated with the measurement $(x,y) $ are represented as in Fig.~(\ref{fig:Triangle-meas-faces-a}).   In this representation $p_{xy}$ is the probability distribution on the collection of triangles labeled by pairs $(a,b)\in  \ZZ_2^2$. The distributions on the edges $p_x$, $p_y$ and $p_{x\oplus y}$ are obtained by marginalization in the usual way:
\begin{equation}\label{eq:prob-mar-xy}
p_{x}^a = \sum_{b} p_{xy}^{ab}, \;\;\;
p_{y}^b =  \sum_{a} p_{xy}^{ab}, \;\;\;
p_{x\oplus y}^c =  \sum_{a+b=c} p_{xy}^{ab},
\end{equation}
or equivalently using the face relations
$$
d_i p_{xy} = \left\lbrace
\begin{array}{ll}
p_y & i=0 \\
p_{x\oplus y} & i=1 \\
p_x & i=2.
\end{array}
\right.
$$

%

For the CHSH scenario we will take a triangle for each pair $(x_i,y_j)$ of measurements and glue them in a way that the topology encodes the nonsignaling conditions given in Eq.~(\ref{eq:nonsignaling-cond}). There are different ways of doing this which produces different spaces at the end. We will consider the one given in Fig.~(\ref{fig:Bell-punctured-torus}). 
The edges labeled by $x_0$ are identified, and similarly the ones labeled by $x_1$. The resulting space is a punctured torus $T^\circ$.
Note that we are simply removing the two triangles (diamond shape) in the middle of the torus representing the   Mermin square scenario in Fig.~(\ref{fig:Mermin-sq-beta}).  (Topologically this space is equivalent to the picture in Fig.~(\ref{fig:Mermin-sq-relative})). 
The collection of contexts labeling the triangles are given by 
\begin{equation}\label{eq:Bell-contexts}
C_{ij} = \set{x_i,y_j} \;\;\text{ where }i,j\in \ZZ_2. 
\end{equation}
The nonsignaling (compatibility) conditions for $p$ can be rewritten in a topological way
\begin{equation}\label{eq:nonsignaling-top}
\begin{aligned}
d_0 p_{y_0x_0} &= p_{x_0} = d_2 p_{x_0y_1} \\
d_0 p_{y_0x_1} &= p_{x_1} =d_2 p_{x_1y_1}   \\
d_2 p_{y_0x_0} &= p_{y_0} =d_2 p_{y_0x_1}    \\
d_0 p_{x_0y_1} &= p_{y_1} =d_0 p_{x_1y_1}. \\
\end{aligned}
\end{equation}
The orientation on the boundary of each triangle determines the label for the contexts (compare with Fig.~(\ref{fig:Triangle-meas-faces-a})), e.g. $p_{y_0x_{0}}^{ab}$ is the probability for the outcome assignment $(y_0,x_{0})\mapsto (b,a)$, and so on.
\begin{figure}[h!]
\centering
\begin{subfigure}{.49\textwidth}
  \centering
  \includegraphics[width=.6\linewidth]{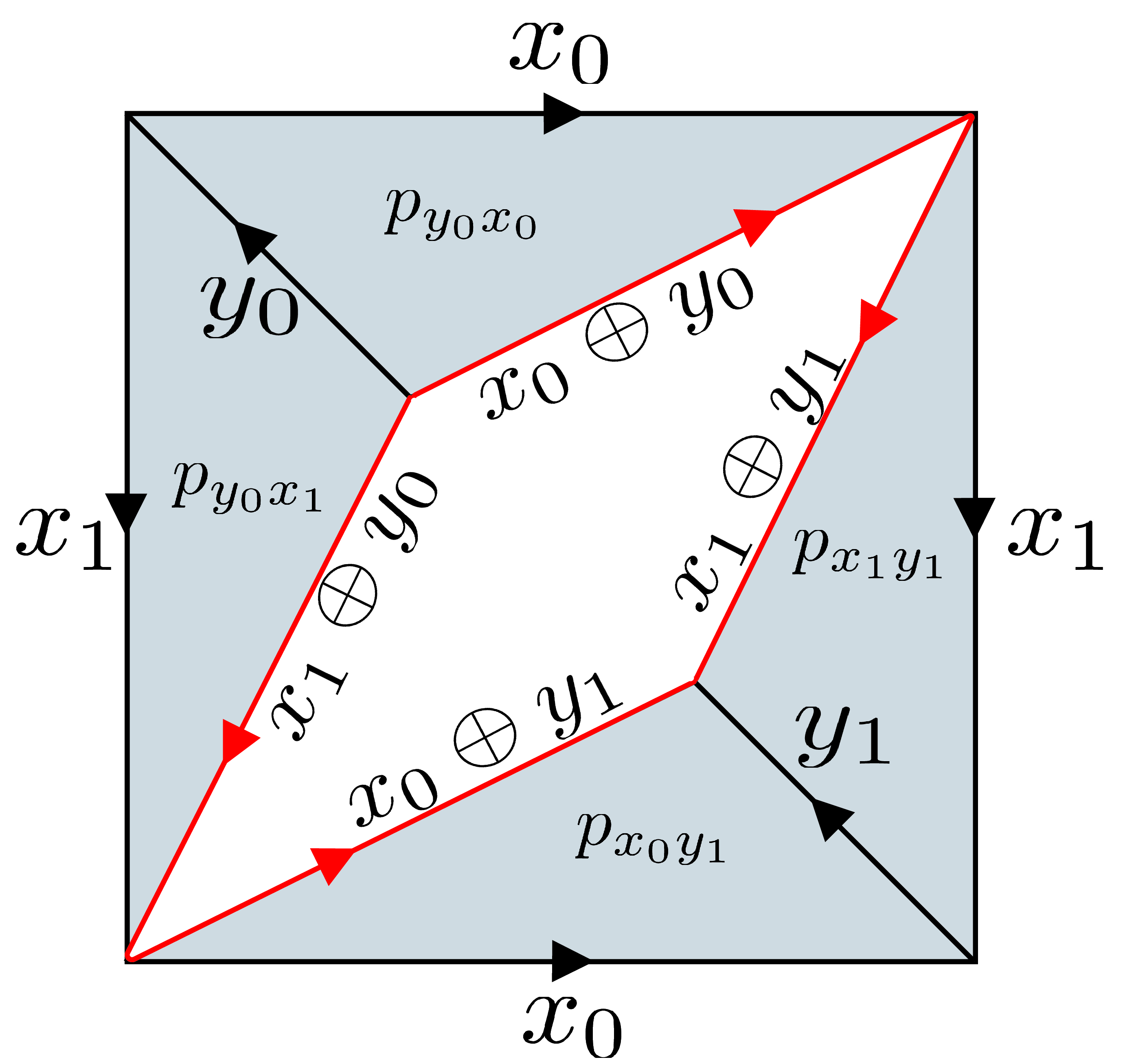}
  \caption{}
  \label{fig:Bell-punctured-torus}
\end{subfigure}%
\begin{subfigure}{.49\textwidth}
  \centering
  \includegraphics[width=.6\linewidth]{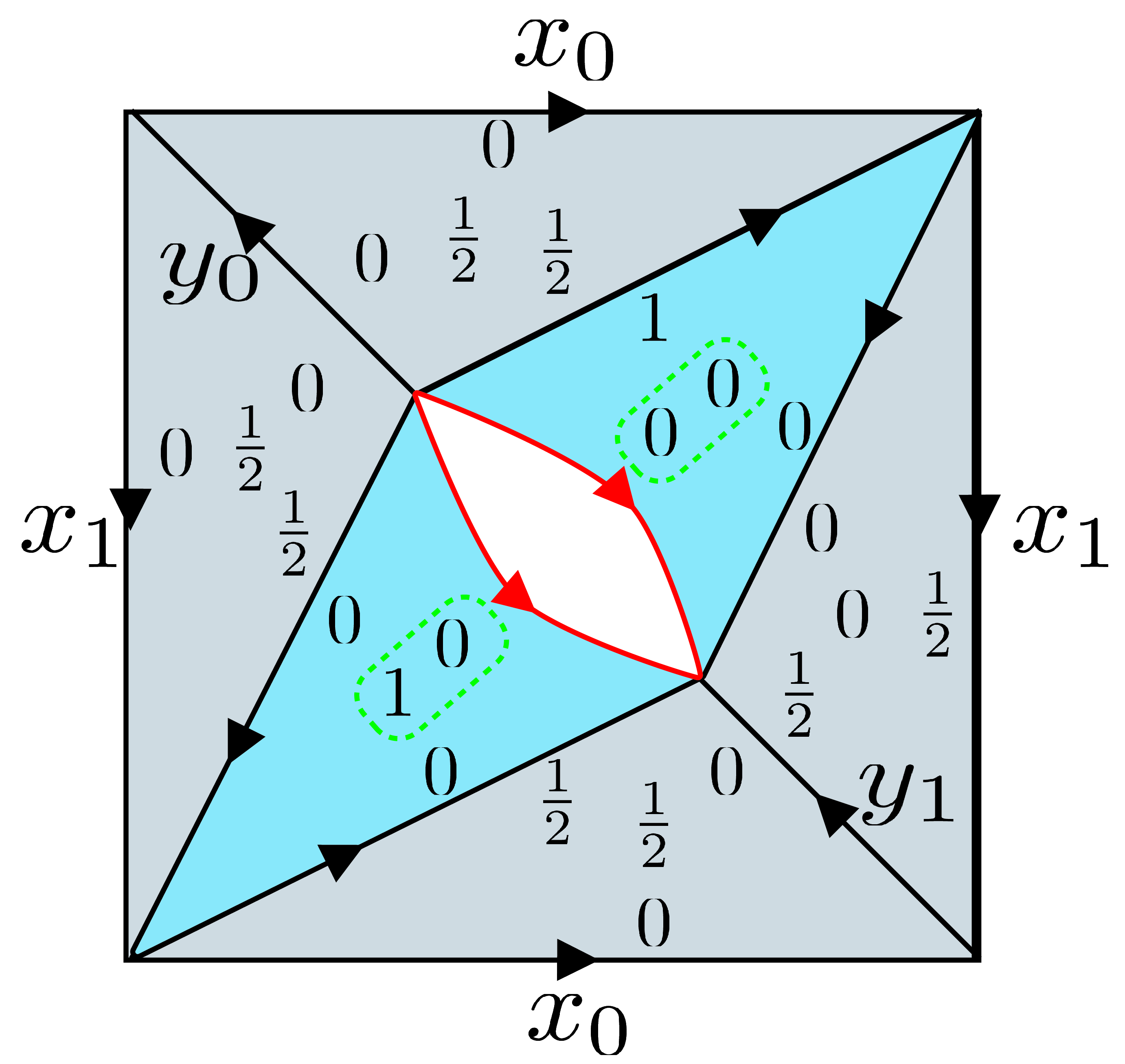}
  \caption{}
  \label{fig:Bell-punctured-torus-ext}
\end{subfigure}
\caption{ (a) The CHSH scenario organized into a surface topologically equivalent to a punctured torus. The edges labeled by $x_0$ (and $x_1$) are identified. (b) The PR box on the punctured torus cannot be extended to a distribution on the torus\rev{; see Fig.~(\ref{fig:Triangle-meas-faces-b}) for notation.}
}
\label{fig}
\end{figure}

\subsection{Contextuality in the new framework}\label{sec:ContNewFramework}

The full definition of contextuality will require the precise notion of a space (given in Section \ref{sec:SpaceOfMeas}) on which our formalism is based. For now we will give an idea how contextuality is formulated using the topological pictures introduced above.
Let us start with deterministic distributions. Those are given by delta distributions associated with an outcome assignment on each measurement. For example, in the CHSH scenario considered in  Fig.~(\ref{fig:Bell-punctured-torus}) deterministic distributions are of the form $\delta^s$ where $s$ is a function from $\set{x_0,x_1,y_0,y_1}$ to $\ZZ_2$. 
By definition \cite{budroni2021quantum, pitowsky1989quantum}  a nonsignaling distribution $p$ is noncontextual if it is a probabilistic mixture of deterministic distributions:
\begin{equation}\label{eq:noncontextual-prob-mix}
p = \sum_{s} \lambda(s)\, \delta^s
\end{equation}
where $\lambda(s)\geq 0$ and $\sum_s \lambda(s) =1$. 

\Ex{[Triangle scenario]\label{ex:SingleTriangle}
{\rm
Any distribution on a scenario consisting of two measurements $\set{x,y}$ as in Fig.~(\ref{fig:Triangle-meas-faces-a}) is noncontextual: 
$$  
\includegraphics[width=.35\linewidth]{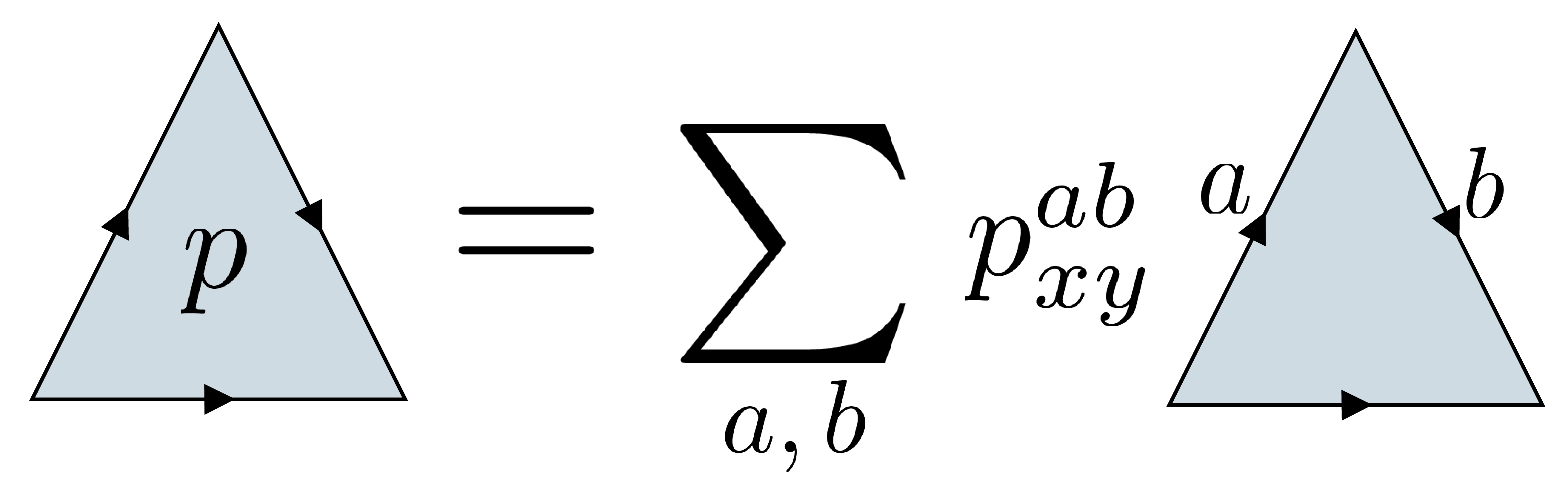}  
$$ 
}}

\rev{
\Ex{[New contextual scenarios]\label{ex:small-contextual}
{\rm
The triangle scenario can be modified to obtain a contextual scenario  by gluing the $d_0$ and $d_1$ faces. This forces the probabilities $p=(p^{00},p^{01},p^{10},p^{11})$ to satisfy the additional relation $p^{10}=p^{11}$ which is enforced by the topological identification. 
For this scenario $p$ is contextual if and only if $p^{00}+p^{01}<1$. For example, $(0,0,1/2,1/2)$ is contextual:
$$ 
\includegraphics[width=.5\linewidth]{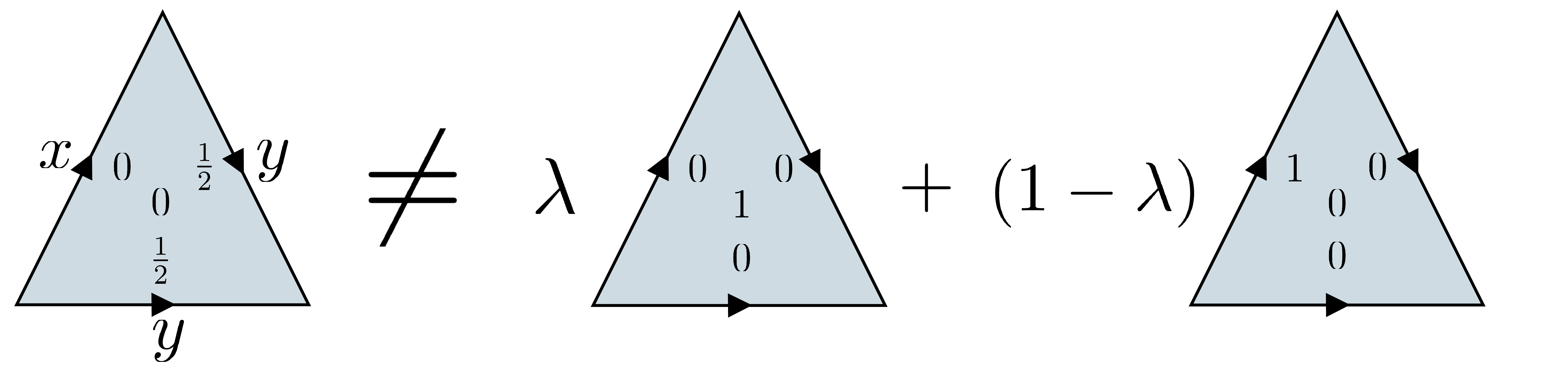}  
$$ 
\co{A noncontextual distribution is one with $p^{00}=\lambda$, $p^{01}=1-\lambda$, and $p^{11}=p^{10}=0$ for some $0\leq \lambda\leq 1$.}
In general  topological identifications can be used to construct  new scenarios that impose additional constraints satisfied by the probabilities.
These kinds of   scenarios that arise in our framework cannot be realized directly as nonsignaling distributions.  In this paper we will not attempt for a systematic study of these new kinds of scenarios, such a  study will appear elsewhere.
}}

}

\Ex{[Diamond scenario]{\rm \label{ex:TwoTriangles}
A slightly harder example is the scenario consisting of four measurements $\set{x_0,y_0,x_1,y_1}$ with the nonsignaling condition $d_1p=d_1q$, or equivalently
$$
\sum_{a+b=e} p^{ab}_{x_0y_0} = \sum_{c+d=e} q^{cd}_{x_1y_1},\;\;\;\; e\in \ZZ_2.
$$
In this case there exists probabilities $\lambda(abc)\in [0,1]$ satisfying the equation:
$$ 
\includegraphics[width=.4\linewidth]{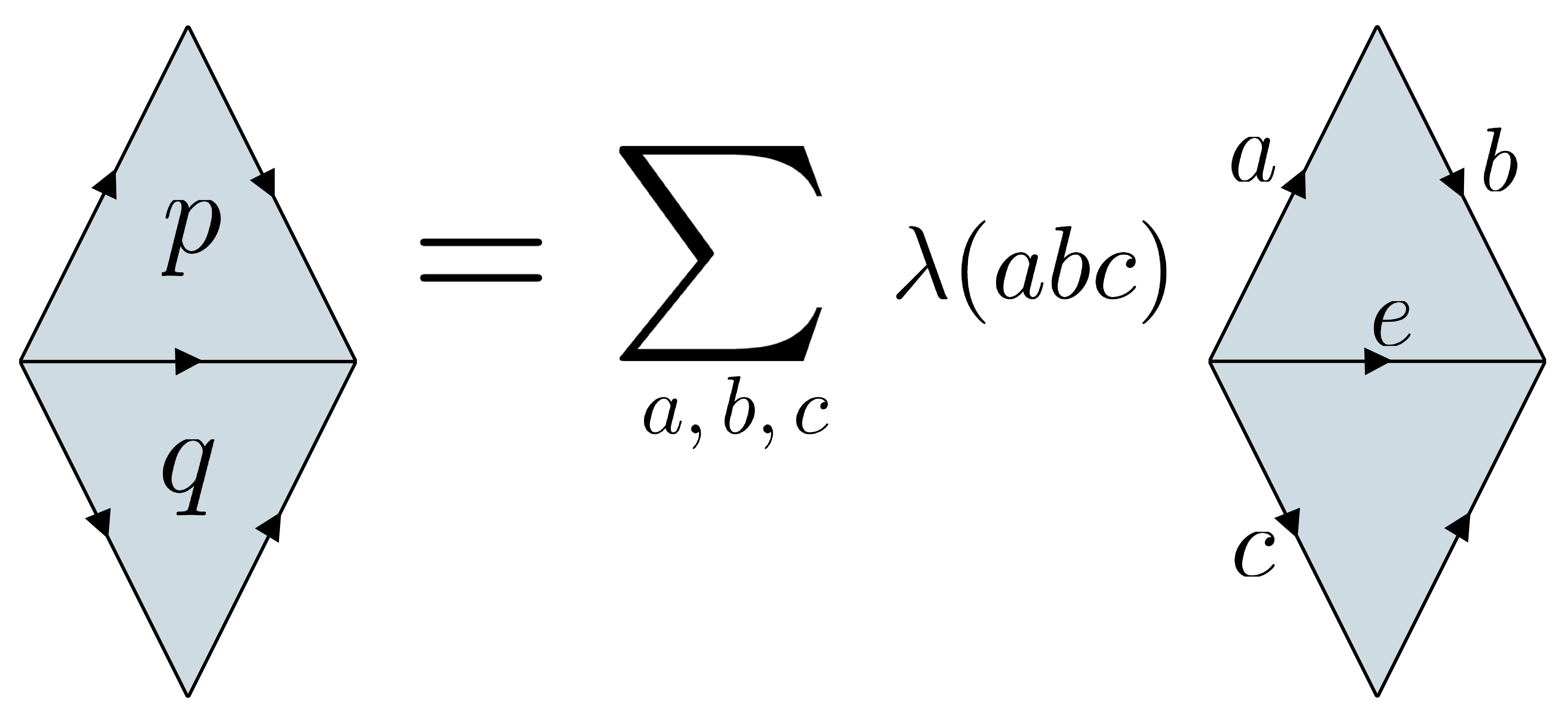}  
$$ 

\noindent One solution can be given as follows:
\begin{equation}\label{eq:lambda}
\lambda(abc) = \left\lbrace
\arraycolsep=3.5pt\def\arraystretch{2}
\begin{array}{ll}
 p^{aa} q^{cc}/(p^{00}+p^{11})   &\quad e=0 \\
 p^{a\bar a} q^{c\bar c}/(p^{01}+p^{10})  &\quad e=1
\end{array}
\right.
\end{equation}
where $\bar a$ stands for $a+1$, and similarly for $\bar c$. 
}
} 

\Ex{[Mermin square scenario]\label{ex:MerminStateDep}
{\rm
Next let us consider a contextual example, such as the state-dependent version of the Mermin square scenario given in Fig.~(\ref{fig:Mermin-sq-relative}). To see that the nonsignaling distribution obtained from the quantum mechanical system consisting of the given observables and the Bell state is contextual it suffices to look at the restriction $p|_{\partial X}$ to the boundary. This is a deterministic distribution specified by $(\delta^e,\delta^f,\delta^g)$ with $e+f+g=1$ as dictated by the eigenvalues on the Bell state. However, a noncontextual   distribution satisfies the property that $e+f+g=0\mod 2$ 
as a consequence of the following equation:
$$ 
\includegraphics[width=.5\linewidth]{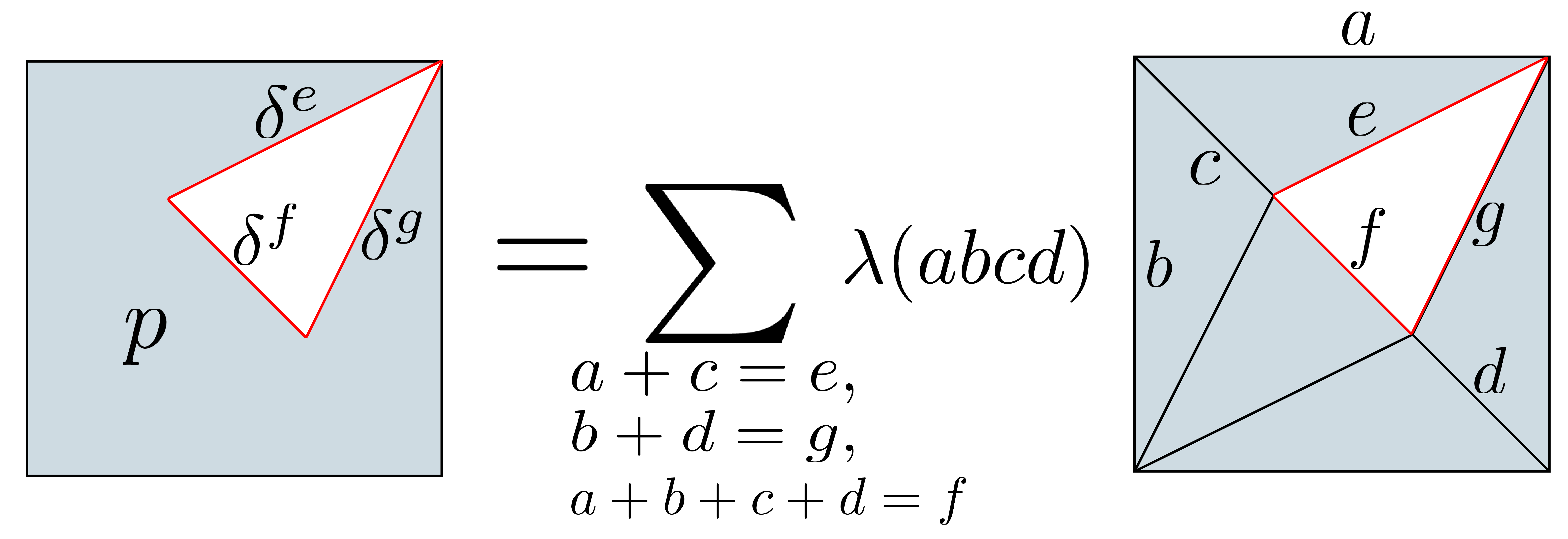}  
$$ 
\noindent 
On the right-hand side the outcomes assigned to the torus satisfy the relations $e=a+c$, $g=b+d$ and $f=(b+c)+(a+d)$.
}
}

\Ex{[CHSH scenario]{\rm \label{ex:Bell222}
Contextuality analysis of the CHSH scenario in  Fig.~(\ref{fig:Bell-punctured-torus}) is more complicated but also well-known. There is a characterization of noncontextual nonsignaling distributions for this scenario due to Fine \cite{fine1982hidden,fine1982joint}.  
The distribution $p|_{\partial X}$, which is the restriction to the boundary of the punctured torus in   Fig.~(\ref{fig:Bell-punctured-torus}), consists of the distributions $p_{x_i\oplus y_j}$ for the measurements $x_i\oplus y_j$ where $i,j\in \ZZ_2$. 
According to Fine's theorem \cite{fine1982hidden,fine1982joint} a nonsignaling distribution on the CHSH scenario  is noncontextual if and only if the following CHSH inequalities \cite{chsh69} hold
\begin{equation}\label{eq:CHSH234}
\begin{aligned}
0\leq p_{x_0\oplus y_0}^0 + p_{x_0\oplus y_1}^0 + p_{x_1\oplus y_0}^0 - p_{x_1\oplus y_1}^0 \leq 2&\\
0\leq p_{x_0\oplus y_0}^0 + p_{x_0\oplus y_1}^0 - p_{x_1\oplus y_0}^0 + p_{x_1\oplus y_1}^0 \leq 2&\\
0\leq p_{x_0\oplus y_0}^0 - p_{x_0\oplus y_1}^0 + p_{x_1\oplus y_0}^0 + p_{x_1\oplus y_1}^0 \leq 2&\\
0\leq -p_{x_0\oplus y_0}^0 + p_{x_0\oplus y_1}^0 + p_{x_1\oplus y_0}^0 + p_{x_1\oplus y_1}^0 \leq 2&.
\end{aligned}
\end{equation}
(These inequalities are equivalent to the ones in terms of the correlations under the  standard relationship \cite[Section 4a]{brunner2014bell} between XOR probabilities and correlation functions.)
We will carry on a careful analysis of this scenario in Section 
\co{\ref{sec:TopProofFine}}
as an application of our framework. Therein 
we use another space (left figure in Fig.~(\ref{fig:S-H})) that realizes this scenario to give a topological proof of Fine's theorem.
Then using the current realization, given by the punctured torus, we characterize noncontextuality in terms of an extension condition. To illustrate this latter point consider a  Popescu--Rohrlich (PR) box \cite{pr94} regarded as a distribution on the punctured torus.
This distribution cannot be extended to a distribution on the torus as demonstrated in   Fig.~(\ref{fig:Bell-punctured-torus-ext}) (where we use the convenient representation given in Fig.~(\ref{fig:Triangle-meas-faces-b})):  Let $p_+$ be the distribution on the context $(x_0\oplus y_0, x_1\oplus y_1)$ and $p_-$ be the one on $( x_1\oplus y_0,x_0\oplus y_1)$. The marginals on the $d_0$ and $d_2$ faces implies that $p_+$ is the deterministic distribution $\delta^{10}$ assigning $(x_0\oplus y_0, x_1\oplus y_1)\mapsto (1,0)$ and $p_-=\delta^{00}$ assigning $( x_1\oplus y_0,x_0\oplus y_1)\mapsto (0,0)$. But then the marginals on the $d_1$ faces do not match since $d_1p_+=0\neq 1 = d_1p_-$. 
This is a general feature of contextual distributions on the CHSH scenario. In Corollary \ref{cor:extension-torus} we show that a distribution on the punctured torus extends to a distribution on the torus if and only if it 
is noncontextual.

}
}

\Rem{
{\rm
Distributions on spaces as illustrated by these examples generalize the theory of nonsignaling distributions and adds an extra layer of flexibility in choice of an underlying space.
In the appendix, Theorem \ref{thm:ComparisonToSheaf} shows that the  sheaf-theoretic formulation of nonsignaling distributions embed into our simplicial framework. For example, the CHSH scenario discussed in Example \ref{ex:Bell222} when regarded as a discrete scenario can be realized as a distribution on a $1$-dimensional space; see Fig.~(\ref{fig:Bell-discrete}). However, by realizing it over a $2$-dimensional space such as the punctured torus as in  Fig.(\ref{fig:Bell-punctured-torus}), or on a square as in (left) Fig.~(\ref{fig:S-H}) reveals more intricate features of contextuality.
This topological freedom plays an important role in the topological proof of Fine's theorem (Theorem \ref{thm:Fines-theorem}) and the characterization of contextuality in terms of extensions (Corollary \ref{cor:extension-torus}). 
}
}

\section{Distributions on spaces and contextuality}
\label{sec:DistSpace}
   
In Section \ref{sec:MotIdea} we explained in an informal way how to interpret nonsignaling distributions as distributions on spaces.
Here we introduce our simplicial framework in a more rigorous fashion. 
We define the notion of a simplicial scenario which consists of a pair of spaces representing both measurements and outcomes.
Then we introduce simplicial distributions on these scenarios generalizing the nonsignaling distributions defined for {\it discrete scenarios} (see Appendix \ref{sec:discrete-scenarios}).
Contextuality defined at this level of generality subsumes the usual notion for the discrete case. Central to our framework is the theory of simplicial sets.
They provide combinatorial descriptions of topological spaces and are the main objects of study in modern homotopy theory \cite{goerss2009simplicial}.



\subsection{Space of measurements}
\label{sec:SpaceOfMeas}

In  the simplicial framework we will work with spaces of measurements represented by simplicial sets.
In this section we introduce simplicial sets as combinatorial models of spaces and describe some of the measurement spaces that appear in Section \ref{sec:MotIdea}.

The topological $n$-simplex consists of the points $(t_0,t_1,\cdots,t_n)$ in $\RR^{n+1}$ such that each $t_i\geq 0$ and $\sum_i t_i=1$. Such a simplex has $(n+1)$-faces. Intuitively a simplicial set is a collection of ``abstract" $n$-simplices representing a topological $n$-simplex together with face maps, telling us how to glue these simplices; and the degeneracies  allowing us the extra freedom to collapse some of the irrelevant simplices. More formally, a
 {\it simplicial set} $X$ consists of the following data:
\begin{itemize}
\item A sequence of sets  $X_0,X_1,\cdots,X_n,\cdots$ for $n\geq 0$ where
each $X_n$ represents the set of $n$-simplices.
\item Face maps
$$
d_i:X_n \to X_{n-1}\;\;\;\text{ for } 0\leq i\leq n
$$
representing the faces of a given simplex.
\item Degeneracy maps
$$
s_j:X_n \to X_{n+1}\;\;\;\text{ for } 0\leq j\leq n
$$
representing the degenerate simplices.
\end{itemize}
The face and the degeneracy maps are subject to the {\it simplicial identities} given in the appendix, Eq.~(\ref{eq:simplicial-identities}).
See Appendix \ref{sec:app-SimplicialIdentities} for more on these identities and \cite{friedman2008elementary} for a reader-friendly introduction to simplicial sets. 
An $n$-simplex is called {\it degenerate} if it  lies in the image of a degeneracy map, otherwise it is called {\it nondegenerate}. Geometrically only the nondegenerate simplices are relevant.
Among the nondegenerate simplices there are ones that are not a face of another nondegenerate simplex. 
Those simplices we will refer to as {\it generating simplices}. 
\co{The set of generating simplices of a simplicial set is nonempty unless the set of $n$-simplices is empty for all $n\geq 0$:  If a generating simplex of dimension $n>1$ does not exist then all the $0$-simplices are generating.   
}
To illustrate the idea let us consider the simplicial set $\Delta^d$ representing the topological simplex of dimension $d$:
\begin{itemize}
\item The set of $n$-simplices is given by
$$
\set{\sigma^{a_0a_1\cdots a_n}\,:\, 0\leq a_0\leq \cdots \leq a_n \leq d}.
$$
\item The face map $d_i$ acts by deleting the $i$-th element
$$
d_i \sigma^{a_0a_1\cdots a_n} = \sigma^{a_0a_1\cdots  a_{i-1}a_{i+1} \cdots a_n}.
$$ 
\item The degeneracy map $s_j$ acts by copying the $j$-th element
$$
s_j \sigma^{a_0a_1\cdots a_n} = \sigma^{a_0a_1\cdots  a_j a_j \cdots a_n}.
$$ 
\end{itemize}
Nondegenerate simplices are given by $\sigma^{a_0a_1\cdots a_n}$ with no repetition in $a_0a_1\cdots a_n$, and the only generating simplex is given by $\sigma^{01\cdots d}$. Note that any other simplex can be obtained from $\sigma^{01\cdots d}$ by applying a sequence of face and degeneracy maps, hence the name generating.

In our framework $X$ will represent a {\it space of measurements}.
The $n$-simplices of $X$ will represent {\it $n$-dimensional contexts}, or briefly   {\it $n$-contexts}. 

\Ex{\label{ex:SingleTriangle-MeasSpace}
{\rm
The triangle scenario in Example \ref{ex:SingleTriangle} consists of two measurements $\set{x,y}$ assembled into a triangle by adding the third measurement $x\oplus y$. The measurement space in this case is $\Delta^2$ whose generating simplex will be denoted by $\sigma_{xy}^{012}$, and 
the three faces are denoted by $\sigma^{12}_y$, $\sigma^{02}_{x\oplus y}$ and $\sigma_x^{01}$ to indicate the corresponding measurements.  
}}

To construct more complicated spaces of measurements one can start from more than one generating simplex, which can live in different dimensions, and specify which faces and degeneracies produced from these distinct generating simplices are identified. 

\Ex{\label{ex:TwoTriangles-MeasSpace}
{\rm
The diamond scenario in Example \ref{ex:TwoTriangles} with contexts $\set{x_0,y_0}$ and $\set{x_1,y_1}$  assembled into two triangles glued along a common $d_1$-face  can be represented by a measurement space $Z$ defined as follows:
\begin{itemize}
\item Generating $2$-simplices: $\sigma_{x_0y_0}$ and $\sigma_{x_1y_1}$.
\item Identifying relation:
$
d_1  \sigma_{x_0y_0} = d_1 \sigma_{x_1y_1}.
$
\end{itemize}
}}
 
\Ex{\label{ex:Bell-measurement-space}
{\rm
The CHSH scenario of Example \ref{ex:Bell222} consists of four triangles organized into a punctured torus $T^\circ$ defined as follows:
\begin{itemize}
\item Generating $2$-simplices: $\sigma_{y_0x_0}$, $\sigma_{y_0x_1}$, $\sigma_{x_0y_1}$ and $\sigma_{x_1y_1}$. 
\item Identifying relations:
\begin{equation}\label{eq:Bell-identifying-relations}
\begin{aligned}
d_0 \sigma_{y_0x_0} &=  d_2 \sigma_{x_0y_1} \\
d_0 \sigma_{y_0x_1} &= d_2 \sigma_{x_1y_1}   \\
d_2 \sigma_{y_0x_0} &=  d_2 \sigma_{y_0x_1}    \\
d_0 \sigma_{x_0y_1} &=  d_0 \sigma_{x_1y_1}. \\
\end{aligned}
\end{equation}
\end{itemize}
} }

\subsection{Space of outcomes and distributions}
\label{sec:SpaceOutcomesDist}

In our framework a {\it space of outcomes} will be represented by a simplicial set $Y$. An $n$-simplex represents an {\it $n$-dimensional outcome}, or briefly an {\it $n$-outcome}.
We will see that distributions on the set of $n$-outcomes can also be assembled into a simplicial set. 
This formalizes the intuitive picture given in  
Section \ref{sec:MotIdea}.
We then specialize to a particular outcome space known as the nerve space, which makes precise the notion of an XOR outcome described in Section \ref{sec:chsh-scenario}.

Let $R$ denote a commutative semiring, e.g. the nonnegative reals $\nnegR$ or the Boolean algebra $\BB=\set{0,1}$. 
An {\it $R$-distribution}, or simply a {\it distribution}, on a set $U$ is a function 
$
p:U\to \rev{R}
$
of finite support, i.e. $p(u)\neq 0$ for finitely many $u\in U$, such that $\sum_{u\in U} p(u)=1$ \cite{jacobs2010convexity}. Given a function $f:V\to U$ and  a distribution $p\in D_R(V)$ one can define a distribution $D_Rf(p)$ on $U$ by the assignment
$$
u \mapsto \sum_{v\in f^{-1}(u)} p(v).
$$

\Def{\label{def:simp-dist}
{\rm
Let $Y$ be a simplicial set representing an outcome space. The {\it space $D_RY$ of  distributions} on $Y$ is the simplicial set defined as follows: 
\begin{itemize}
\item The set of $n$-simplices  is given by the set $D_R(Y_n)$ of distributions on $Y_n$ for $n\geq 0$.
\item The simplicial structure maps are given by $D_Rd_i$ and $D_Rs_j$.
\end{itemize}
For simplicity of notation we will write $d_i$ and $s_j$ for these simplicial structure maps. 
}
}

Motivated by the topological approach of \cite{Coho}, an essential feature of our framework is that quantum observables, or abstract measurements more broadly, are assigned to edges of a topological space as in Fig.~(\ref{fig:Triangle-faces}) and Fig.(\ref{fig:Triangle-meas-faces-a}), respectively. Following the discussion of Section \ref{sec:MotIdea}, the triangle in Fig.~(\ref{fig:Triangle-faces}) is represented by a pair $(A,B)$ of observables; similar notions apply to abstract measurements $(x,y)$ in Bell scenarios. 
We can also represent the pair $(A,B)$ as a projective measurement given in Fig.~(\ref{fig:Triangle-proj-faces}) on the collection of triangles whose edges are labeled by the outcomes:  $\sigma^{01}\mapsto a$, $\sigma^{12}\mapsto b$ and $\sigma^{02}\mapsto a+b$. 
We will generalize this idea to a set of compatible measurements $x_1,x_2,\cdots,x_n$ labeling the edges of an $n$-simplex: $\sigma^{(i-1)i}\mapsto x_i$ for $1\leq i\leq n$, where the remaining edges are XOR measurements whose outcomes are inferable from those of $x_i$.  In this case the outcomes associated with an XOR measurement are inferred by summing the outcomes of the performed measurements. (The sum is modulo $d$ if the outcomes take values in $\ZZ_d$.
This will be made more precise in Eq.~(\ref{eq:outcome-map-nerve-edges}).)
 Suppose the measurements come from quantum observables that pairwise commute. Then  the associated projective measurement will be on the collection of $n$-simplices labeled by tuples $(a_1,a_2,\cdots,a_n)$ of outcomes generalizing the two-dimensional case. There is a nice way to assemble these simplices into a space, known as the nerve space\footnote{In algebraic topology nerve spaces are also known as {\it classifying spaces}. They play a prominent role in bundle theory and group cohomology \cite{adem2013cohomology}.}.

\Def{\label{def:NZd}
{\rm
The {\it nerve space} $N\ZZ_d$ is the simplicial set 
whose set of   $n$-simplices  consists of  $n$-tuples $(a_1,a_2,\dots,a_n)$ of outcomes in $\ZZ_d$ together with the face and the degeneracy maps\footnote{More generally, the nerve construction can be applied to a group $G$. Observe that a face map multiplies two successive elements in a tuple (except the first and the last one). A degeneracy map basically inserts the identity element.  
}
\begin{equation}\label{eq:nerve-Face-Deg}
\begin{aligned}
d_i(a_1,a_2,\cdots,a_n) &= \left\lbrace
\begin{array}{ll}
(a_2,a_3,\cdots,a_n) & i=0 \\
(a_1,a_2\dots,a_i+a_{i+1},\cdots,a_n) & 0<i< n \\
(a_1,a_2,\cdots, a_{n-1}) & i=n 
\end{array} 
\right. \\
s_j(a_1,a_2,\cdots,a_n) &= (a_1,a_2,\cdots,a_{j-1},0,a_j,\cdots,a_n)\;\;\;\;\; 0\leq j\leq n.
\end{aligned}
\end{equation}
Note that $(N\ZZ_d)_0$ consists of the empty tuple $(\,)$, the unique $0$-simplex.
}}

Next we describe the space $D_R N\ZZ_d$ of distributions on the nerve:  
An $n$-simplex of  $D_R N\ZZ_d$ is a distribution of the form $p:\ZZ_d^n \to R$ and the simplicial structure maps are given by 
$$
\begin{aligned}
d_ip (a_1,\cdots,a_{n-1}) &= \left\lbrace
\begin{array}{ll}
\sum_{a} p(a,a_1,\cdots,a_{n-1}) & i=0 \\
\sum_{a} p(a_1,\cdots, a,a_i-a,\cdots a_{n-1}) & \rev{0<i<n} \\
\sum_{a} p(a_1,\cdots,a_{n-1},a) & i=n 
\end{array}
\right. \\
s_jp(a_1,\cdots,a_{n+1}) &= \left\lbrace
\begin{array}{ll}
 p(a_1,\cdots,a_{j-1},a_{j+1},\cdots,a_{n+1})  & a_j = 0 \\
 0 & \text{otherwise.}
\end{array}
\right.
\end{aligned}
$$
\Nota{{\rm \label{nota:DNZ2}
To describe $D_RN\ZZ_2$   
in dimensions $\leq 2$ we will use the following notation:
For $n=2$ we identify $p\in D_R(\ZZ_2^2)$ with the tuple $(p^{00},p^{01},p^{10},p^{11})$.  
There are three face maps $d_i:D_R(\ZZ_2^2)\to D_R(\ZZ_2)$  given by
\begin{equation}\label{eq:DRNZ2-face}
d_i(p^{00},p^{01},p^{10},p^{11}) = \left\lbrace
\begin{array}{ll}
p^{10}+p^{00} & i=0 \\
p^{11}+p^{00} & i=1 \\
p^{01}+p^{00} & i=2.
\end{array}
\right.
\end{equation}
For $n=1$ the distribution $p$ is identified with $(p^0,p^1)$, and since $p^0+p^1=1$ it suffices to keep $p^0$. There is  one face map $d_0=d_1:D_R(\ZZ_2) \to \set{1}$ sending every distribution to $1$ and there are two degeneracy maps $s_j:D_R(\ZZ_2) \to D_R(\ZZ_2^2)$ given by
$$
s_j(p^0) = \left\lbrace
\begin{array}{ll}
(p^0,p^1,0,0) & j=0 \\
(p^0,0,p^1,0) & j=1.  
\end{array}
\right.
$$
For $n=0$ there is a unique distribution $D_R\ZZ_2^0 = \set{1}$. The degeneracy map $s_0:\set{1} \to \ZZ_2$ is defined by $s_0(1)=\delta^0$, where $\delta^0(0)=1$ and $\delta^0(1)=0$.
}}



\subsection{Contextuality for simplicial distributions}\label{sec:ContSimpSce}

We begin with the notion of a map between spaces in the simplicial set formalism.  
A {\it map $f:X\to Z$ of simplicial sets}, or a {\it map of spaces}, consists of a sequence of functions
\begin{itemize}
\item $f_n:X_n\to Z_n$ for $n\geq 0$,
\item compatible with the face and the degeneracy  maps in the sense that
\begin{equation}\label{eq:Simp-Map-Eq}
d_i f_n(\sigma) = f_{n-1}(d_i \sigma)\;\; \text{ and } \;\;s_j f_n(\sigma) = f_{n+1}(s_j \sigma)
\end{equation}
for all $0\leq i,j\leq n$ and $\sigma\in X_n$.
\end{itemize}
We will write $f_\sigma$ for the image $f_n(\sigma)$. With this notation Eq.~(\ref{eq:Simp-Map-Eq}) can be written as $d_if_\sigma = f_{d_i\sigma}$ and $s_j f_\sigma = f_{s_j\sigma}$. 

\Rem{\label{rem:desc-a-space-map}
{\rm
A map $f:X\to Y$ of spaces is determined by its values on the generating simplices. \co{(Here we assume that $X_n\neq \emptyset$ for some $n\geq 0$. Otherwise, the map is unique.)}
Let $\set{\sigma_i\,:\,\sigma_i\in X_{m_i}}$ denote the set of generating simplices of $X$.
Then $f$ consists of assignments $\sigma_i\mapsto f_{\sigma_i}$ for each generating simplex, where $f_{\sigma_i}$ is a $m_i$-simplex of $Y$, compatible with the simplicial structures.
The latter condition simply amounts to the requirement that if $\theta$ is a $m$-simplex of $X$ which can be obtained from two distinct generating simplices, say $\sigma_i$ and $\sigma_j$, then the assignment $\theta\mapsto f_\theta$ determined by the assignments for $\sigma_i$ and $\sigma_j$ should match. More precisely, if $\theta=\alpha_i(\sigma_i)=\alpha_j(\sigma_j)$ for some maps $\alpha_i$ and $\alpha_j$ given by a composition of faces and degeneracies then our assignment has to satisfy $\alpha_i(f_{\sigma_i})= \alpha_j(f_{\sigma_j})$.
In the special case  $X=\Delta^n$ space maps $f:\Delta^n\to Y$ are in bijective correspondence with the elements of $Y_n$.
}}

 Next we introduce distributions in the simplicial framework  generalizing the notion of nonsignaling distributions.

\Def{\label{def:SimpScenario}
{\rm
A {\it simplicial scenario} is a pair  $(X,Y)$ consisting of a space $X$ of measurements and a space $Y$ of outcomes. A {\it simplicial distribution} on this scenario  is a map $p:X\to D_RY$ of spaces. We write $\siDist(X,Y)$ for the set of simplicial distributions on $(X,Y)$. 
}
}

A map $r:X\to Y$ of spaces will be called an {\it outcome map}, or sometimes an {\it outcome assignment}. There is an associated simplicial distribution 
$\delta^r:X\to D_RY$  defined by sending an $n$-context $\sigma$ to the delta distribution $\delta^r_\sigma=\delta^{r_\sigma}$ on the set of $n$-outcomes. For $\theta\in Y_n$ we have
\begin{equation}\label{eq:delta-dist}
\delta^{r_\sigma}(\theta) = \left\lbrace
\begin{array}{ll}
1 & \theta=r_\sigma \\
0 & \text{otherwise.}
\end{array}
\right.
\end{equation}
A {\it deterministic distribution} on a  scenario $(X,Y)$ is a simplicial distribution of the from $\delta^r:X\to D_RY$. We write $\dDist(X,Y)$ for the set of deterministic distributions. 
A {\it classical distribution}  is a probabilistic mixture of deterministic distributions: $d=\sum_r \lambda(r)\,\delta^r $ where $\lambda(r)\geq 0$ and $\sum_r \lambda(r)=1$.
We denote the set of classical distributions by $\clDist(X,Y)$.

In the simplicial setting the definition of contextuality relies on the following map:
\begin{equation}\label{diag:Theta}
\Theta: \nonC(X,Y) \to \siDist(X,Y)
\end{equation}
that sends a classical distribution $d=\sum_r \lambda(r)\,\delta^r$ 
to the simplicial distribution $\Theta(d):X\to D_RY$  defined by
$$
\Theta(d)_\sigma: \theta \mapsto \sum_{r\,:\,r_\sigma=\theta} \lambda(r) 
$$
where $\sigma\in X_n$, $\theta\in Y_n$ and $r:X\to Y$ runs over outcome assignments such that $r_\sigma=\theta$. This makes precise the intuitive pictures that appear in the examples of Section \ref{sec:ContNewFramework}.

\Rem{{\rm \label{rem:notation-C-S}
By definition classical distributions are 
$R$-convex combinations 
of deterministic distributions, that is $\clDist(X,Y)=D_R(\dDist(X,Y))$. 
On the other hand, we can take 
convex combinations
in  $\siDist(X,Y)$: Given  simplicial distributions $q^i$ on $(X,Y)$ and $\lambda_i\in R$ with $\lambda_i\geq 0$ and $\sum_i \lambda_i=1$ we can construct $p=\sum_i \lambda_i\,q^i  \in \siDist(X,Y)$ 
by defining $p_\sigma=\sum_i \lambda_i\,q^i_\sigma$ for an $n$-context $\sigma\in X_n$.  
To avoid confusion we will indicate which convex combination we mean. 
We emphasize that in general $\clDist(X,Y)$ is not a subset of $\siDist(X,Y)$.
}}

\Def{\label{def:simp-contextuality}
A simplicial distribution $p \in \siDist(X,Y)$ is called {\it contextual} if it does not lie in the image of $\Theta$. Otherwise, it is called {\it noncontextual}.
}

This definition generalizes the usual notion of contextuality for nonsignaling distributions; see Definition \ref{def:contextual-discrete} and Theorem \ref{thm:ComparisonToSheaf}. 
In this sense our framework extends the sheaf-theoretic approach \cite{abramsky2011sheaf} that formalizes the theory of contextuality for nonsignaling distributions.

\Ex{\label{ex:SingleTriangle-SimpScenario}
{\rm 
Using Remark \ref{rem:desc-a-space-map} we can describe  distributions on the simplicial scenario $(\Delta^n,Y)$ for an arbitrary outcome space $Y$. We have $D(\Delta^n,Y)=Y_n$ for the deterministic distributions.
This implies that the set of classical distributions is given by
$$
\clDist(\Delta^n,Y) = D_R(Y_n).
$$
Similarly,  simplicial distributions $p:\Delta^n \to D_R Y$ are  given by  
$$
\siDist(\Delta^n,Y) =( D_RY)_n=D_R(Y_n).
$$ 
Consequently the $\Theta$-map in 
(\ref{diag:Theta}) is the identity map, and hence every distribution on the simplicial scenario $(\Delta^n,Y)$ is noncontextual. The triangle scenario of Example  \ref{ex:SingleTriangle} is the special case where $n=2$ and $Y=N\ZZ_2$.
}}

The notion of contextuality depends on the underlying semiring $R$. We will say $R$-(non)contextual when we want to emphasize this dependence. Similarly this dependence will be emphasized for the simplicial and classical distributions by including $R$ in the notation. We write $\siDist_R(X,Y)$ and $\clDist_R(X,Y)$ for the set of simplicial $R$-distributions and classical $R$-distributions; respectively. A semiring homomorphism $\phi:R\to S$ gives us a commutative diagram
\begin{equation}\label{eq:change-of-R}
\begin{tikzcd}
\clDist_R(X,Y) \arrow[r,"\Theta"] \arrow[d,"\phi_*"] & \siDist_R(X,Y) \arrow[d,"\phi_*"]\\ 
\clDist_{S}(X,Y) \arrow[r,"\Theta"] & \siDist_{S}(X,Y)
\end{tikzcd}
\end{equation}
The vertical maps are obtained simply by applying $\phi$ to the distributions on the simplices of $Y$ for a simplicial distribution $p:X\to D_RY$, or to the coefficients of a classical distribution $d=\sum_r \lambda(r)\,\delta^r$.

\Pro{\label{pro:change-of-semiring}
A simplicial $R$-distribution $p$ is $R$-contextual if the simplicial $S$-distribution $\phi_*(p)$ is $S$-contextual.
}
\Proof{
Let $p$ be a simplicial $R$-distribution. If it is noncontextual, say $\Theta(d)=p$ for some classical $R$-distribution, then $\bar d=\phi_*(d)$ satisfies $\Theta(\bar d)=\phi_*(p)$ showing that $\phi_*(p)$ is noncontextual as well.
}

Probability distributions will be our main application. A simplicial $\nnegR$-distribution will also be called a {\it simplicial probability distribution}. 
An important example of a semiring homomorphism is $\phi:\nnegR\to \BB$ that sends positive numbers to $1$ and zero to $0$. A simplicial $\nnegR$-distribution is said to be {\it logically (or possibilistically) contextual} if $\phi_*(p)$ is $\BB$-contextual. This extends the usual definition for nonsignaling distributions \cite{abramsky2012logical}. A famous example is Hardy's bipartite scenario \cite{hardy1993nonlocality}. 
In Section \ref{sec:StrongContCoho} we introduce strong contextuality for simplicial distributions and explore its connection with cohomology.
Now we turn to the examples of Section \ref{sec:ContNewFramework} and revisit them formally. 

\subsection{Simplicial scenarios with nerve as the outcome space}

Our canonical choice for the outcome space will be $N\ZZ_d$. We begin by a description of outcome assignments for a simplicial scenario $(X,N\ZZ_d)$ where $X$ is an arbitrary measurement space.

\Pro{\label{pro:outcome-map-nerve-edges}
Outcome maps $r:X\to N\ZZ_d$  are in bijective correspondence with functions $f:X_1\to \ZZ_d$ satisfying 
\begin{equation}\label{eq:s-boundary}
f(d_1\sigma) = f(d_2\sigma) +f(d_0\sigma)
\end{equation}
for all $\sigma\in X_2$.
Moreover, for a generating $n$-simplex $\sigma=\sigma^{01\cdots n}$ of  $X_n$ 
and $r_{\sigma} = (a_{1},a_{2},\cdots,a_{n})$ with $1\leq k<i\leq n$ we have
\begin{equation}\label{eq:outcome-map-nerve-edges}
r_{\sigma^{(i-k)i}} = a_{i-k+1}+a_{i-k+2}+\cdots+a_i.
\end{equation}
}
\Proof{ 
We will use the description of a space map given in Remark \ref{rem:desc-a-space-map}.
First we prove the statement for $X=\Delta^n$, in which case it suffices to prove Eq.~(\ref{eq:outcome-map-nerve-edges}).
In this case $r$ is determined by an $n$-tuple $r_\sigma=(a_1,a_2,\cdots,a_n)$ of outcomes in $\ZZ_d$. 
We will do induction on $k$. 
Let us define 
\begin{equation}\label{eq:varphi-k}
\varphi_i = \underbrace{d_0\circ d_0 \cdots \circ d_0}_{i-1}\circ \; d_{i+1}\circ d_{i+2}\circ \cdots \circ d_n.
\end{equation}
Using compatibility of $r$ with the face maps and Eq.~(\ref{eq:nerve-Face-Deg})  we obtain
$
r_{\sigma^{(i-1)i}} = r_{\varphi_i \sigma} = \varphi_i r_{\sigma} = a_i$ proving the statement for $k=1$. For $k>1$ assume that Eq.~(\ref{eq:outcome-map-nerve-edges}) holds for $k-1$. Consider the triangle $\sigma^{(i-k)(i-1)i}$ with vertices $(i-k,i-1,i)$. By compatibility with $d_1$ we have
$$
\begin{aligned}
r_{\sigma^{(i-k)i}} &= d_1 r_{\sigma^{(i-k)(i-1)i}} \\
& = d_1(r_{\sigma^{(i-k)(i-1)}},r_{\sigma^{(i-1)i}})\\
& = r_{\sigma^{(i-k)(i-1)}}+r_{\sigma^{(i-1)i}}\\
& = (a_{i-k+1}+a_{i-k+2}+\cdots+a_{i-1})+a_i
\end{aligned}
$$
where in the last step we used our induction hypothesis and the $k=1$ case; see Fig.~(\ref{fig:tetrahedron-ai}). 
\begin{figure}[h!]
\begin{center}
\includegraphics[width=.25\linewidth]{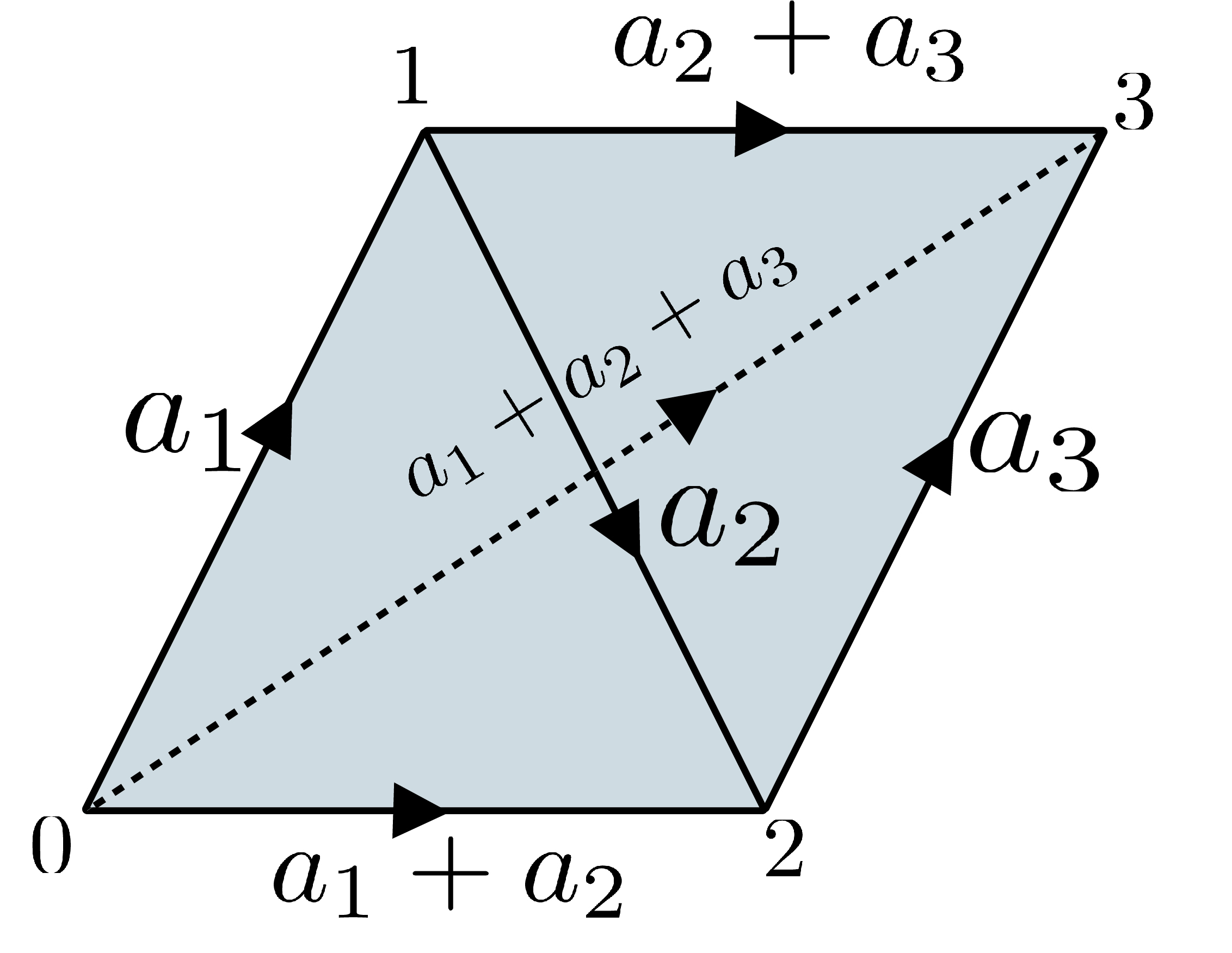} 
\end{center} 
\caption{To show that the edge $(0,3)$ is assigned $a_1+a_2+a_3$ we can use the triangle $\sigma^{023}$. For this we need the assignment for the edge $(0,2)$, which can be determined by using the triangle $\sigma^{012}$. 
}
\label{fig:tetrahedron-ai}
\end{figure}
For the general case an outcome map $r:X\to N\ZZ_d$ consists of assignments $\sigma_i\mapsto r_{\sigma_i}$ compatible in the sense that $\alpha_i(r_{\sigma_i})=\alpha_j(r_{\sigma_j})$ whenever $\alpha_i(\sigma_i)=\alpha_j(\sigma_j)$.  
Each such assignment specifies a map $r|_{\sigma_i}:\Delta^{d_i}\to N\ZZ_d$, a map determined by its restriction to $1$-contexts, i.e. the edges of the simplex $\sigma_i$. Compatibility follows if we require that $r|_{\sigma_i}$ and $r|_{\sigma_j}$ matches when restricted to the edges of $\theta$. This gives the bijective correspondence between outcome assignments and functions on the set of $1$-contexts satisfying Eq.~(\ref{eq:s-boundary}).
}

\Rem{\label{rem:f-vanishes-on-degenerate-simplices}
{\rm
A function $f:X_1\to \ZZ_d$ satisfying 
Eq.~(\ref{eq:s-boundary}) vanishes on degenerate edges. That is, if $e = s_0v$ for some $v\in X_0$ then  the simplicial identities  given in Eq.~(\ref{eq:simplicial-identities})  implies that $d_i \sigma = e$ for all $i=0,1,2$ when $\sigma = s_0e$. Thus  Eq.~(\ref{eq:s-boundary}) gives 
$f(e)=0$.
}}

\Ex{\label{ex:diamond-noncontextual}
{\rm 
Let $Z$ denote the measurement space of the diamond scenario
in Example \ref{ex:TwoTriangles-MeasSpace}.  
An outcome map $r:Z\to N\ZZ_2$ consists of a pair $(r_{x_0y_0},r_{x_1y_1})$ of $2$-outcomes, where $r_{x_iy_i}=(a_i,b_i)$ is the image of the generating simplex $\sigma_{x_iy_i}$ for $i=0,1$, such that $d_1 r_{x_0y_0} = d_1r_{x_1y_1}$. Denoting by $r_{x_i}$ and $r_{y_j}$ the $1$-outcomes associated to the $1$-contexts $\sigma_{x_i}$ and $\sigma_{y_j}$,
Proposition \ref{pro:outcome-map-nerve-edges} implies that $r_{x_i}=a_i$ and $r_{y_j}=b_j$.
Classical distributions on this scenario are given by
$$
\clDist(Z,N\ZZ_2) = D_R\set{(a_0,b_0,a_1,b_1)\in \ZZ_2^4\,:\, a_0+b_0=a_1+b_1}.
$$
A simplicial distribution $p:Z \to D_RN\ZZ_2$ is specified by a pair of distributions $(p_{x_0y_0},q_{x_1y_1})$ such that $d_1p_{x_0y_0} = d_1 q_{x_1y_1}$. 
Therefore using Notation \ref{nota:DNZ2} we have
$$
\siDist(Z,N\ZZ_2) =\set{(p^{00},p^{01},p^{10},p^{11}, q^{00},q^{01},q^{10},q^{11})\in [0,1]^8:\, \sum_{ab}p^{ab}=\sum_{cd}q^{cd}=1,\; p^{00}+p^{11} = q^{00}+q^{11} }. 
$$
The formula in Eq.~(\ref{eq:lambda}) for $\lambda$ in terms of the initial pair $(p,q)$ of distributions implies that $\Theta$ is surjective. Therefore every  distribution on this simplicial scenario is noncontextual. A generalization of this formula (given in Eq.(\ref{eq:classical-ansatz})) appears in the proof of   Lemma \ref{lem:gluing-lemma}, a result on gluing classical distributions.
}}
 
 

\Ex{\label{ex:Bell-distributions}
{\rm
Next we describe the set of deterministic and classical distributions on the CHSH scenario. The measurement space is the punctured torus $T^\circ$ described in  Example \ref{ex:Bell-measurement-space}.
An outcome assignment   $r:T^\circ\to N\ZZ_2$  is determined by the images of the four generating simplices: $r_{y_0x_0}$, $r_{y_0x_1}$, $r_{x_0y_1}$ and $r_{x_1y_1}$. Each of these $2$-outcomes are specified by a pair in $\ZZ_2$. 
They further satisfy a set of relations imposed by the ones among the generating simplices given in Eq.~(\ref{eq:Bell-identifying-relations}):
\begin{equation}\label{eq:Bell-identifying-relations-r}
\begin{aligned}
d_0 r_{y_0x_0} &=  d_2 r_{x_0y_1} \\
d_0 r_{y_0x_1} &= d_2 r_{x_1y_1}   \\
d_2 r_{y_0x_0} &=  d_2 r_{y_0x_1}    \\
d_0 r_{x_0y_1} &=  d_0 r_{x_1y_1}. \\
\end{aligned}
\end{equation}
Using these equations together with Proposition \ref{pro:outcome-map-nerve-edges} we obtain  $r_{y_0x_0}=(b_0,a_0)$, $r_{y_0x_1}=(b_0,a_1)$, $r_{x_0y_1}=(a_0,b_1)$ and $r_{x_1y_1}=(a_1,b_1)$ for some $a_i,b_j\in \ZZ_2$. Note that $r_{x_i}=a_i$ and $r_{y_j}=b_j$.
Therefore $\dDist(T^\circ,\ZZ_2)=\ZZ_2^4$ where under this identification a quadruple $(a_0,b_0,a_1,b_1)$ uniquely specifies the outcome map $r$, and the set of classical distributions is given by
$$
\clDist(T^\circ,N\ZZ_2) = D_R(\ZZ_2^4).
$$
A simplicial distribution  $p:T^\circ\to D_R N\ZZ_2$ is determined by the images of the generating simplices, given by the  distributions $p_{y_0x_0},p_{y_0x_1},p_{x_0y_1},p_{x_1y_1}$ on $\ZZ_2^2$, together with a set of relations imposed by Eq.~(\ref{eq:Bell-identifying-relations}). This set of relations is precisely the one given in Eq.~(\ref{eq:nonsignaling-top}), i.e. the  usual nonsignaling conditions for the CHSH scenario.
 For a deterministic distribution $r$ represented by a quadruple $(a_0,b_0,a_1,b_1)$ the distribution $\delta^r:T^\circ\to D_R N \ZZ_2$  sends each of the generating simplices to a delta distribution: $\delta^{b_0a_0},\delta^{b_0a_1},\delta^{a_0b_1},\delta^{a_1b_1}$. In general, a classical distribution $d=\sum_r \lambda(r)\,\delta^r$, where $r$ runs over the quadruples $(a_0,b_0,a_1,b_1)$, will be mapped to a simplicial distribution $p=\Theta(d)$ such that
$$
\begin{aligned}
p_{y_0x_0}^{b,a} = \sum_{r\,:\, r_{y_0x_0}=(b,a)} \lambda(r) = \sum_{a_1,b_1} \lambda(a,b,a_1,b_1), 
\end{aligned}
$$
 similarly for the remaining distributions: $p_{y_0x_1},p_{x_0y_1},p_{x_1y_1}$. This means that $p$ is precisely a classical distribution in the usual sense. Therefore for the CHSH scenario contextuality in the sense of Definition \ref{def:simp-contextuality} coincides with the usual notion of contextuality (see Definition \ref{def:contextual-discrete}).  
}}

\subsection{Distributions on the circle} \label{sec:DistOnCircle}

So far we have been considering the nerve space $N\ZZ_d$ as the space of outcomes. Now we will consider a subspace\footnote{A subspace (subsimplicial set) of a simplicial set $X$ is a simplicial set $Z$ such that $Z_n\subset X_n$ for all $n\geq 0$ together with the face and the degeneracy maps inherited from $X$.}, the circle $S^1$, as the outcome space and describe the corresponding distributions on a simplicial scenario $(X,S^1)$. Distributions on the circle will be important in the formulation of the Gleason's theorem (Theorem \ref{thm:Gleason}). Also in Section \ref{sec:state-indep-contex} the circle will play a role in the formulation of state-independent contextuality as a distinguished measurement space.

As a simplicial set the circle $S^1$ is defined as follows:
\begin{itemize}
\item Generating $1$-simplex: $\sigma^{01}$.
\item Identifying relation: $d_0\sigma^{01}=d_1\sigma^{01}$
\end{itemize}
\noindent
In other words, as in the topological case a circle is obtained from $\Delta^1$ by  identifying the two $0$-simplices: $\sigma^0= \sigma^1$. This identification also implies $\sigma^{00\cdots 0}=\sigma^{11\cdots 1}$ for the degenerate simplices in higher dimensions. Let us write $0_n=00\cdots 0$ for a string of zeros of length $n+1$ and similarly $1_n$ for a string of ones of the same length.
With this notation the $n$-simplices of $S^1$ are given by
\begin{equation}\label{eq:cicrcle-n-simplices}
\sigma^{0_{n}}= \sigma^{1_{n}},\;\; \sigma^{01_{n-1}}  ,\;\; \sigma^{001_{n-2}},\;\;\cdots,\;\;  \sigma^{0_{n-1}1}.
\end{equation}
For example, the $2$-simplices are given by $\sigma^{000}=\sigma^{111}$, $\sigma^{011}$ and $\sigma^{001}$; see Fig.~(\ref{fig:S1-dim2}).
\begin{figure}[h!]
\begin{center}
\includegraphics[width=.5\linewidth]{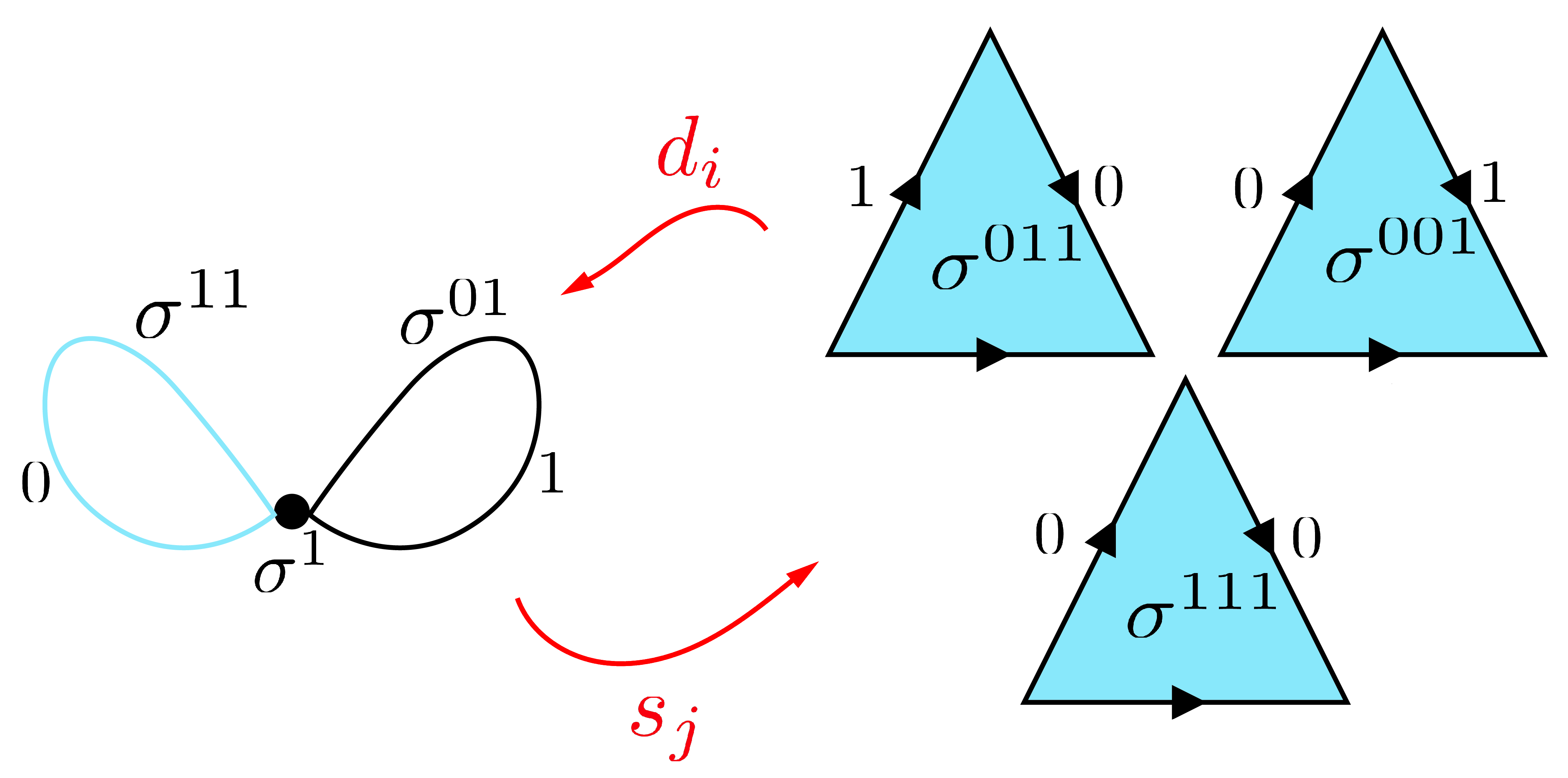} 
\end{center} 
\caption{The generating simplex of the circle is $\sigma^{01}$, an edge whose both faces are given by the unique vertex $\sigma^0=\sigma^1$.
All the other simplices can be obtained by acting by a degeneracy map. 
}
\label{fig:S1-dim2}
\end{figure}
Faces and degeneracies can be obtained by applying the deleting and copying operations to the string $01\cdots n$.
The simplicial structure maps in dimensions $\leq 2$ are illustrated in Fig.~(\ref{fig:S1-dim2}). Alternatively, the circle can be seen as a subspace of $N\ZZ_d$ if we map the simplices in Eq.~(\ref{eq:cicrcle-n-simplices}), respectively, to the following $n$-tuples:
\begin{equation}\label{eq:circle-alt-rep}
(0,0,\cdots,0),\;\; (1,0,\cdots,0),\;\;(0,1,\cdots,0),\cdots,\;\; (0,0,\cdots,0,1)
\end{equation}  
In particular, the $2$-simplices are mapped as follows in Fig.~(\ref{fig:S1-dim2}).
We will denote this map by
\begin{equation}\label{diag:h-map-S1-NZd}
h:S^1 \to N\ZZ_d.
\end{equation}

Next we describe the space of distributions on the circle. 
An $n$-simplex of $D_R S^1$ is a distribution $p:(S^1)_n\to R$ on the set of $n$-simplices listed in Eq.~(\ref{eq:cicrcle-n-simplices}). 
If we set $p^k=p(\sigma^{0_{k-1}1_{n-k}})$ then such distributions are in bijective correspondence with $n$-tuples $(p^1,p^2,\cdots,p^n)$ satisfying  $p^i\geq 0$ for all $1\leq i\leq n$ and  $\sum_{i=1}^np^i\leq 1$.
In this representation the faces and the degeneracies are given by 
$$
\begin{aligned}
d_i(p^1,p^2,\cdots,p^{n}) &= \left\lbrace
\begin{array}{ll}
(p^2,p^3,\cdots,p^{n}) & i=0\\ 
(p^1,p^2,\cdots,p^i+p^{i+1},\cdots, p^{n}) & 0<i<n \\
(p^1,p^2,\cdots,p^{n-1}) & i=n
\end{array}
\right.  \\
s_j(p^1,p^2,\cdots,p^{n}) &= (p^1,p^2,\cdots,p^j,0,p^{j+1},\cdots,p^n).
\end{aligned}
$$
We can see $D_RS^1$ as a subspace of $D_RN\ZZ_d$ via the map 
$$
D_Rh:D_RS^1 \to D_RN\ZZ_d
$$
where $h$ is as given in 
(\ref{diag:h-map-S1-NZd}).
Under this map a tuple $(p^1,p^2,\cdots,p^n)$ representing an $n$-simplex of $D_R S^1$ maps to the distribution $\tilde p:\ZZ_d^n \to R$ given by 
\begin{equation}\label{eq:embed-dist-on-S1-to-dist-on-nerve}
\tilde p^{a_1\cdots a_n} = \left\lbrace
\begin{array}{ll}
p^k & a_k=1 \text{ and } a_j=0 \text{ for all } j\neq k \\
1- \sum_k p^k & a_1=a_2=\cdots=a_n=0 \\
0 & \text{otherwise.}
\end{array}
\right.
\end{equation}
In particular, a $2$-simplex of $D_RS^1$ given by the tuple $(p^1,p^2)$ can be represented as $(p^{00},p^{01},p^{10},0)$ using Notation~(\ref{nota:DNZ2}) where $p^{10}=p^1$, $p^{01}=p^2$ and $p^{00}=1-p^1-p^2$.

\Ex{\label{ex:diamond-S1}
{\rm
For the diamond scenario an outcome map $r:Z\to S^1$ is determined by the $2$-outcomes $r_{x_iy_i}=(a_i,b_i)$, where $i=0,1$, as in the case where the outcome space is $N\ZZ_2$, however this time the pair can not be $(1,1)$. For a simplicial probability distribution $p:Z\to D_\nnegR S^1$ determined by the pair  $(p,q)$ of distributions
there is an additional constraint: $p^{11}=q^{11}=0$. 
Compatibility at the $d_1$-face then implies $p^{00}=q^{00}$. Therefore we have
$$
\siDist(Z,S^1) = \set{(p^{01},p^{10},q^{01},q^{10})\in [0,1]^4\,:\,  p^{01}+p^{10}=q^{01}+q^{10} \leq 1}.
$$
Any distribution on the scenario  $(Z,S^1)$ is still noncontextual. 
Given $p\in \siDist(Z,S^1)$ we can use  Eq.~(\ref{eq:lambda}) to construct a classical distribution $d=\sum_{abc} \lambda(abc)\, \delta^{abc}$ such that $\lambda(abc)=0$ whenever $(a,b)$ or $(c,a+b+c)$ is equal to $(1,1)$. This implies that $d\in \clDist(Z,S^1)$, which can be used to show that $p$ is noncontextual.
}
}

\section{Gluing and extending distributions}
\label{sec:Glue-Extend}

In the previous section a simplicial scenario was defined as a pair $(X,Y)$ of spaces. Using the flexibility afforded by the language of simplicial sets we can readily define maps between measurement spaces or outcome spaces. %
%
%
%
These maps allow us to compare simplicial distributions on two different scenarios and study their contextual behavior by comparison.
Here we formalize this approach and deduce new results on contextual properties of distributions using topological reasoning. 
Our main tool is the ``gluing lemma" formulated in Lemma \ref{lem:gluing-lemma} that significantly generalizes Fine's approach to gluing classical distributions.  
As an application of our formalism we provide a new topological proof of Fine's theorem \cite{fine1982hidden,fine1982joint} for the CHSH scenario (Theorem \ref{thm:Fines-theorem}).
Our approach is to study distributions on simpler measurement spaces then using the gluing lemma to draw conclusions about more complicated measurement spaces.
We also give a new characterization of contextuality in terms of extensions to larger measurement spaces as formulated in Corollary \ref{cor:extension-torus}.

\subsection{Changing measurement and outcome spaces} 

Let $X$ and $Z$ be spaces of measurements and $Y$ be a space of outcomes. A map
$f:Z\to X$ of spaces
induces the following functions distributions:
\begin{itemize}
\item $f^*:\dDist(X,Y)\to \dDist(Z,Y)$ defined by sending $\delta^r$ to the distribution $\delta^{r\circ f}$.  

\item $D_Rf^*: \clDist(X,Y) \to \clDist(Z,Y)
$ that sends  $d=\sum_r \lambda(r) \, \delta^r$ to the classical distribution $\sum_r \lambda(r)\, \delta^{r\circ f}$.

\item $f^*: \siDist(X,Y) \to \siDist(Z,Y)
$ defined by sending $p:X\to D_RY$ to the composition $Z\xrightarrow{f} X\xrightarrow{p} D_RY$.
\end{itemize}
When $f:Z\to X$ is the inclusion map of a subspace   we simply write $\delta^r|_Z$ (or $r|_Z$), $d|_Z$ and $p|_Z$ for the corresponding distributions on $Z$ obtained by restricting an outcome map, a classical distribution and a simplicial distribution; respectively. The $\Theta$-map is compatible with changing the measurement space in the sense that there is a commutative diagram
\begin{equation}\label{diag:comm-diag-meas-space-change-Theta}
\begin{tikzcd}
\clDist(X,Y) \arrow[d,"\co{D_Rf^*}"'] \arrow[r,"\Theta_{X,Y}"] & \siDist(X,Y) \arrow[d,"f^*"] \\
\clDist(Z,Y) \arrow[r,"\Theta_{Z,Y}"] & \siDist(Z,Y)
\end{tikzcd}
\end{equation}
In the notation of the $\Theta$-map we emphasize dependence on the measurement and the outcome spaces. Usually we omit the outcome space when only the measurement space is changed.
When $Z$ is a subspace of $X$ the commutativity of the diagram is expressed as
$
\Theta_X(d)|_Z = \Theta_Z(d|_Z).
$

\Pro{\label{pro:ext-noncontextual}
If $\tilde p\in \siDist(X,Y)$ is noncontextual then $p=f^*(\tilde p)$ is also noncontextual.
}
\Proof{
Since $\tilde p$ is noncontextual there exists a classical distribution $\tilde d$ on $(X,Y)$ such that $\Theta_X(\tilde d)=\tilde p$. By commutativity of  Diag.~(\ref{diag:comm-diag-meas-space-change-Theta}) the classical distribution $d=f^*(\tilde d)$ maps to $p$ under the map $\Theta_Z$ showing that $p$ is noncontextual.
}

Changing the outcome space also induces maps between the sets of distributions and a commutative diagram expressing the compatibility with the $\Theta$-map. A map $g:Y\to W$ of spaces induces:
\begin{itemize}
\item $g_*:\dDist(X,Y) \to \dDist(X,W)$ defined by sending $\delta^r$ to the deterministic distribution $\delta^{g\circ r}$.

\item $D_Rg_*: \clDist(X,Y) \to \clDist(X,W)$ that sends $d=\sum_r d(r)\, \delta^r$ to the classical distribution $\sum_r d(r)\, \delta^{g\circ r}$.

\item $g_*: \siDist(X,Y)\to \siDist(X,W)$ defined by sending $p:X\to D_RY$ to the simplicial distribution $X\xrightarrow{p} D_R Y \xrightarrow{D_Rg_*} D_RW$.
\end{itemize}

\noindent The corresponding commutative diagram in this case is given by
\begin{equation}\label{diag:comm-diag-out-space-change-Theta}
\begin{tikzcd}
\clDist(X,Y) \arrow[d,"\co{D_Rg_*}"'] \arrow[r,"\Theta_{X,Y}"] & \siDist(X,Y) \arrow[d,"g_*"] \\
\clDist(X,W) \arrow[r,"\Theta_{X,W}"] & \siDist(X,W)
\end{tikzcd}
\end{equation}
A result analogous to Proposition \ref{pro:ext-noncontextual} can be formulated: If a distribution $p$ on $(X,Y)$ is noncontextual then $g_*(p)$ on $(X,W)$ is noncontextual as well.  

\co{

\Pro{\label{pro:outcome-injective}
Assume that $g:Y\to W$ is an injective simplicial set map, i.e., each $g_n:Y_n\to W_n$ is injective for $n\geq 0$.
Then $p\in \siDist(X,Y)$ is noncontextual if and only if $g_*(p)\in \siDist(X,W)$ is noncontextual.
}
\Proof{
As observed above, the commutativity of Diag.~(\ref{diag:comm-diag-out-space-change-Theta}) implies that if $p$ is noncontextual then $g_*(p)$ is noncontextual. 
Conversely, assume that $g_*(p)$ is noncontextual. That is, there exists $d\in \clDist(X,W)$ such that $\Theta_{X,W}(d)=g_*(p)$. 
Note that both of the vertical maps in Diag.~(\ref{diag:comm-diag-out-space-change-Theta}), $g_*$ and $D_Rg_*$, are injective since $g$ is injective.
Let $S$ denote the set of outcome assignments $s:X\to W$ such that $d(s)\neq 0$. An outcome assignment $s:X\to W$ consists of a family of simplices $\set{s_\sigma}_\sigma$, where $\sigma$ runs over the generating simplices of $X$, compatible under the simplicial structure maps (Remark \ref{rem:desc-a-space-map}). Since $\Theta_{X,W}(d)=g_*(p)$ every $s\in S$ belongs to the support of $g_*(p)$.
Injectivity of $g_*$ implies that $\supp(g_*p)$ is contained in $\supp(p)$. Therefore $s_\sigma$ for every generating simplex $\sigma$ of $X$ belongs to $Y$. The simplices $\set{s_\sigma}_\sigma$ define a simplicial set map $\tilde s: X\to Y$ such that $s=g\circ \tilde s $.
By commutativity of Diag.~(\ref{diag:comm-diag-out-space-change-Theta}) the deterministic distribution   $\tilde d = \sum_{s\in S} d(s) \delta^{\tilde s}$ satisfies $\Theta_{X,Y}(\tilde d)=p$. Therefore $p$ is noncontextual.
}

}

\Ex{
{\rm
Consider the CHSH scenario $(T^\circ,N\ZZ_2)$ and the  inclusion $h:S^1\to N\ZZ_2$ of the outcome spaces given in 
(\ref{diag:h-map-S1-NZd}). 
Using Eq.~(\ref{eq:embed-dist-on-S1-to-dist-on-nerve}) we can represent simplicial probability distributions on $(T^\circ,S^1)$ as follows:
\begin{equation}\label{eq:S(T0,S1)}
\siDist(T^\circ,S^1) = \set{(p_{x_0}^1,p_{x_1}^1,p_{y_0}^1,p_{y_1}^1)\in [0,1]^4\,:\,  p_{x_i}^1+p_{y_j}^1\leq 1,\; i,j\in \ZZ_2}
\end{equation}
where $p_{x_i}=p_{\sigma_{x_i}}$ and $p_{y_j}=p_{\sigma_{y_j}}$. 
In this case the distribution $q=h_*(p)$ on $(T^\circ,N\ZZ_2)$ is given by
$$
\begin{aligned}
q_{y_0x_0} &=  (1-p_{x_0}^1-p_{y_0}^1,p_{x_0}^1,p_{y_0}^1,0)  \\
q_{y_0x_1} &= (1-p_{x_1}^1-p_{y_0}^1,p_{x_1}^1,p_{y_0}^1,0)   \\
q_{x_0y_1} &= (1-p_{x_0}^1-p_{y_1}^1,p_{y_1}^1,p_{x_0}^1,0)   \\
q_{x_1y_1} &= (1-p_{x_1}^1-p_{y_1}^1,p_{y_1}^1,p_{x_1}^1,0)   
\end{aligned}
$$
where we use Notation \ref{nota:DNZ2} and write $q_\sigma=(q^{00}_\sigma,q^{01}_\sigma,q^{10}_\sigma,q^{11}_\sigma)$ for a $2$-context $\sigma$. 
\co{By Proposition \ref{pro:outcome-injective}, $p$ is noncontextual if and only if $q$ is noncontextual.
} 
Now, using CHSH inequalities for $q$ we obtain the following set of inequalities
$$
p_{x_i}^1 + p_{y_j}^1 \leq 1 \quad i,j\in \ZZ_2,
$$
which are always satisfied by $p$ as a consequence of the description given in Eq.~(\ref{eq:S(T0,S1)}). Therefore every $p\in \siDist(T^\circ,S^1)$ is noncontextual and can be written as a probabilistic mixture of the following $7$ deterministic distributions: $\delta^{0000},\delta^{1000},\delta^{0100},\delta^{0010},\delta^{0001},\delta^{1010},\delta^{0101}$. 
}
}

\subsection{Gluing classical distributions}

When the space of measurements can be written as a union of smaller spaces deterministic distributions and simplicial distributions can be described as compatible distributions on the subspaces.  

\Pro{\label{pro:union-det-simp-distributions}
Suppose that the space of measurements is a union $X= A\cup B$ of two subspaces.
\begin{enumerate}
\item Distributions in  $\dDist(X,Y)$ can be identified with pairs $(r_A,r_B)$ of deterministic distributions on $(A,Y)$ and $(B,Y)$ satisfying $r_A|_{A\cap B}=r_B|_{A\cap B}$.

\item Distributions in $\siDist(X,Y)$ can be identified with pairs $(p_A,p_B)$ of simplicial distributions on $(A,Y)$ and $(B,Y)$ satisfying $p_A|_{A\cap B}=p_B|_{A\cap B}$.
\end{enumerate}
}
\Proof{
This is a consequence of the fact that a map $r:X\to Y$ (or $p:X\to D_RY$) of spaces is uniquely determined by the restrictions $r|_A$ and $r|_B$. Conversely, a pair $(r_A,r_B)$ of maps compatible on the intersection $A\cap B$ can be glued to give a map $r:X\to Y$. 
}

Note that an analogous observation does not apply to classical distributions. In general $\clDist(X,Y)$ is different from compatible pairs $(d_A,d_B)$ of classical distributions. 
However, we can associate a classical distribution to any such pair $(d_A,d_B)$ generalizing the constructions used in the diamond scenario 
(Example \ref{ex:diamond-noncontextual}), or
Fine's theorem \cite{fine1982joint}  \rev{including} its versions \cite{halliwell2019fine} (Example \ref{ex:two-copies-Delta-n}) \rev{and joint probability distributions constructed in \cite{flori2013compositories}}.

\Lem{
\label{lem:gluing-lemma}
\co{Let $R$ be a semifield.}
Suppose that the space of measurements is a union $X= A\cup B$. For $p\in \siDist(X,Y)$
 the following are equivalent:
\begin{enumerate}
\item The distribution $p$ is noncontextual, i.e. there exists $d\in \clDist(X,Y)$ such that $\Theta_X(d)=p$.
\item There exists distributions $d_A \in \clDist(A,Y)$,  $d_B\in \clDist(B,Y)$ such that $d_A|_{A\cap B} = d_B|_{A\cap B}$ in $\clDist(A\cap B,Y)$ and $\Theta_A(d_A)=p|_A$, $\Theta_B(d_B)=p|_B$.  
\end{enumerate}
}
\Proof{
Assuming (1) holds we can take $d_A=d|_A$ and $d_B=d|_B$. Then $d_A|_{A\cap B} = d|_{A\cap B} = d_B|_{A\cap B}$ in $\clDist(A\cap B,Y)$. Conversely, assume that (2) holds. Using the pair $(d_A,d_B)$ we can construct a distribution $d\in \clDist(X,Y)$. For an outcome map $r:X\to Y$ the distribution is defined as follows
\begin{equation}\label{eq:classical-ansatz}
d(r) = \frac{d_A(r|_A) d_B(r|_B)}{d_A|_{A\cap B}(r|_{A\cap B}) }.
\end{equation}
Note that the denominator is also equal to $d_B|_{A\cap B}(r|_{A\cap B})$ by the compatibility of the classical distributions $d_A$ and $d_B$.
First we verify that $d$ is a probability distribution on $\dDist(X,Y)$: Clearly $d(r)\geq 0$ and we have
$$
\begin{aligned}
\sum_r d(r) &= \sum_r \frac{d_A(r|_A) d_B(r|_B)}{d_A|_{A\cap B}(r|_{A\cap B})} \\
&= \sum_{r':A\to Y} \frac{d_A(r')}{d_A|_{A\cap B}(r'|_{A\cap B})} \sum_{t:t|_{A\cap B} = r'|_{A\cap B}} d_B(t) \\
&= \sum_{r':A\to Y} \frac{d_A(r')}{d_A|_{A\cap B}(r'|_{A\cap B})} d_B|_{A\cap B}(r'|_{A\cap B})\\
&= \sum_{r':A\to Y}  d_A(r') =1.
\end{aligned}
$$
In the second line we used Proposition \ref{pro:union-det-simp-distributions} to identify $r$ as a pair $(r',t)$, where $t:B\to Y$, of compatible outcome maps.
Next we verify that $\Theta_X(d)=p$. For this it suffices to show that $d|_A=d_A$ and $d|_B=d_B$. Note that in this case we have $\Theta_X(d)|_A = \Theta_A(d|_A)=\Theta_A(d_A)=p|_A$ and similarly for $B$. Since $p$ is uniquely specified by its restrictions $p|_A$ and $p|_B$ we conclude that $\Theta_X(d)=p$. To verify that $d|_A=d_A$ we compute
$$
\begin{aligned}
d|_A(r') &= \sum_{r:r|_A=r'} d(r) \\
&= \sum_{r:r|_A=r'} \frac{d_A(r|_A) d_B(r|_B)}{d_B|_{A\cap B}(r|_{A\cap B})} \\
&= \frac{d_A(r')}{d_B|_{A\cap B}(r'|_{A\cap B})}\sum_{r:r|_A=r'} d_B(r|_B) \\
&= \frac{d_A(r')}{d_B|_{A\cap B}(r'|_{A\cap B})}\sum_{t:t|_{A\cap B}=r'|_{A\cap B}} d_B(t) \\
& = d_A(r').
\end{aligned}
$$
In the fourth line we used Proposition \ref{pro:union-det-simp-distributions} to identify $\set{r\in D(X,Y)\,:\,r|_A=r'}$ with $\set{t\in D(B,Y)\,:\, t|_{A\cap B}=r'|_{A\cap B}}$.
Similarly one can show that $d|_B=d_B$.
}

\Cor{\label{cor:gluing-along-simplex}
\co{Let $R$ be a semifield.}
Suppose that $X=A\cup B$ with $A\cap B=\Delta^n$ for some $n\geq 0$. Then $p\in \siDist(X,Y)$ is noncontextual if and only if both $p|_A\in \siDist(A,Y)$ and $p|_B\in \siDist(B,Y)$ are noncontextual. 
}
\Proof{
When $A\cap B=\Delta^n$ we have $\clDist(\Delta^n,Y)=D_R(Y_n)=\siDist(\Delta^n,Y)$. Hence as a consequence $d_A|_{A\cap B}=d_B|_{A\cap B}$ whenever $\Theta_A(d_A)=p|_A$ and $\Theta_B(d_B)=p|_B$ by commutativity of  Diag.~(\ref{diag:comm-diag-meas-space-change-Theta}) for the measurement spaces $A$ and $B$. Therefore the result follows from Lemma \ref{lem:gluing-lemma}.
}

\Ex{[Generalized Fine ansatz]\label{ex:two-copies-Delta-n}
{\rm
Consider a measurement space $Z_{d,n}$ obtained by gluing 
two copies of a $d$-simplex along an $n$-dimensional face where $n\leq d$. 
Example \ref{ex:SingleTriangle-SimpScenario} and Corollary \ref{cor:gluing-along-simplex} imply that any simplicial distribution on $(Z_{d,n},Y)$  is noncontextual. This observation was proved in the special case of the diamond scenario in 
Example \ref{ex:diamond-noncontextual}.
Fig.~(\ref{fig:Two-delta3}) illustrates $Z_{3,2}$, in which case Eq.~(\ref{eq:classical-ansatz}) specializes to the ansatz used by Fine \cite{fine1982joint}. The case of $Z_{d,2}$ is also studied in \cite{halliwell2019fine}.
}}

\begin{figure}[h!]
\begin{center}
\includegraphics[width=.25\linewidth]{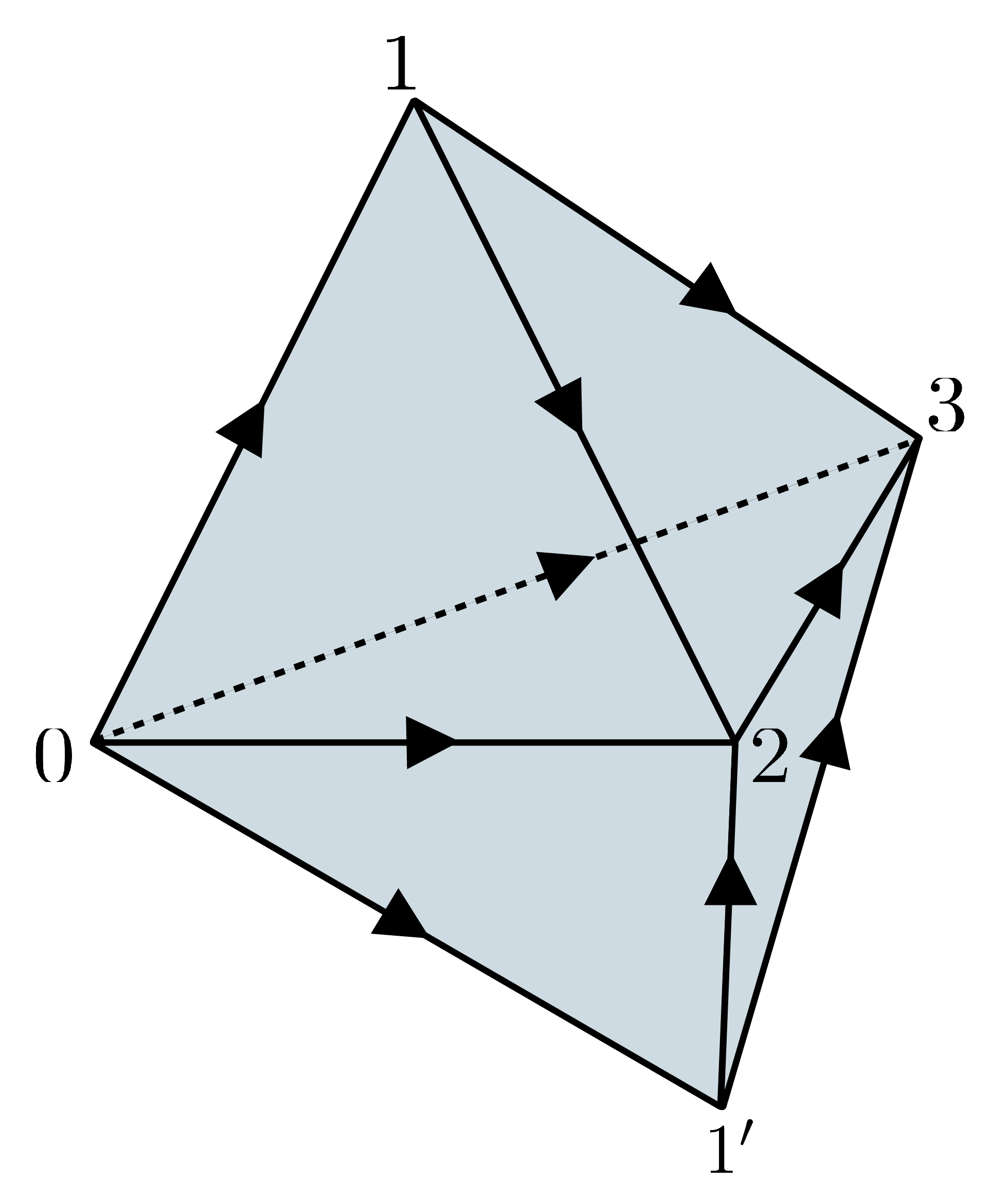} 
\end{center} 
\caption{Two copies of $\Delta^3$ are glued along their $d_1$ faces.  
}
\label{fig:Two-delta3}
\end{figure}

\subsection{Extensions of distributions}

In certain situations noncontextuality can be characterized as the possibility of extending the distribution to a larger measurement space. This technique will be very useful to extract the essence of Fine's argument for the sufficiency of the CHSH inequalities as a criteria for noncontextuality.

\Pro{
\label{pro:extending-dist}
Consider two simplicial scenarios $(Z,Y)$ and $(X,Y)$ with a map $f:Z\to X$ between the spaces of measurements. Assume that both $f^*:\clDist(X,Y)\to \clDist(Z,Y)$ and $\Theta_X: \clDist(X,Y)\to \siDist(X,Y)$ are surjective. Then $p\in \siDist(Z,Y)$ is noncontextual if and only if $p$ extends to a simplicial distribution on $(X,Y)$, i.e. there exists $\tilde p \in \siDist(X,Y)$ such that $f^*(\tilde p)=p$.
}
\begin{proof}
This follows from the   commutative Diag.~(\ref{diag:comm-diag-meas-space-change-Theta}).
Assume $p$ is a noncontextual distribution on $(Z,Y)$. Let $d$ be a classical distribution such that $\Theta_Z(d)=p$. By surjectivity of $f^*$ between the classical distributions  there exists $\tilde d$, a classical distribution on $(X,Y)$, such that $f^*(\tilde d)=d$. Commutativity of the diagram implies that $\tilde p =\Theta_X(\tilde d)$ is an extension of $p$. Converse implication follows from Proposition \ref{pro:ext-noncontextual}.
\end{proof}
 
Note that the surjectivity of the map between the classical distributions is satisfied  if $f^*:\dDist(X,Y)\to \dDist(Z,Y)$ is surjective.  
A typical application of Proposition \ref{pro:extending-dist} concerns the case where $f$ is a subspace inclusion.  

\Ex{[Extending from the boundary]\label{ex:boundary-delta-n}
{\rm 
Consider the inclusion of the boundary $f:\partial \Delta^n \to \Delta^n$ and take $Y=N\ZZ_d$ as the space of outcomes. 
By Proposition \ref{pro:outcome-map-nerve-edges}
we know that $r\in \dDist(\Delta^n,N\ZZ_d)$ is determined by the edges $r_{\sigma^{(i-1)i}}$. 
The $1$-outcomes  for the remaining $1$-contexts   can be obtained (inferred) from $a_i$'s as a consequence of the formula in Eq.~(\ref{eq:outcome-map-nerve-edges}). 
A similar analysis works for $\partial \Delta^n$  when $n\geq 3$. Though for $n=2$, note that $\partial \Delta^2$ consists of three edges and an outcome map $t:\partial \Delta^2\to N\ZZ_d$ is specified by a triple $(t_{\sigma^{01}}, t_{\sigma^{12}}, t_{\sigma^{02}})\in \ZZ_d^3$. Unlike $\Delta^2$ there is no compatibility relation imposed among the $1$-outcomes, i.e. we do not require $t_{\sigma^{02}} = t_{\sigma^{01}}+t_{\sigma^{12}}$.
But for $n\geq 3$ the subspace $\partial\Delta^n$ has the same set of $2$-contexts as  $\Delta^n$ and an outcome map on the boundary is still determined by the $1$-outcomes assigned to $\sigma^{(i-1)i}$.
Therefore in this case $f^*:\dDist(\Delta^n,N\ZZ_d)\to \dDist(\partial \Delta^n,N\ZZ_d)$ is a bijection, and Proposition \ref{pro:extending-dist} implies that a distribution on $(\partial \Delta^n,N\ZZ_d)$ is noncontextual if and only if it extends to $(\Delta^n,N\ZZ_d)$. In Proposition \ref{pro:boundary-delta-3} we show that for $n=3$ and $d=2$ extension always exists, hence any distribution on $(\partial \Delta^3,N\ZZ_2)$ is noncontextual. 
}} 



\subsection{Diamond scenarios and CHSH inequalities}

For the rest of this section we restrict to simplicial probability distributions, that is we take $R=\nnegR$.
Let $Z$ denote the diamond scenario obtained by gluing two copies of $\Delta^2$ along a common edge. Here we don't make any restrictions on the choice of the face on each triangle. Our goal is to study distribution on the boundary $\partial Z$ and determine when an extension to the whole space exists. For this we will decompose the boundary into simpler spaces.
Let $\Lambda^n_k$ for $0\leq k\leq n$ denote the subspace of $\Delta^n$ generated by all the faces of the form $d_i\sigma^{01\cdots n}$ where $i\neq k$. This space is called the {\it $n$-dimensional $k$-th horn}.
For a distribution $p$ on $(\Delta^n,Y)$ we will write $p_{a_0a_1\cdots a_k}=p_{\sigma^{a_0a_1\cdots a_k}}$ for the distribution assigned to the simplex $\sigma^{a_0a_1\cdots a_k}$.
In dimension $n=2$ a horn is obtained by omitting one of the edges in the boundary of a $2$-simplex. We will first describe  when a simplicial  distribution on $(\Lambda^2_1,N\ZZ_2)$ extends to $(\Delta^2,N\ZZ_2)$ along the  inclusion map $\Lambda^2_1 \to \Delta^2$. This is equivalent to understanding when such a distribution is noncontextual as a consequence of Proposition \ref{pro:extending-dist}. Let $p$ be a distribution on $(\Lambda^2_1,N\ZZ_2)$  specified by the tuple $(p_{01}^0,p_{12}^0)\in [0,1]^2$.   If $p$ extends to a distribution $\tilde p$ on $\Delta^2$ then we require  
$$
\tilde p^{ab} = \frac{1}{2}(p_{01}^a+p_{12}^b-p_{02}^{a+b+1}) \geq 0\;\;\text{ for all }a,b\in \ZZ_2
$$
for some distribution $p_{02}$ on the $d_1$-face. Note that $p_{ij}^1=1-p_{ij}^0$ for $0\leq i<j\leq 2$, so that we can express the inequalities only in terms of $p_{ij}^0$. A suitable $\tilde{p}^{ab}$ can be found if and only if $p_{02}^{0}$ satisfies
\begin{eqnarray}
|p_{01}^0+p_{12}^0-1| \leq  p_{02}^{0}\leq 1-|p_{01}^0-p_{12}^0|.
\end{eqnarray}
Applying Fourier--Motzkin elimination to $p_{02}^0$ we see that $p$ extends to $\Delta^2$ if and only if 
\begin{equation}\label{eq:horn-ext-triangle}
|p_{01}^0+p_{12}^0-1| \leq 1-|p_{01}^0-p_{12}^0|.
\end{equation}
Notice that such a $p_{02}^{0}$ can always be found for any $p_{01}^{0}$ and $p_{12}^{0}$ since expanding the absolute value in Eq.~(\ref{eq:horn-ext-triangle}) and canceling like terms yields the ``trivial" inequalities $0\leq p_{01}^{0}, p_{12}^{0}\leq1$. Similar inequalities exist for the other horn inclusions $\Lambda^2_{k}\to \Delta^2$ for $k=0,2$. We will use this analysis to study extensions from the boundary of the diamond scenario. 

\Pro{
\label{pro:diamond-CHSH}
A distribution $p$ on the boundary of the diamond scenario extends to the diamond if and only if $p$ satisfies the CHSH inequalities in  Eq.~(\ref{eq:CHSH234}).
}
\Proof{For $Z$ we will assume that it is obtained by gluing along the $d_1$ faces. The argument for the other choices is similar.
The distribution $p$ on the boundary of the diamond is specified by $(p_{01}^0,p_{12}^0,q_{01}^0,q_{12}^0)\in [0,1]^4$. The distribution $p$ extends to a distribution on the diamond if and only if two copies of the inequalities, one for $p_{ij}$'s and one for $q_{i'j'}$, have a common solution for $p_{02}^{0} = q_{02}^{0}$. Generalizing (\ref{eq:horn-ext-triangle}), this occurs precisely when 
\begin{eqnarray}
\max\set{|p_{01}^0+p_{12}^0-1|,|q_{01}^0+q_{12}^0-1|} \leq  p_{02}^{0} \leq \min\set{1-|p_{01}^0-p_{12}^0|,1-|q_{01}^0-q_{12}^0|}.\notag
\end{eqnarray}
By Fourier-Motzkin elimination this single inequality is equivalent to the following four 
$$
\begin{aligned}
|p^0_{01}+p^0_{12}-1| &\leq 1-|p^0_{01}-p^0_{12}|\\
|p^0_{01}+p^0_{12}-1| &\leq 1-|q^0_{01}-q^0_{12}|\\
|q^0_{01}+q^0_{12}-1| &\leq 1-|p^0_{01}-p^0_{12}|\\
|q^0_{01}+q^0_{12}-1| &\leq 1-|q^0_{01}-q^0_{12}|,
\end{aligned}
$$
in addition to the trivial inequalities corresponding to $0\leq p_{ij}^{0}, q_{ij}^{0}\leq 1$. Expanding the absolute values gives the inequalities
\begin{equation}\label{eq:CHSH-prob}
\begin{aligned}
0&\leq  p_{01}^0 + p_{12}^0 + q_{01}^0 - q_{12}^0 \leq 2 \\
0&\leq  p_{01}^0 + p_{12}^0 - q_{01}^0 + q_{12}^0 \leq 2 \\
0&\leq  p_{01}^0 - p_{12}^0 + q_{01}^0 + q_{12}^0 \leq 2 \\
0&\leq  -p_{01}^0 + p_{12}^0 + q_{01}^0 + q_{12}^0 \leq 2. 
\end{aligned}
\end{equation}
These equations are formally identical to the CHSH inequalities appearing in Eq.(\ref{eq:CHSH234}).
}%


Next we analyze the situation for two copies of diamonds glued along a triangle. This produces a space  as in Fig.~(\ref{fig:twodiamonds}). Again orientation of the edges can be different, our choice is the one that will be used later in the topological proof of Fine's theorem.

\begin{figure}[h!]
\begin{center}
\includegraphics[width=.2\linewidth]{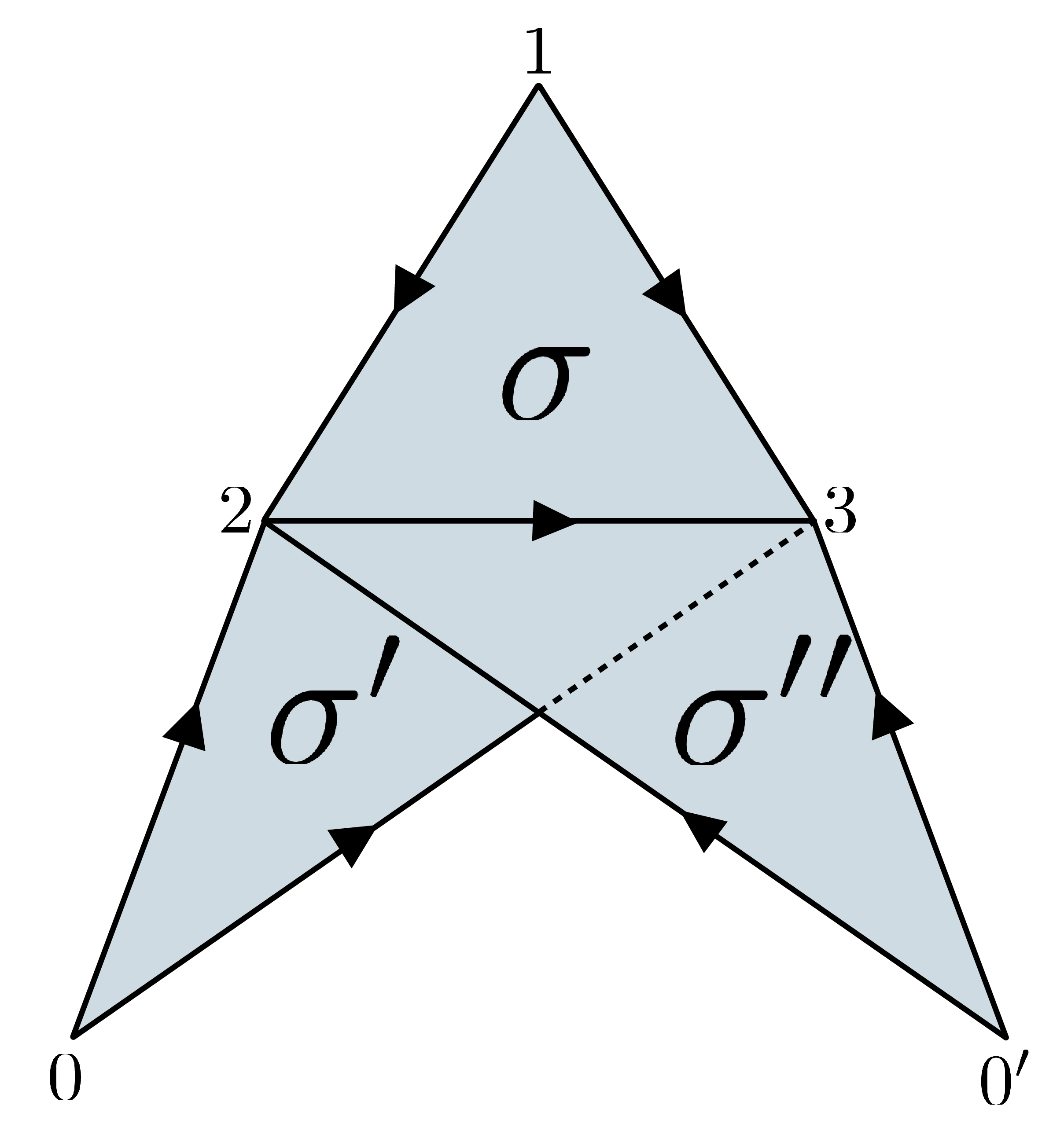}
\end{center} 
\caption{This space is obtained by gluing two diamonds, denoted by $Z(\sigma\sigma')$ and $Z(\sigma\sigma'')$, along a copy of $\Delta^2$ with vertices $(1,2,3)$. The third diamond $Z(\sigma'\sigma'')$ consists of the other two triangles.   
}
\label{fig:twodiamonds}
\end{figure}  

\Pro{\label{pro:2diamonds-3CHSH}
Let $A$ denote the space (as in Fig.~(\ref{fig:twodiamonds})) obtained by gluing two copies of $Z$ along $\Delta^2$. 
A simplicial probability distribution on $(\partial A,N\ZZ_2)$ extends to $(A,N\ZZ_2)$ if and only if the distributions on the boundary of the three diamonds $Z(\sigma\sigma')$, $Z(\sigma\sigma'')$ and $Z(\sigma'\sigma'')$ satisfy the CHSH inequalities.
}
\Proof{
Let $p$ denote a simplicial probability distribution on  $(\partial A,N\ZZ_2)$.
Let us write $p_{ij}$ for a distribution on an edge of $A$ between the vertices $(i,j)$ where $i<j$. The distribution on the boundary of $A$ specifies the values of $p_{ij}^0$ for $(i,j)\neq (2,3)$. Then $p$ extends to $A$ if and only if  
$$
\max\set{|p_{12}^0+p_{13}^0-1|,|p_{02}^0+p_{03}^0-1|,|p_{0'2}^0+p_{0'3}^0-1|} \leq 
\min\set{1-|p_{12}^0-p_{13}^0|,1-|p_{02}^0-p_{03}^0|,1-|p_{0'2}^0-p_{0'3}^0|}. 
$$
This inequality is obtained by applying the extension criterion in the proof of Proposition \ref{pro:diamond-CHSH} to the two diamonds $Z(\sigma\sigma')$ and $Z(\sigma\sigma'')$, and then eliminating $p_{23}^0$ by Fourier--Motzkin elimination.
This inequality is equivalent to the set of CHSH inequalities for the three diamonds. 
}

\subsection{Topological proof of Fine's theorem}
\label{sec:TopProofFine}

Let $H$ denote  the space of measurements obtained by gluing two copies of $\partial\Delta^3$ along a common face given by a triangle. For concreteness we will assume they are glued along their $d_0$-faces as in Fig~(\ref{fig:S-H}).
The subspace generated by the faces $\sigma^{013}$ and $\sigma^{012}$ of the first $\partial \Delta^{3}$ is a copy of the diamond scenario $Z$. Similarly the subspace generated by $\sigma^{0'13}$ and $ \sigma^{0'12}$ is another copy of the diamond scenario $Z$. The intersection of these two copies is a horn $\Lambda^2_0$ generated by the edges $d_0 \sigma^{013}$ and $d_0\sigma^{012}$. Let $Q$ denote the subspace given by the union of the two diamonds; see Fig~(\ref{fig:S-H}). We will refer to $(Q,N\ZZ_2)$ as the {\it square scenario}. We begin by showing that any distribution on $\partial \Delta^3$ is noncontextual.

\Pro{\label{pro:boundary-delta-3}
Any simplicial probability  distribution on $(\partial \Delta^3,N\ZZ_2)$ is noncontextual. 
}
\Proof{
By Example \ref{ex:boundary-delta-n} it suffices to show that any distribution $p$ on $(\partial \Delta^3,N\ZZ_2)$ extends to a distribution on $(\Delta^3,N\ZZ_2)$. The distribution $p$ is determined by the marginals $p_{ijk}=p_{\sigma^{ijk}}$ on the generating simplices $\sigma^{ijk}$ where $0\leq i<j<k\leq 3$. Extending $p$ to $\Delta^3$ amounts to specifying a distribution $p_{0123}=p_{\sigma^{0123}}$ on the generating simplex $\sigma^{0123}$. This distribution is defined as follows: Set $p_{0123}^{000}$ to be the minimum of $\set{ p_{ijk}^{00}\,:\, 0\leq i<j<k\leq 3 }$ and determine the remaining values $p_{0123}^{abc}$ from the marginals. To see how this works assume that the minimum is $p_{123}^{00}$. The cases where the minimum is $p_{023}^{00}$, $p_{013}^{00}$ and $p_{012}^{00}$ are dealt with similarly. The marginals for the $d_i$-faces give us
\begin{equation}\label{eq:mar-faces-delta3}
\begin{aligned}
p_{0123}^{100} &= p_{123}^{00}- p_{0123}^{000} \\ 
p_{0123}^{110} &= p_{023}^{00}-p_{0123}^{000} \\
p_{0123}^{011} &= p_{013}^{00}-p_{0123}^{000} \\
p_{0123}^{001} &= p_{012}^{00}-p_{0123}^{000} 
\end{aligned}
\end{equation}
which are  all nonnegative because of our choice of $p_{0123}^{000}$.
The remaining components of $p_{0123}$ can be obtained from first marginalizing to the $d_0$-th face and then further marginalizing to the edges on the boundary of this face, which also turn out to be nonnegative:
$$
\begin{aligned}
p_{0123}^{101} &= p_{123}^{01}- p_{0123}^{001}= p_{12}^0- p_{123}^{00} - p_{012}^{00}+p_{0123}^{000} = p_{12}^0 - p_{012}^{00} \geq 0\\
p_{0123}^{010} &=   p_{123}^{10}- p_{0123}^{110} = p_{23}^0- p_{123}^{00} -p_{023}^{00}+p_{0123}^{000} =p_{23}^0-p_{023}^{00} \geq 0 \\
p_{0123}^{111} &= p_{123}^{11}- p_{0123}^{011} = p_{13}^0- p_{123}^{00}-p_{013}^{00}+p_{0123}^{000} = p_{13}^0-p_{013}^{00}\geq 0.
\end{aligned}
$$ 
It remains to show that marginalizing to the $d_i$-th face produces the original distribution. This is clear for $i=0$ by construction, and can be shown for the remaining faces by direct verification. The distribution is also normalized as can be shown by summing up the elements in $p_{0123}$ given above.
}

\begin{figure}[h!]
\begin{center}
\includegraphics[width=.6\linewidth]{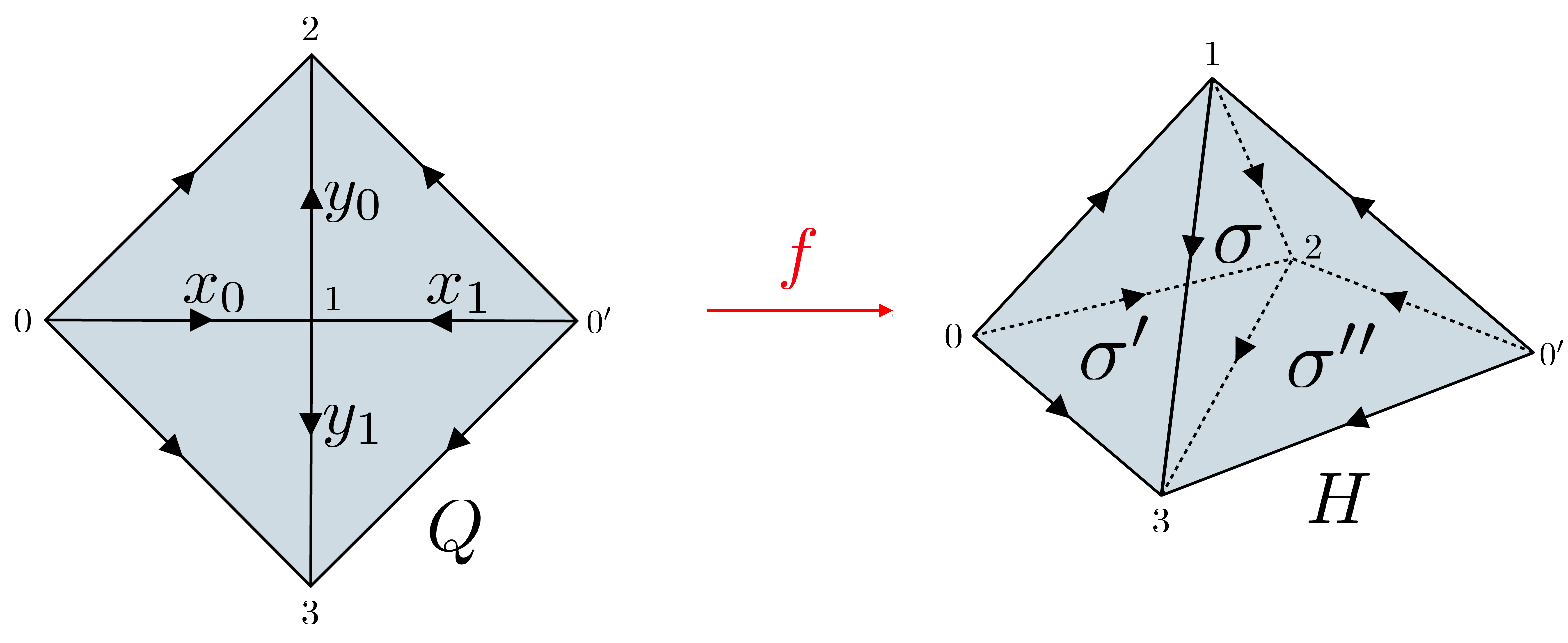} 
\end{center} 
\caption{$H$ is obtained from $Q$ by adding the triangles with vertices $\sigma=(1,2,3)$, $\sigma'=(0,2,3)$ and $\sigma''=(0',2,3)$.  
}
\label{fig:S-H}
\end{figure}

\Thm{[Fine's theorem]
\label{thm:Fines-theorem}
Let $Q$ denote the measurement space of the square scenario and $p$ be a simplicial probability distribution on $(Q,N\ZZ_2)$. Then  $p|_{\partial Q}$ satisfies the CHSH inequalities if and only if $p$ is noncontextual.
}
\Proof{ 
By  Corollary \ref{cor:gluing-along-simplex} and Proposition \ref{pro:boundary-delta-3}  any simplicial probability distribution on $(H,N\ZZ_2)$ is noncontextual. The map $f^*:\clDist(H,N\ZZ_2)\to \clDist(Q,N\ZZ_2)$ between the classical distributions is a bijection since both sets are given by $\ZZ_2^4$ and an outcome assignment is determined by the outcomes assigned to the measurements $x_i$'s and $y_j$'s (Proposition \ref{pro:outcome-map-nerve-edges}).
By Proposition \ref{pro:extending-dist} $p$ is noncontextual if and only if it extents to $H$.
Observe that $H$ is the union of $Q$ and $A$ (the space in Fig.~(\ref{fig:twodiamonds})) along the boundary $\partial A$. Now by Proposition \ref{pro:2diamonds-3CHSH} $p|_{\partial A}$ extends to $A$ if and only if the three CHSH inequalities are satisfied. Two of those corresponding to $Z(\sigma\sigma')$ and $Z(\sigma\sigma'')$ are satisfied because these diamonds   also lie on the boundary of the diamonds with vertices $(0,1,2,3)$ and $(0',1,2,3)$ in $Q$, and Proposition \ref{pro:diamond-CHSH} implies that at their boundary the CHSH inequalities are satisfied. The third one corresponds precisely to the CHSH inequalities for $p|_{\partial Q}$. 
}

We turn back to  the CHSH scenario where the measurement space is represented as a punctured torus $T^\circ$ as in Example \ref{ex:Bell222}; see Fig.~(\ref{fig:Bell-punctured-torus}).

\Cor{\label{cor:extension-torus}
A simplicial probability distribution  $p$ on $(T^\circ,N\ZZ_2)$ extends to $(T,N\ZZ_2)$ if and only if $p$ is noncontextual.
}
\Proof{
By Proposition \ref{pro:diamond-CHSH} the distribution $p$ extends to the torus, i.e. to the diamond in the middle, if and only if   $p|_{\partial T^\circ}$ satisfies the CHSH inequalities. Fine's theorem implies that this is equivalent to $p$ being noncontextual. 
} 

For example, Corollary \ref{cor:extension-torus} implies that the PR box in  Fig.~(\ref{fig:Bell-punctured-torus-ext}) is contextual since it fails to extend to the torus.   
\co{
An alternative proof of Fine's theorem that uses the punctured torus realization (Fig.~\ref{fig:Bell-punctured-torus}) is given in \cite{okay2022mermin}.
}

\section{Strong contextuality and cohomology}
\label{sec:StrongContCoho}

In this section we show that strong contextuality can be detected by cohomology. This is the first step towards verifying our claim that the simplicial formalism generalizes the earlier topological approach. In fact, significant control is achieved by constructing our cohomological classes for general simplicial distributions, as opposed to cohomology classes constructed in \cite{Coho} that come essentially from algebraic relations among quantum observables. 
In the next section we apply our constructions to state-independent contextuality making the link to the topological approach.





\subsection{Strong contextuality}

There is a stronger version of contextuality whose definition relies on the notion of support.
The {\it support of a simplicial distribution} $p:X\to D_RY$  is defined by
\begin{equation}\label{eq:support-simp-dist}
\supp(p)=\set{s:X\to Y\,:\, p_{\rev{\sigma}}(s_{\rev{\sigma}})>0\;\forall {\rev{\sigma}}\in X_n,\;n\geq 0}.
\end{equation}

\Pro{\label{pro:StrongImpliesContextual}
If the support of $p:X\to D_RY$ is empty then $p$ is contextual 
}
\Proof{
Assume that $p$ is noncontextual, that is, it can be written as a probabilistic mixture $p=\sum_r \lambda(r) \, \delta^r$. Then for $s$ such that $\lambda(s)\neq 0$ we have
$$
p_{\rev{\sigma}}(s_{\rev{\sigma}}) = \sum_r \lambda(r) \, \delta^{r_{\rev{\sigma}}}(s_{\rev{\sigma}}) = \lambda(s) >0.  
$$
Therefore $s$ belongs to the support of $p$.
}

\Def{\label{def:StronglyContextual}
{\rm
A simplicial distribution $p$ on   $(X,Y)$ is called {\it strongly contextual} if its  support $\supp(p)$ is empty.
}}

This notion generalizes strong contextuality for nonsignaling distributions; see Corollary \ref{cor:strong-comparison-sheaf}. For example, the PR box described in  Fig.~(\ref{fig:Bell-punctured-torus-ext}) is strongly contextual.

\subsection{Deterministic distributions and cohomology}\label{sec:det-dist-nerve}

In this section we relate the set $\dDist(X,N\ZZ_d)$ of deterministic distributions to the first cohomology group of $X$.
Let us introduce the cohomology group of a simplicial set. We omit the discussion of homology groups as they are not used in this paper. We will  restrict to $\ZZ_d$ as the coefficients of our cohomology groups.
Given a space $X$ we can construct a {\it cochain complex} 
$$
C^*(X):\;\;  C^0(X) \xrightarrow{\delta_0} C^1(X) \xrightarrow{\delta_1} C^2(X) \xrightarrow{\delta_2}  \cdots \to C^n(X) \xrightarrow{\delta_n} C^{n+1}(X) \to \cdots
$$ 
where
\begin{itemize}
\item $C^n(X)$ consists of functions $f:X_n\to \ZZ_d$ such that $f(\sigma)=0$ whenever $\sigma$ is a degenerate $n$-simplex\footnote{The conventional way to introduce the cochain complex of a simplicial set is to consider all functions $X_n\to \ZZ_d$, not just those that vanish on degenerate simplices, i.e. those simplices that lie in the image of a degeneracy map. However, it is well-known that ignoring degenerate simplices does not affect homology \cite[Chapter III]{goerss2009simplicial}.  In the dual fashion considering functions that vanish on degenerate simplices suffices for the computation of cohomology. Our choice simplifies the exposition by avoiding relative complexes in the cohomology long exact sequence.
},
\item $\delta_n: C^n(X)\to C^{n+1}(X)$ is given by $\delta_nf(\sigma)=\sum_{i=0}^{n+1} (-1)^i f(d_i\sigma)$ for all $\sigma\in X_{n+1}$.
\end{itemize}  
 The {\it $n$-th cohomology group} of $X$ is defined to be the quotient group given by the kernel of $\delta_n$ modulo the image of $\delta_{n-1}$:
$$
H^n(X) = \frac{\ker(\delta_n)}{\im(\delta_{n-1})}.
$$ 
A map $g:Z\to X$ of spaces induces a homomorphism $g^*:H^n(X)\to H^n(Z)$ between the cohomology groups. Under this map a cohomology class $[f]$ is sent to   $g^*([f])=[f\circ g_n]$. 

By Proposition \ref{pro:outcome-map-nerve-edges} (and Remark \ref{rem:f-vanishes-on-degenerate-simplices}) sending an outcome assignment $r:X\to N\ZZ_d$ to the $1$-cochain $r_1:X_1\to \ZZ_d$ induces a function 
\begin{equation}\label{diag:alphaX}
\alpha_X:\dDist(X,N\ZZ_d) \to  H^1(X).
\end{equation}

\Cor{\label{cor:reduced-deterministic}
Let $X$ be a space of measurements with a single $0$-context. Then the function $\alpha_X$ given in 
(\ref{diag:alphaX}) is a bijection.
}
\Proof{
When $X_0=\set{\star}$ the first cohomology group is given by the kernel of $\delta_1:C^1(X)\to C^2(X)$ since  $\delta_0$ is the zero map. In this case $H^1(X)$ consists of those functions $f$ satisfying the condition in Proposition \ref{pro:outcome-map-nerve-edges}. Therefore deterministic distributions are given by the elements of $H^1(X)$.  
}

\Ex{\label{ex:H1-circle}
{\rm
An example to a measurement space with a single $0$-context is the circle $S^1$.
The first cohomology group of $S^1$ is computed using the cochain complex $C^0\to C^1\to C^2$, which is given by
$$
\ZZ_d \xrightarrow{0} \ZZ_d^{\set{0,1}} \xrightarrow{\delta_1} \ZZ_d^{\set{00,01,10}}.
$$
We use the simplices given in Eq.~(\ref{eq:circle-alt-rep}) and write  $ab$ instead of $(a,b)$ for simplicity. Given $f:\set{0,1}\to \ZZ_d$ we have $\delta_1f(ab)=f(b)-f(a+b)+f(a)$. Thus $\delta_1f(ab)=f(0)$ for $ab=00,01,10$ and $f$ is in the kernel if and only if $f(0)=0$. This implies that 
$$
H^1(S^1) = \set{f\in \ZZ_d^{\set{0,1}}\,:\, f(0)=0} =\ZZ_d
$$ 
since such a function is determined by $f(1)\in \ZZ_d$. Therefore 
$
\dDist(S^1,N\ZZ_d)  \cong H^1(S)=\ZZ_d
$.
}}


\subsection{Cohomological witness for contextuality} \label{sec:CohoWit}

Let $X$ be a space of measurements and $Z$ be a subspace. One can construct the quotient space $\bar X=X/Z$ where the subspace $Z$ is identified to a point. More formally, the  $n$-simplices of $\bar X$ are given by the difference set $X_n-Z_n$ union with $\set{\star_n}$ representing the collapsed simplices coming from $Z_n$. The simplicial structure is given as follows:  
\begin{itemize}
\item For $\sigma \in X_n-Z_n$ the face and the degeneracy maps act as those of $X_n$, except that $d_i\sigma = \star_{n-1}$ if $d_i\sigma\in Z_{n-1}$. 
\item $d_i\star_n=\star_{n-1}$ and $s_j\star_n = \star_{n+1}$.
\end{itemize}
We will construct a cohomological witness on the quotient space $\bar X$ that detects contextuality. 
Let us write $i:Z\to X$ for the inclusion and $q:X\to \bar X$ for the quotient map.
The main tool is the following exact sequence\footnote{A sequence of homomorphisms $A\xrightarrow{f} B\xrightarrow{g} C$ is said to be {\it exact} if the image of $f$ is equal to the kernel of $g$.} in cohomology
\begin{equation}\label{diag:connecting-exact}
H^1(\bar X) \to H^1(X) \to H^1(Z) \xrightarrow{\zeta} H^2(\bar X).
\end{equation}
Given $[f]\in H^1(Z)$ the cohomology class $\zeta[f]$ is defined as follows:
\begin{itemize}
\item Lift $f$ to a cochain $\tilde f\in C^1(X)$ by setting 
$$
\tilde f(\sigma) = \left\lbrace
\begin{array}{ll}
f(\sigma) &  \sigma\in Z_1 \\
0 & \text{otherwise.}
\end{array}
\right.
$$  
\item Compute $\delta_1\tilde f: X_2 \to \ZZ_d$, a $2$-cochain in $C^2(X)$: 
$$
\delta_1\tilde f(\sigma)  = \tilde f(d_0 \sigma) - \tilde f(d_1\sigma) + \tilde  f(d_2\sigma).
$$ 
\item Using $\delta_1 \tilde f$ construct a cochain $\bar f$ in $C^2(\bar X)$ by setting
\begin{equation}\label{eq:bar-f}
\bar f(\sigma)  = \left\lbrace
\begin{array}{ll}
\delta_1\tilde f(\sigma) &  \sigma\in X_2-Z_2 \\
0 & \sigma = \star_2.
\end{array}
\right.
\end{equation}
\end{itemize}
The connecting homomorphism $\zeta$ sends $[f]$ to the cohomology class $[\bar f]$.
Example \ref{ex:coho-wit-Mermin-state-dep} below  explains how this is done in practice. Exactness of 
(\ref{diag:connecting-exact}) is a standard fact in algebraic topology; see \cite[\S 1.3]{weibel1995introduction}. Next result will be useful in Section \ref{sec:state-indep-contex}.

\Lem{\label{lem:split-connecting-zero}
If the inclusion map $i:Z\to X$ splits, i.e. there exists a map $j:X\to Z$ of spaces such that $j\circ i:Z\to Z$ is the identity map, then the connecting homomorphism $\zeta$ is the zero map.
}
\Proof{The map $j$ can be used to show that  $i^*:H^1(X)\to H^1(Z)$ is surjective so that by exactness of 
(\ref{diag:connecting-exact}) the connecting homomorphism is zero. To see the surjectivity observe that the composite $H^1(Z) \xrightarrow{j^*} H^1(X) \xrightarrow{i^*} H^1(Z)$ is the identity map.
}

Using the map in 
(\ref{diag:alphaX}) we  
will compare the sequence in 
(\ref{diag:connecting-exact}) to the sequence of deterministic distributions to obtain the following commutative diagram
\begin{equation}\label{diag:diag-dist-coho}
\begin{tikzcd}
 D(\bar X,N\ZZ_d) \arrow{r}{q^*} \arrow{d}{\alpha_{\bar X}} &  D(X,N\ZZ_d) \arrow{r}{i^*} \arrow{d}{\alpha_{ X}} &  D(Z,N\ZZ_d) \arrow{d}{\alpha_{Z}} &  \\
H^1(\bar X)  \arrow{r}{q^*} & H^1(X)   \arrow{r}{i^*} & H^1(Z)   \arrow{r}{\zeta} & H^2(\bar X)
\end{tikzcd}
\end{equation}

\Pro{\label{pro:det-ext-iff-alpha-delta-zero}
A deterministic distribution $r\in \dDist(Z,N\ZZ_d)$ extends to $X$, i.e. there exists $s \in \dDist(X,N\ZZ_d)$ such that $i^*(s)=r$, if and only if $\zeta\circ \alpha_Z(r)=0$ in $H^2(\bar X)$.
}
\Proof{
If an extension  exists then by the commutativity of Diag.~(\ref{diag:diag-dist-coho}) and the exactness of the bottom part of that diagram we have
$$
\zeta \circ \alpha_Z (r) = \zeta \circ\alpha_Z \circ i^*(s) = \zeta \circ i^* \circ\alpha_X (s) =0.
$$
Conversely, assume that $\zeta \circ \alpha_Z (r)=0$. 
Recall from 
(\ref{diag:alphaX}) that $\alpha_Z$   sends a deterministic distribution represented by $r:Z\to N\ZZ_d$ to the function $r_1:Z\to \ZZ_d$, the restriction onto the set of $1$-contexts
(similarly $\alpha_X$ and $\alpha_{\bar X}$).
Now we apply the construction of the connecting homomorphism described above to $f=r_1$ to obtain the cochains $ \tilde r_1 \in C^1(X)$ and $\bar r_1 \in C^2(\bar X)$; see Eq.~(\ref{eq:bar-f}), so that $[\bar r_1]=\zeta[r_1]=0$. 
Therefore there exists $t\in C^1(\bar X)$ such that $\delta_1(t)=\bar r_1$.
We claim that the function $s_1=\tilde r_1-q^*(t)$, which corresponds to a deterministic distribution $s$ by Proposition \ref{pro:outcome-map-nerve-edges}, gives the desired extension of $r$. For this we need to verify that $s_1$ satisfies Eq.~(\ref{eq:s-boundary}) and $i^*(s_1)=r_1$. First one follows from  $\delta_1(s_1)=\delta_1(\tilde r_1 - q^*(t))=
\delta_1\tilde r_1 - q^*(\delta_1t)=\delta_1\tilde r_1 - q^*(\bar r_1)=0$. For the second one we have $i^*(s_1)=i^*(\tilde r_1)-i^*(q^*(t)) = i^*(\bar r_1) = r_1$ since $i^*(q^*(t)=0$. 
}

Given $p\in \siDist(X,N\ZZ_d)$ 
and the subspace $Z$ we define a set of cohomology classes:
\begin{equation}\label{eq:coho-classes}
\cl_Z(p) =  \zeta\circ \alpha_Z (\supp(p|_Z)) \subset H^2(\bar X).
\end{equation}

\Cor{\label{cor:Coho-StrongContex}
Let $p\in \siDist(X,N\ZZ_d)$ and $Z\subset X$ be a subspace. If $\cl_Z(p)$ does not contain the zero class, i.e. $0\notin \cl_Z(p)$, then $p$ is strongly contextual.
}
\Proof{
Proposition \ref{pro:det-ext-iff-alpha-delta-zero} implies that none of the elements in $\supp(p|_Z)$ extends to $X$
since 
$0\notin \cl_Z(p)$. Therefore the support $\supp(p)$ is empty. 
}

A typical application of this result is the case when $p|_Z$ is a deterministic distribution so that $\cl_Z(p)$ consists of a single cohomology class $[\beta_p] \in H^2(\bar X)$. Then $[\beta_p]$ serves as a {\it witness
for strong  contextuality}, in the sense that $[\beta_p]\neq 0$ implies strong contextuality for $p$.  

\rev{ 
Our construction of the cohomology witnesses can be compared to the \v Cech cohomological construction of \cite{abramsky2015contextuality}. 
Therein the construction is with respect to one of the contexts in the measurement cover, whereas in our construction $Z$ can be any subspace; see  Example \ref{ex:coho-wit-Mermin-state-dep}.
Another difference is that 
\co{our construction relies on the abelian group structure of the set of outcomes.}
}

\begin{figure}[h!]
\begin{center}
\includegraphics[width=.8\linewidth]{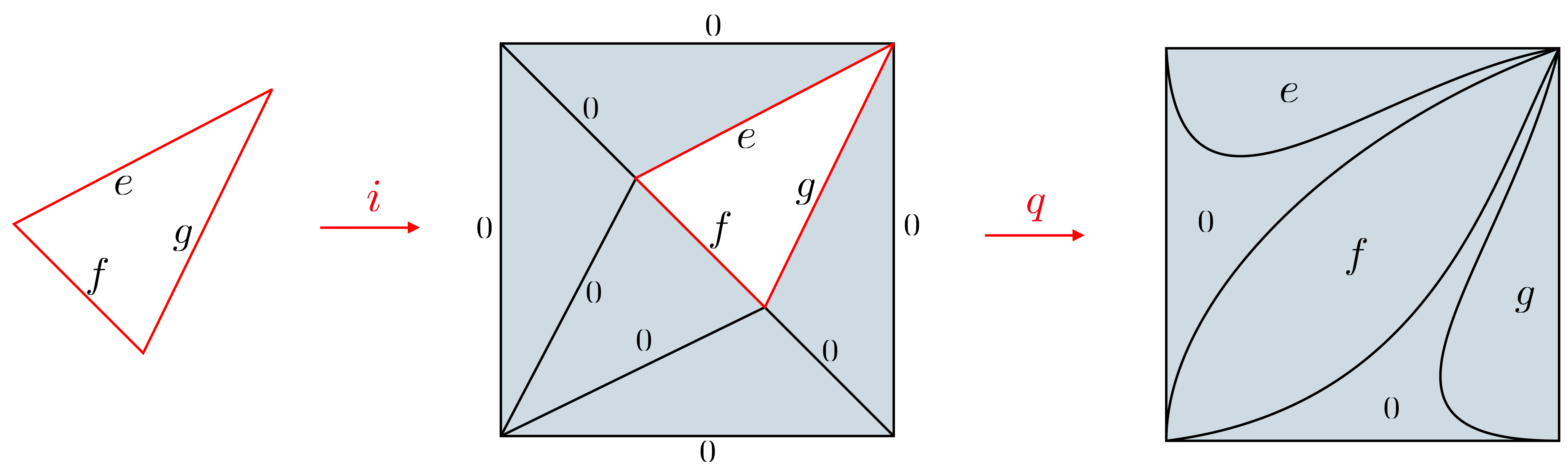} 
\end{center} 
\caption{The $1$-cochain specified by $(e,f,g)$ on the boundary is first extended to a $1$-cochain on the whole space by assigning zeros to the remaining edges. The coboundary $\delta_1$ turns this $1$-cochain to a $2$-cochain that assigns to each triangle the sum of the numbers on the edges of its boundary. This produces a $2$-cochain on the torus where the values assigned to each $2$-simplex are as indicated in the figure.  
}
\label{fig:coho-witness}
\end{figure} 

\Ex{\label{ex:coho-wit-Mermin-state-dep}
{\rm
Consider the state-dependent version of Mermin square discussed in Example \ref{ex:MerminStateDep}. The torus minus a single triangle is our space $X$ of measurements. The subspace $Z$ will be taken as the boundary of $X$ consisting of the three edges. Let  $p$ be a simplicial distribution on $(X,N\ZZ_2)$ such that 
$p|_Z$ is a deterministic distribution specified by $(\delta^e,\delta^f,\delta^g)$ where $e,f,g\in\ZZ_2$. The  cocycle $\beta_p$ that lives on the quotient space $\bar X$  is calculated as in Fig.~(\ref{fig:coho-witness}).
The quotient $\bar X$ is a torus, in particular, it is a closed surface.
Therefore $[\beta_p]\neq 0$ if and only if the cochain $\beta_{p}\neq 0$, i.e. when $e+f+g=1 \mod 2$. 
}
}

\section{Quantum measurements on spaces}
\label{sec:QuantumMeasSpace} 

In this section we generalize quantum measurements with a discrete, usually finite, set of outcomes to quantum measurements associated to a space of outcomes.
This construction gives us 
the notion of simplicial quantum measurements. In addition, if we are given a quantum state, then the Born rule, or a simplicial version thereof, can be used to obtain a simplicial probability distribution. This affords us the ability to examine whether a quantum state is contextual by examining the resulting simplicial distribution. Similar constructions also work for projective measurements. Applying the cohomological witness introduced in Section \ref{sec:CohoWit} to state-independent contextuality results in a precise connection to the earlier topological approach. We then use these tools to rigorously analyze the state-independent Mermin square scenario introduced in Section \ref{sec:MotIdea}. 



\subsection{Simplicial quantum measurements} 

Let $\hH$ denote a finite-dimensional complex Hilbert space and $L(\hH)$ denote the complex vector space of linear maps on $\hH$. 
We will write $\Pos(\hH)\subset L(\hH)$ for the set of positive semidefinite operators.
A  {\it quantum measurement} \cite{watrous2018theory}  on a set $U$ is a function 
$$
P:U\to \Pos(\hH)
$$ 
of finite support, i.e. $P(u)\neq \zero$ for only finitely many $u\in U$, such that $\sum_{u\in U} P(u)=\one$. 
When  $\hH=\CC$ a quantum measurement is  the same as a probability distribution on $U$.
We will write $Q_\hH(U)$ for the set of quantum measurements on $U$. 
Given a function $s: V\to U$ and a quantum measurement $P\in Q_\hH(V)$ we define the quantum measurement $Q_\hH s(P)$ on $U$  by the assignment
$$
 u\mapsto \sum_{v\in s^{-1}(u)} P(v).
$$

\Def{
{\rm
Let $Y$ be a simplicial set representing a space of outcomes. 
The simplicial set $Q_\hH(Y)$, called the {\it space of quantum measurements} on $Y$, consists of the $n$-simplices $Q_\hH(Y_n)$ for $n\geq 0$. 
The simplicial structure maps are given by $Q_\hH d_i$ and $Q_\hH s_j$, which for simplicity of notation will be denoted by $d_i$ and $s_j$. 
A {\it simplicial quantum measurement} on a simplicial scenario $(X,Y)$ is a map $X\to Q_\hH Y$ of spaces. We will write $\SQM(X,Y)$ for the set of simplicial quantum measurements on $(X,Y)$. 
}
}

Our construction is natural with respect to the underlying Hilbert space in the sense that a linear map $\Phi:L(\hH)\to L(\kK)$ that is positive and unital \cite{watrous2018theory} induces a map $\Phi_*:Q_\hH Y \to Q_\kK Y$ of spaces.
In particular, given a quantum state $\rho$ the Born rule can be used to define  a function $\rho_*:Q_\hH U\to D_\nnegR U$ by sending a quantum measurement  $P$ to the probability distribution $\rho_*(P)(u)=\Tr(\rho P(u))$.  
We will extend this  to the level of spaces:
A quantum state $\rho$ can be used to define a map of spaces
\begin{equation}\label{diag:simp-Born}
\rho_*:Q_\hH Y\to D_{\nnegR} Y
\end{equation} 
by sending a quantum measurement $P\in Q_\hH Y_n$ to the distribution  on $Y_n$ defined by
$
 \theta \mapsto \Tr(\rho P(\theta) )
$. 
Compatibility of $\rho_*$ with the simplicial identities follows from the linearity of the trace. 
 
 \Def{\label{def:quantum-contextuality}
{\rm
A quantum state $\rho$ is called {\it (non)contextual with respect to a simplicial quantum measurement} $P:X\to Q_\hH Y$ on  $(X,Y)$ if the composite map $\rho_*P: X \xrightarrow{P} Q_\hH Y \xrightarrow{\rho_*} D_\nnegR Y$ is (non)contextual. A quantum state $\rho$ is called {\it strongly contextual} with respect to $P$ if the distribution $\rho_*P$ is strongly contextual.
}
}


\subsection{Projective measurements on the nerve}
 
Let $\Proj(\hH)\subset \Pos(\hH)$ denote the subset of projection operators. 
A {\it projective measurement} on a set $V$ is a quantum measurement of the form $\Pi:V\to \Proj(\hH)$ \cite{watrous2018theory}. For such measurements it holds that for distinct elements of $V$ the projectors are orthogonal: $\Pi(v)\Pi(v')=0$ for $v\neq v'$ in $V$ \cite[Proposition 2.40]{watrous2018theory}. 
We will write $P_\hH(V)$ for the set of projective measurements  on $V$. 
The simplicial constructions  for quantum measurements also work for projective measurements: The simplicial set $P_\hH Y$, called the {\it space of simplicial projective measurements} on $Y$, consists of the $n$-simplices $P_\hH(Y_n)$ for $n\geq 0$ with the simplicial structure maps $P_\hH d_i$ and $P_\hH s_j$, which for simplicity denoted by $d_i$ and $s_j$. 
We will write $\SPM(X,Y)$ for the set of projective measurements on $(X,Y)$.
Note that $P_\hH Y$ is a subspace of $Q_\hH Y$. 
The simplicial version of the Born rule given in 
(\ref{diag:simp-Born}) induces a commutative diagram of spaces
\begin{equation}\label{diag:delta-simp-Born}
\begin{tikzcd}
P_\hH Y  \arrow{r}{\rho_*} &  D_\nnegR Y \\
Y \arrow{u}{\delta_{\hH,Y}} \arrow[ur,"\delta_{Y}"'] &
\end{tikzcd}
\end{equation}
where 
\begin{itemize}
\item $\delta_Y$ sends $\theta \in Y_n$ to the delta distribution $\delta^\theta\in D_R Y_n$ defined in Eq.~(\ref{eq:delta-dist}),
\item $\delta_{\hH,Y}$ sends $\theta \in Y_n$ to the projective measurement 
\begin{equation}\label{diag:delta-proj-meas}
\delta^\theta_\hH:Y_n\to  \Proj(\hH)
\end{equation}
defined by $\delta_\hH^\theta(\theta')=\one$ for $\theta'=\theta$ and $\zero$ otherwise. 
\end{itemize}
Commutativity of the diagram is a consequence of  the trace $1$ condition satisfied by quantum states.
There is also a similar diagram for $Q_\hH Y$.

Projective measurements on the nerve space $N\ZZ_d$ have an alternative description.  Let $U(\hH)$ denote the group of unitary operators acting on $\hH$. 
Let $N(\ZZ_d,U(\hH))$ denote the simplicial set\footnote{This space is a version of the {\it classifying space for commutativity} introduced in \cite{adem2012commuting}.} whose $n$-simplices are given by
$$
N(\ZZ_d,U(\hH))_n = \set{(A_1,A_2,\cdots,A_n)\,:\, A_i \in U(\hH),\;A_i^d=\one,\; A_iA_j=A_jA_i} 
$$
together with the face maps
\begin{equation}\label{eq:faces-nerve}
d_i(A_1,A_2,\cdots,A_n) = \left\lbrace
\begin{array}{ll}
(A_2,A_3,\cdots,A_n) & i=0 \\
(A_1,A_2,\cdots,A_i A_{i+1},\cdots,A_n) &  0<i<n \\
(A_1,A_2,\cdots,A_{n-1}) & i=n
\end{array}
\right.
\end{equation}
and the degeneracy maps 
$$
s_j(A_1,A_2,\cdots,A_n) = (A_1,A_2,\cdots,A_j,\one,A_{j+1},\cdots,A_n)\;\text{ for } 0\leq j\leq n.
$$

\Pro{\label{pro:iso-PNZd}
The spectral decomposition of unitary operators induces an isomorphism of simplicial sets 
$$\spec:N(\ZZ_d,U(\hH)) \to P_\hH (N\ZZ_d) .$$ 
}
\Proof{
Let $\omega$ denote the $d$-th root of unity $e^{2\pi i/d}$.
In dimension $n$ the simplicial set map $\spec$ sends a tuple $(A_1,\cdots,A_n)$ of unitaries to the projective measurement $\Pi: \ZZ_d^n \to \Proj(\hH)$ where  $\Pi(a_1,\cdots,a_n)$ projects onto the simultaneous eigenspace with eigenvalues $(\omega^{a_1},\cdots,\omega^{a_n})$. This assignment respects the simplicial structure.
The inverse map $\spec^{-1}$ sends a projective measurement $\Pi$ to the tuple of unitaries $(A_1,\cdots,A_n)$ where 
$$ 
A_i = \sum_{a_1\cdots a_n} \omega^{a_i} \Pi(a_1,\cdots,a_n).  
$$
}

\Ex{\label{ex:state-dep-Mer-sq}
{\rm
The state-dependent Mermin square scenario together with the Pauli observables in Fig.~(\ref{fig:Mermin-sq-relative}) can be regarded as a morphism $A:X\to N(\ZZ_2,U(\hH))$ where $X$ is the underlying space and $\hH=(\CC^2)^{\otimes 2}$. This morphism assigns a pair of observables for each $2$-context of $X$. For example, $f_{y_0x_0}$ is sent to the pair $(\one\otimes X,X\otimes \one)$. If we compose with the isomorphism in  Proposition \ref{pro:iso-PNZd} we obtain a simplicial projective measurement $P:X\xrightarrow{A} N(\ZZ_2,U(\hH)) \xrightarrow{\spec} P_\hH N\ZZ_2$. Under this morphism a $2$-context is sent to the projective measurement given by the simultaneous diagonalization of the pair of observables. For example, $f_{y_0x_0}$ will be sent to the  measurement determined by the projectors
$$
\Pi_{y_0x_0}^{ab} = \frac{\one+(-1)^a X}{2} \otimes \frac{\one+(-1)^b X}{2}.
$$

}
}

\subsection{State-independent contextuality}
 \label{sec:state-indep-contex}
   
State-independent contextuality arises from the impossibility of being able to  assign eigenvalues to a set of observables in such a way that all product relations\footnote{Normally all functional relations among mutually commuting observables are expected to be satisfied by the assigned eigenvalues. However, in this note we restrict to product relations only, similar to those that appear in the contextuality proof of the Mermin square scenario.} among mutually commuting observables are also satisfied by the assigned eigenvalues \cite{mermin1993hidden}. 
As mentioned in the introduction (Section \ref{sec:TopProofCont}) this type of contextuality is detected by a cohomology class constructed in \cite{Coho}.
This cohomology class can also be described in the simplicial framework by using the cohomological witness introduced in Section \ref{sec:CohoWit}.

We begin by introducing a version of contextuality that does not depend on a quantum state.  

\Def{\label{def:ContextualMeas}
{\rm
A simplicial quantum measurement $P:X\to Q_\hH Y$  is called {\it contextual} if $\rho_*(P)$ (defined by 
(\ref{diag:simp-Born})) is strongly contextual for every quantum state $\rho$.
}
}

To capture state-independent contextuality proofs studied in \cite{Coho} we consider
projective measurements on the outcome space $Y=N\ZZ_d$ and specialize the setting of Section \ref{sec:CohoWit} to the case where the subspace is a circle. 
More precisely, the setting is as follows:
\begin{itemize}
\item $X$ is a space of measurements with a subspace $Z=S^1$,
\item $\Pi:X\to P_\hH N\ZZ_d$ is a simplicial projective measurement such that 
the generating simplex $\sigma=\sigma^{01}$ of the circle is mapped to 
\begin{equation}\label{eq:proj-meas-sigma}
\Pi_\sigma =\delta^1_\hH,
\end{equation}
the projective measurement on $\ZZ_d$ defined by
 $\delta^1_\hH(a)=\one$ for $a=1$ and $\zero$ otherwise (i.e. the measurement defined in 
 (\ref{diag:delta-proj-meas}) for the $1$-simplex $\theta=1$)). 
\end{itemize}
Let $\rho$ be a quantum state. 
Consider the simplicial distribution $p_\rho:X\to D_\nnegR N\ZZ_d$ obtained using the Born rule $p_\rho=\rho_*(\Pi)$. 
The distribution $p_\rho|_{S^1}$ obtained by restriction to the circle 
is a deterministic distribution:
the generating simplex $\sigma$ to the delta distribution $\delta^1\in D_\nnegR \ZZ_d$ concentrated at $1\in \ZZ_d$. 
Moreover, by Example \ref{ex:H1-circle} we have 
$$
\alpha_{S^1}:D(S^1,N\ZZ_d)\xrightarrow{\cong} H^1(S^1) = \ZZ_d,
$$
under which the deterministic distribution $p|_{S^1}$ is mapped to $1\in \ZZ_d$. 
The cohomology witness that constitutes the set $\cl_{S^1}(p_\rho)$ (defined in Eq.~(\ref{eq:coho-classes})) is given by 
$$
[\beta] = \zeta\circ \alpha_{S^1}(p_\rho|_{S^1})= \zeta(1) \in H^2(\bar X).
$$
Note that the cohomology class $[\beta]$ is independent of the quantum state since it only depends on the restriction of $p_\rho$ to the circle. 

\Cor{
If $[\beta]\neq 0$ then $\Pi$ is contextual.
}
\Proof{
Since $\cl_{S^1}(p_\rho)=\set{[\beta]}$ Corollary \ref{cor:Coho-StrongContex} implies that if $[\beta]\neq 0$ the distribution $p_\rho$ is strongly contextual for any quantum state $\rho$. 
}

In practice a simplicial projective measurement $\Pi$ satisfying Eq.~(\ref{eq:proj-meas-sigma}) comes from a map $A:X\to N(\ZZ_d,U(\hH))$ of spaces such that
\begin{equation}\label{eq:A-map-sigma}
A_\sigma=\omega \one.
\end{equation}
Using the isomorphism in Proposition \ref{pro:iso-PNZd} we can obtain a simplicial projective measurement $\Pi=\spec\circ A$ that satisfies   Eq.~(\ref{eq:proj-meas-sigma}). The typical example is the Mermin square scenario interpreted in the simplicial framework.

\begin{figure}[h!]
\begin{center}
\includegraphics[width=.5\linewidth]{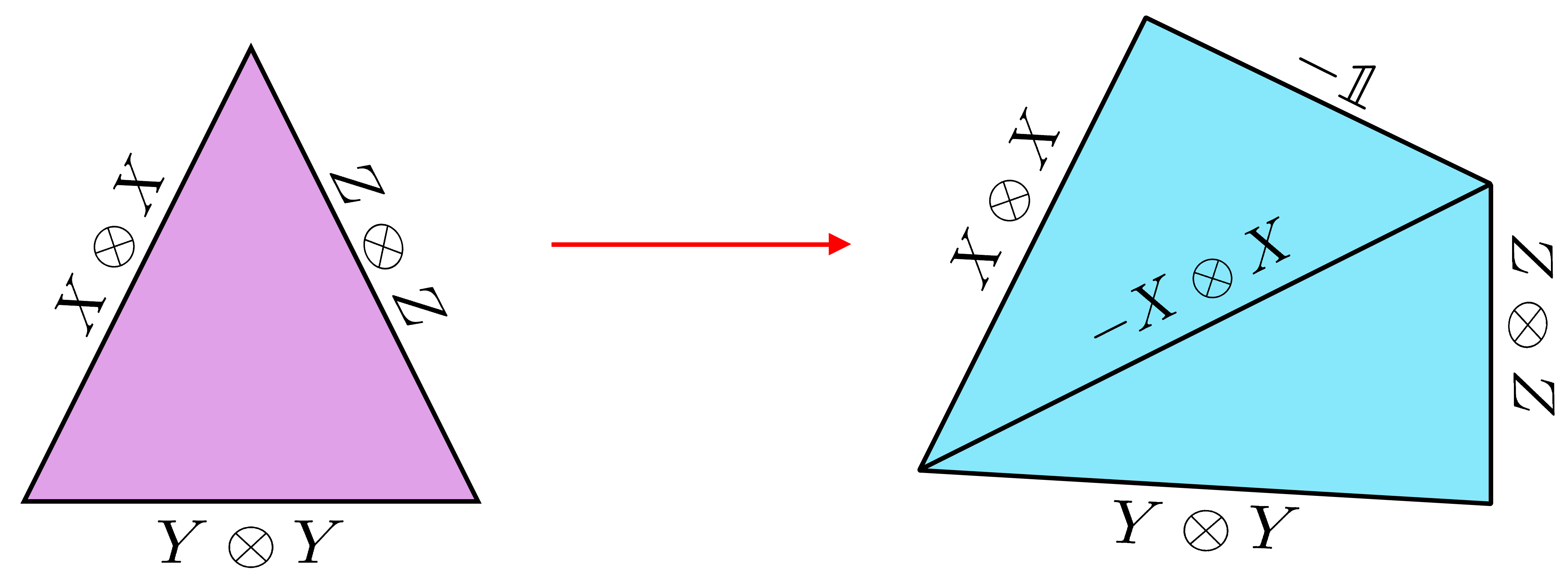} 
\end{center} 
\caption{ 
}
\label{fig:TwoSimplexRefinement}
\end{figure}  
 
\Ex{\label{ex:mermin-sq-quotient-beta}
{\rm
Let us consider the  state-independent Mermin square. 
The space of measurements is the torus $T$ depicted in Fig.~(\ref{fig:Mermin-sq-beta}) and the space of outcomes is $N\ZZ_2$. 
As in the case of state-dependent Mermin square discussed in Example \ref{ex:state-dep-Mer-sq}
we will interpret the assignment of observables to the edges as a map of spaces 
$$
A: X \to N(\ZZ_d,U(\CC^4))
$$
where $X$ is a space of measurements obtained by modifying the torus $T$. 
In the torus the $2$-context for which the cochain $\beta=1$ does not define a simplex of $N(\ZZ_2,U_4)$ since $(X\otimes X)(Z\otimes Z)=-Y\otimes Y$. 
This context can be decomposed into  two simplices as in Fig.~(\ref{fig:TwoSimplexRefinement}). 
Replacing this single simplex in the torus with the two simplices produces a punctured torus with a triangulation 
  as in the middle figure in Fig.~(\ref{fig:Mermin-Sq-Simp}). 
Note that this procedure produces an extra loop $\sigma$, which generates a subspace $Z=S^1$, such that $A_\sigma = -\one$.
Let us write $i:S^1 \to X$ for the inclusion of the loop. 
We can go back to the torus by contacting this loop to a point, that is by considering the quotient map $q:X\to \bar X=X/S^1$. This time though the triangulation of the torus will be different than the original one at the modified $2$-simplex. By a computation similar to Example \ref{ex:state-dep-Mer-sq} we see that the cochain $\beta=1$ on the torus, and hence the associated cohomology class $[\beta]$ is nonzero. This is our cohomological witness for the contextuality of the simplicial projective measurement $\Pi=\spec\circ A$.
\begin{figure}[h!]
\begin{center}
\includegraphics[width=.8\linewidth]{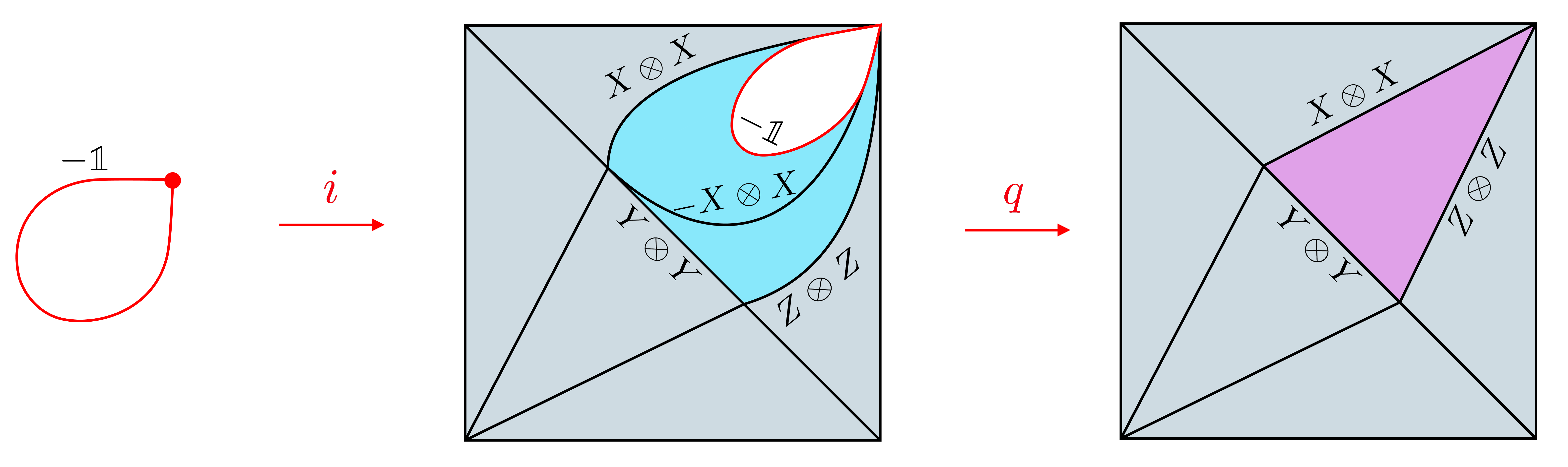} 
\end{center} 
\caption{The inclusion $i:S^1\to X$   maps the circle to the loop corresponding to $-\one$. The second map is the quotient map $q:X\to \bar X$ collapsing the loop to a point. 
}
\label{fig:Mermin-Sq-Simp}
\end{figure}
}}

Other examples studied in \cite{Coho} can also be put into this framework by a similar modification. There is an important consequence of this example which relates to   state-independent contextuality in terms of the failure of eigenvalue assignments. 

\Cor{\label{cor:PNZ2-does-not-split}
The inclusion map $\delta_{N\ZZ_2}:N\ZZ_2 \to P_\hH N\ZZ_2$ (see Diag.~(\ref{diag:delta-simp-Born})) does not split for $\dim(\hH)\geq 4$.
}
\begin{proof}
Let $\Pi:X\to P_{\CC^4} N\ZZ_2$ denote the
simplicial projective measurement considered in Example \ref{ex:mermin-sq-quotient-beta}. We can elevate this to a projective measurement defined over an Hilbert space $\hH\cong\CC^d$ of dimension $d\geq 4$ by using the linear isometry $\CC^4 \to \CC^d$ onto the subspace spanned by the first $4$ canonical basis vectors. Using this observation we can construct the right-hand square of the following commutative diagram 
$$
\begin{tikzcd}
S^1 \arrow[r,"h"] \arrow[d,"i"] &  N\ZZ_2 \arrow[r,equal] \arrow[d,"\delta_{N\ZZ_2}"] &  N\ZZ_2 \arrow[d,"\delta_{N\ZZ_2}"] \\
X \arrow[r,"\Pi"] & P_{\CC^4} N\ZZ_2 \arrow[r] & P_\hH N\ZZ_2
\end{tikzcd}
$$
In Example \ref{ex:mermin-sq-quotient-beta} we showed that $[\beta]\neq 0$. By Lemma \ref{lem:split-connecting-zero} this implies that $i:S^1\to X$ does not split. Therefore $N\ZZ_2\to P_\hH N\ZZ_2$ does not split since such a splitting would provide a splitting for $i$.
\end{proof}
We will see in Section \ref{sec:KS} that this result can be improved to include Hilbert spaces of dimension $3$ and the cases of all $d\geq 2$ as a consequence of the Kochen--Specker theorem (Theorem \ref{thm:KS}). 
To make connection to eigenvalue assignments consider the composition
$$
N\ZZ_d \xrightarrow{\delta_{N\ZZ_d}} P_\hH N\ZZ_d \xrightarrow{\spec^{-1}} N(\ZZ_d,U(\hH)).
$$ 
Under this map an $n$-simplex $(a_1,\cdots,a_n)$ 
of the nerve space
is sent to the tuple $(\omega^{a_1}\one,\cdots,\omega^{a_n}\one)$. 
Splitting of this map would, in particular, yield an assignment of values $\nu(A)\in \ZZ_d$ in such a way that
\begin{enumerate} 
\item $\nu(\omega^{a}\one)=a$ for all $a\in \ZZ_d$,
\item $\nu(A_1)+\nu(A_2)=\nu(A_1A_2)$ for all pairwise commuting unitaries $A_1$ and $A_2$ (with the $d$-torsion condition: $A_i^d=\one$). 
\end{enumerate}
This is precisely an eigenvalue assignment that respects product relations among commuting $d$-torsion unitaries. Corollary \ref{cor:PNZ2-does-not-split} has the interpretation that for $d=2$ and $\dim(\hH)\geq 4$ such an eigenvalue assignment is impossible.

\section{Foundational theorems in the simplicial setting} \label{sec:FoundThms}

In this section we describe two foundational results in quantum theory, Gleason's theorem and Kochen--Specker theorem, in our simplicial framework.

\subsection{Gleason's theorem}

In this section we will  formulate Gleason's theorem in our simplicial framework; see \cite[Prop 1.7.1]{mansfield2013mathematical} for a related result in the sheaf-theoretic framework. 
For this we will consider the simplicial scenario $(X,Y)$ where the space of measurements is  $X=P_\hH S^1$ and the space of outcomes is $Y= S^1$. 
We begin  our discussion with a description of $P_\hH S^1$. The situation is quite similar to simplicial distributions on the circle (see \S \ref{sec:DistOnCircle}) except that instead of probabilities we will use projectors. We can represent an $n$-simplex $\Pi\in (P_\hH S^1)_n$ by an $n$-tuple $(\Pi^1,\Pi^2,\cdots,\Pi^n)$ where $\Pi^k = \Pi(\sigma^{0_{k-1}1_{n-k}})$. 
In this representation the face maps are given by
$$
d_i(\Pi^1,\Pi^2,\cdots,\Pi^{n}) = \left\lbrace
\begin{array}{ll}
(\Pi^2,\Pi^3,\cdots,\Pi^{n}) & i=0\\ 
(\Pi^1,\Pi^2,\cdots,\Pi^i+\Pi^{i+1},\cdots, \Pi^{n}) & 0<i<n \\
(\Pi^1,\Pi^2,\cdots,\Pi^{n-1}) & i=n
\end{array}
\right.
$$
and the degeneracy maps for $0\leq j\leq n$ are given by
$$
s_j(\Pi^1,\Pi^2,\cdots,\Pi^{n}) = (\Pi^1,\Pi^2,\cdots,\Pi^j,\zero,\Pi^{j+1},\cdots,\Pi^n).
$$ 

In our formalism Gleason's theorem 
will be used to obtain a precise description of 
simplicial distributions $f:P_\hH S^1 \to D_{\nnegR} S^1$.  
We are interested in the maps $f$ which make the following diagram commute
\begin{equation}\label{diag:diag-PS1-DS1}
\begin{tikzcd}
P_\hH S^1  \arrow{r}{f} &  D_\nnegR S^1 \\
S^1 \arrow{u}{\delta_{\hH,S^1}} \arrow[ur,"\delta_{S^1}"'] &
\end{tikzcd}
\end{equation}
As we have seen in Diag.~(\ref{diag:delta-simp-Born}) quantum states (density operators) produce such maps via the Born rule.  
We will write $\siDist_{S^1}(P_\hH S^1 ,   S^1)$ for the set of simplicial distributions   making Diag.~(\ref{diag:diag-PS1-DS1}) commute.

\Thm{[Gleason]\label{thm:Gleason}
Let $\hH$ be a finite-dimensional Hilbert space such that $\dim(\hH)\geq 3$. Then the map
$$
 \Den(\hH) \to \siDist_{S^1}(P_\hH S^1 ,  S^1)
$$
defined by sending a density operator $\rho$ to the  simplicial distribution $\rho_*:P_\hH S^1\to D_\nnegR S^1$ given by the Born rule in  Diag.~(\ref{diag:simp-Born})
is a bijection of sets.
}
\Proof{
If we represent the $n$-simplices of $P_\hH S^1$ by tuples of projectors then the Born rule is simply the assignment   
$$
(\Pi^1,\cdots,\Pi^n) \mapsto (\Tr(\rho\Pi^1),\cdots,\Tr(\rho\Pi^n)).
$$ 
Gleason's theorem \cite{gleason1975measures} says that when $\dim(\hH)\geq 3$ finitely-additive normalized measures $\mu:\co{\Proj(\hH)}\to \nnegR$ are given by density operators. Recall that a finitely-additive $R$-measure is a function $\mu:\Proj(\hH)\to R$ such that for any set $\set{\Pi^1,\cdots,\Pi^n}$ of pairwise orthogonal projectors  $\mu(\sum_i \Pi^i)=\sum_i \mu(\Pi^i)$.
If in addition $\mu(\one)=1$ then it is called a normalized measure. 
Then $\mu$ gives a simplicial distribution $\mu_*:P_\hH S^1 \to D_\nnegR S^1$ 
which sends an 
$n$-context 
represented by $(\Pi^1,\cdots,\Pi^n)$ to $(\mu(\Pi^1),\cdots,\mu(\Pi^n))$. 
Such a map makes Diag.~(\ref{diag:diag-PS1-DS1}) commute. 
Therefore $\mu_*$ is an element of $\siDist_{S^1}(P_\hH S^1, S^1)$. Conversely, consider a simplicial distribution $f:P_\hH S^1\to D_\nnegR S^1$: 
\begin{itemize}
\item For $n=1$ the map $f_1$ sends a projector $\Pi^1$ to a number $0\leq p\leq 1$.
\item For $n=2$ the map $f_2$ sends a pair of projectors $(\Pi^1,\Pi^2)$ to a pair $(p_1,p_2)$ of nonnegative numbers such that $0\leq p_1 + p_2\leq 1$.
\item The compatibility of $f$ with   the face maps $d_0$, $d_1$ and $d_2$ implies that $f_1$ assigns $p_1+p_2$ to the projector $\Pi^1+\Pi^2$.
\end{itemize}
The up shot is that $f_1:\Proj(\hH) \to \nnegR$ is a finitely-additive measure. The commutativity of Diag.~(\ref{diag:diag-PS1-DS1}) implies:
\begin{itemize}
\item For $n=1$ the generating simplex $\sigma^{01}$ of the circle
is mapped to the projector $\one$ under $\delta_{\hH,S^1}$, and then mapped to $f_1(\one)$ under $f$, which (by the commutativity of the diagram) is the same as the image  
of the diagonal map $\delta_{S^1}$, which is $1\in \ZZ_d$.
\end{itemize}
This condition implies that $f_1$ is also normalized. Hence by the Gleason's theorem $f_1$ is induced by a quantum state $\rho$, that is $f_1(\Pi)=\Tr(\rho\Pi)$. To finish the proof we observe that $f$ is uniquely determined by $f_1$: Suppose $f(\Pi^1,\cdots,\Pi^n)=(p^1,\cdots,p^n)$. Using $\varphi_i$ defined in Eq.~(\ref{eq:varphi-k}) we compute
$$
p^i = \varphi_i (p^1,\cdots,p^n) 
= \varphi_if(\Pi^1,\cdots,\Pi^n) 
= f(\varphi_i (\Pi^1,\cdots,\Pi^n)) 
= f_1(\Pi^i)=\Tr(\rho \Pi^i).
$$
}

\subsection{Kochen--Specker theorem}\label{sec:KS}

As  is well-known \cite[Page 189]{moretti2019fundamental}  
Gleason's theorem can be used to deduce the Kochen--Specker theorem. The main idea is to consider the embedding 
\begin{equation}\label{diag:embeddingR3-C3}
\varphi:\Proj(\RR^3) \to \Proj(\CC^3)
\end{equation}
 defined as follows:
\begin{itemize}
\item We set $\varphi(\zero)=\zero$ and $\varphi(\one)=\one$.

\item For a nonzero $v \in \RR^3$ let $\Pi_v$ denote the projector onto the line spanned by $v$. Let $J_v =  \sum_{i=1}^3 v_i J_i $ where $J_i$'s are given by
$$
J_1 =  \frac{1}{\sqrt 2} \left( \begin{matrix}
 0 & 1 & 0\\
 1 & 0 & 1\\
 0 & 1 & 0
\end{matrix} \right),\;\;
J_2 =  \frac{1}{\sqrt 2} \left( \begin{matrix}
 0 & -i & 0\\
 i & 0 & -i\\
 0 & i & 0
\end{matrix} \right),\;\;
J_3 =    \left( \begin{matrix}
 1 & 0 & 0\\
 0 & 0 & 0\\
 0 & 0 & -1
\end{matrix} \right).
$$
We set  $\varphi(\Pi_v)= J_v^2$.

\item For a $2$-dimensional subspace $V\subset \RR^3$ let $\Pi_V$ denote the associated projector. Let $v$ denote a nonzero vector orthogonal to $V$. We set $\varphi(\Pi_V) = \one - J_v^2$.
\end{itemize}

\Thm{[Kochen--Specker]\label{thm:KS}
Let $\hH$ be a finite-dimensional Hilbert space such that $\dim(\hH)\geq 3$.
Then the map of spaces 
$$
\delta_{S^1} :S^1 \to P_\hH S^1
$$
 does not split, i.e. there does not exist a map $ \phi:P_\hH S^1 \to S^1$ of spaces such that the composite $\phi\circ \delta_{S^1}$ is  the identity map on $S^1$.   
}
\begin{proof}
It suffices to prove the theorem for $\CC^3$ since we can embed it as a subspace in $\hH$ when $\dim(\hH)\geq 3$ and deduce the result for the larger Hilbert space. For the rest we take $\hH=\CC^3$. A similar argument as in Theorem \ref{thm:Gleason} shows that  maps $g:P_\hH S^1 \to S^1$ that make the following diagram commute 
\begin{equation}\label{diag:diag-PS1-S1}
\begin{tikzcd}
P_\hH S^1  \arrow{r}{g} &  S^1 \\
S^1 \arrow{u}{\delta_{S^1}} \arrow[ur,equals] &
\end{tikzcd}
\end{equation}
coincide with  finitely-additive normalized measures 
\co{$\mu:\Proj(\hH) \to \NN$}. Such a measure satisfies the property that for any set $\set{\Pi^1,\cdots,\Pi^n}$ of orthogonal projectors summing to $\one$ the function $\mu$ assigns $1$ exactly to one of the projectors and assigns $0$ to the rest in the set. Consider the embedding $\varphi:\Proj(\RR^3)\to \Proj(\hH)$ given in 
(\ref{diag:embeddingR3-C3}) and  define $\pi:S^2 \to \Proj(\RR^3)$ by sending a point $v$ to the projector $\Pi_v$. 
Now observe that the following function 
$$
\bar\mu:S^2\xrightarrow{\pi} \Proj(\RR^3) \xrightarrow{\varphi} \Proj(\hH) \xrightarrow{\mu} \NN
$$
has image given by the set $\set{0,1}$ since for any triple of pairwise orthogonal lines the map $\mu$ assigns $1$ to exactly of one them and $0$ to the rest. 
Let $\tau:\NN\to \RR_{\geq 0}$ denote the inclusion of the semirings. 
Gleason's theorem says that the composite $\Proj(\hH)\xrightarrow{\mu} \NN \xrightarrow{\tau} \RR_{\geq 0}$, which is a finitely additive $\RR_{\geq 0}$-measure, is given by $\Pi\mapsto \Tr(\rho \Pi)$ for some $\rho\in \Den(\hH)$. But this implies that $\bar\mu:S^2\to \NN$ is given by \co{$v\mapsto \Tr(\rho J_v)$}. Therefore $\bar\mu$ is a continuous map and since $S^2$ is connected its image cannot be  $\set{0,1}$, a disconnected set.  Hence there does not exist any finitely-additive $\NN$-measures, which translates into the statement that there does not exist any map $g$ of spaces making Diag.~(\ref{diag:diag-PS1-S1}) commute.
\end{proof}
 
Kochen--Specker theorem can be used to improve Corollary \ref{cor:PNZ2-does-not-split}.

\Cor{\label{cor:NZd-PNZd-doesnotsplit}
If $\dim(\hH)\geq 3$ then for all $d\geq 2$ the inclusion map $
\delta_{N\ZZ_d} : N\ZZ_d \to P_\hH N\ZZ_d
$
does not split.}
\begin{proof}
This follows from the commutative diagram of spaces
$$
\begin{tikzcd}
S^1 \arrow[r,"h"] \arrow[d,"\delta_{S^1}"] & N\ZZ_d \arrow[d,"\delta_{N\ZZ_d}"] \\
P_\hH S^1 \arrow[r,"P_\hH h"] & P_\hH N\ZZ_d
\end{tikzcd}
$$
A splitting for $\delta_{N\ZZ_d}$ would give a splitting for $\delta_{S^1}$, which is impossible by Theorem \ref{thm:KS}.
\end{proof}

\section{Conclusion}\label{sec:conclusion}

Here we have presented a new framework based on the theory of simplicial sets which 
\co{extends}
both the topological approach of \cite{Coho} and the sheaf-theoretic approach of \cite{abramsky2011sheaf}. Our formalism generalizes the standard notion of nonsignaling distributions to that of simplicial distributions which are based on scenarios consisting of \emph{spaces} of measurements and outcomes. Contextuality is introduced at this level of generality and fundamental theorems such as Fine's theorem, Gleason's theorem, and the Kochen--Specker theorem are presented in this new formalism. The flexibility afforded by the simplicial framework allows us, moreover, to introduce novel characterizations of contextuality both topological and cohomological.

A chief contribution of our approach is to promote the \emph{set} of measurements and outcomes to a \emph{space} of measurements, given by a simplicial set $X$, and a space of outcomes, given by a simplicial set $Y$. In this picture contexts correspond to the $n$-simplices of $X$ and outcomes correspond to the $n$-simplices of $Y$, both of which have an intrinsic dimensionality. The basic objects of our formalism, called simplicial distributions, are then represented by maps of spaces $p:X\to D_RY$. Our constructions are natural with respect to the change of measurement space or the change of outcome space, as well as the underlying semiring $R$ in which distributions take values. In fact, those familiar with category theory  may notice that with our definitions everything takes place in the Kleisli category \cite{kleisli1965every} of the category of simplicial sets with respect to $D_R$ (regarded as a monad on this category).
Simplicial distributions appear as morphisms in this category. \co{See \cite{kharoof2022simplicial} for the role of convexity in this category-theoretic perspective.} 
It would be interesting to compare this category and its natural variants to other categories associated to nonsignaling distributions, e.g.  \cite{karvonen2018categories,barbosa2021closing}. 
Many more  aspects of these objects, in particular the role of simplicial homotopy, remain to be explored, which we leave for future work to be addressed elsewhere.

Another appealing aspect of the current formalism is how topology can aid in the study of contextuality. This was demonstrated explicitly in our  topological proof of Fine's theorem for \rev{characterizing noncontextuality in}  the CHSH scenario.
A similar analysis can be carried over to other scenarios;
e.g., Bell scenarios \cite{brunner2014bell}, KS-scenarios \cite{budroni2021quantum}, $n$-cycle scenarios \cite{araujo2013all}, etc. In fact, the framework of simplicial sets allows for a much richer landscape of measurement scenarios than those typically encountered, which may be of interest for physics and mathematics alike. We also highlight that noncontextuality can be characterized here in an economical fashion through the constructions of gluing and extension. Such notions provide a potentially powerful analytic tool for applications in quantum computation, such as extending \cite{raussendorf2016cohomological} to temporally ordered MBQC.

\paragraph{Acknowledgments.}
This  work is supported by the Air Force Office of Scientific Research under award
number FA9550-21-1-0002. The first author would like to thank Robert Raussendorf and Ben Williams for useful discussions.

\appendix

\section{Simplicial identities}
\label{sec:app-SimplicialIdentities}

A simplicial set consists of a sequence of sets $X_0,X_1,\cdots,X_n,\cdots$ together with the face maps $d_i:X_n\to X_{n-1}$ and the degeneracy maps $s_j:X_n \to X_{n+1}$ satisfying the following {\it simplicial identities}:
\begin{equation}\label{eq:simplicial-identities}
\begin{aligned}
d_i d_j &= d_{j-1} d_i  \;\;\;\; \text{ if } i<j \\
s_i s_j &= s_j s_{i-1} \vspace{2cm} \;\;\;\;\text{ if } i>j \\
d_i s_j &= \left\lbrace 
\begin{array}{ll}
s_{j-1} d_i & \text{ if } i<j \\
\idy  & \text{ if } i=j,\,j+1 \\
s_jd_{i-1} & \text{ if } i>j+1.
\end{array}
\right. 
\end{aligned}
\end{equation}
Intuitively a simplicial set is a combinatorial description of a space obtained by gluing $n$-simplices. The simplicial relations encode the faces and the degeneracies of an $n$-simplex and they can be derived by specializing to
$X=\Delta^n$. 
We can do this by applying $d_i$ and $s_j$ to the generating simplex $\sigma^{01\cdots n}$ of $\Delta^n$. Recall that $d_i$ deletes the $i$-th entry and $s_j$ copies the $j$-th entry in the string $01\cdots n$. For example, when $i<j$ first copying the $j$-th index and then deleting the $i$-th index is the same as first deleting the $i$-th index and then copying the $(j-1)$-th:
$$
d_{i} s_j \sigma^{01\cdots n} = \sigma^{01\cdots (i-1)(i+1) \cdots jj \cdots n} = s_{j-1}d_i \sigma^{01\cdots n}.
$$ 
Similarly other identities in Eq.~(\ref{eq:simplicial-identities}) can be verified.
  
Simplicial identities in low dimensions  are the most relevant to us. Although,   face maps are already described in the main text we include them for completeness:
\begin{itemize}
\item For an edge $e\in X_1$, the target vertex is $d_0e$ and the source vertex is $d_1e$. 
\item For a triangle $\sigma \in X_2$, the three faces are given by $d_i\sigma$ for $i=0,1,2$; see Fig.~(\ref{fig:Triangle-faces}).
\end{itemize} 
Degeneracy maps in dimensions $\leq 2$ are depicted in Fig.~(\ref{fig:degeneracies}) and they are given as follows:
\begin{itemize}
\item For a vertex $v\in X_0$, the edge $e=s_0v$ has the target and the source vertices given by $v$ itself.
\item For an edge $e\in X_1$, the triangles $s_0e$ and $s_1e$ are described as in Fig.~(\ref{fig:degeneracies}).
\end{itemize}
\begin{figure}[h!]
\begin{center}
\includegraphics[width=.5\linewidth]{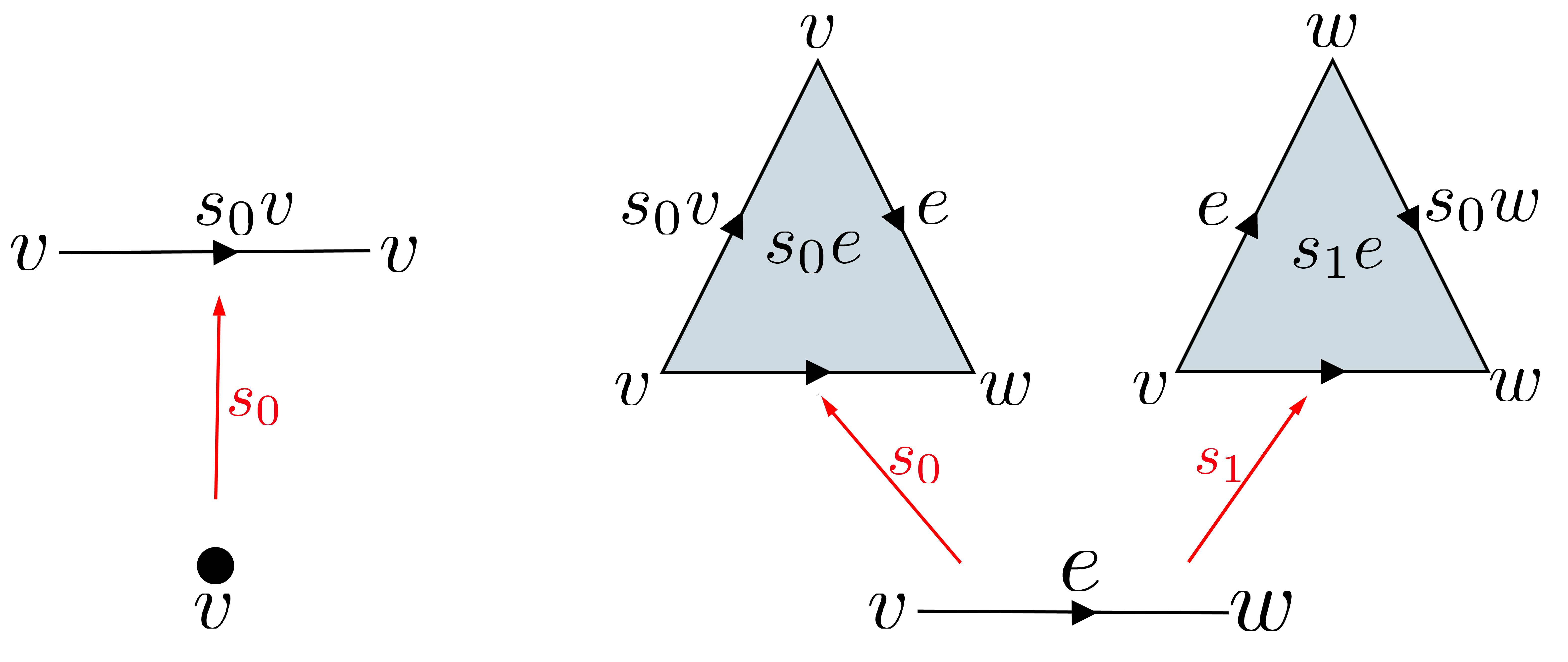} 
\end{center} 
\caption{
}
\label{fig:degeneracies}
\end{figure}

\section{Discrete scenarios}  \label{sec:discrete-scenarios}
 
A {\it discrete scenario}  consists of a triple $(M,\cC,\ZZ_d)$ where 
\begin{itemize}
\item $M$ is a finite set of measurements,
\item $\cC$ is a collection\footnote{The collection $\cC$ is assumed to be an antichain in the sense that any two contexts with $C\subset C'$ implies that $C'=C$.} of contexts covering $M$,
\item $\ZZ_d$ represents the set of outcomes.
\end{itemize}
Such scenarios are called measurement scenarios in the sheaf-theoretic formulation of contextuality \cite{abramsky2011sheaf}.
A {\it nonsignaling distribution} on the scenario $(M,\cC,\ZZ_d)$ is a collection of distributions $p_C\in D_R(\ZZ_d^C)$ satisfying the nonsignaling conditions
$$
p_C|_{C\cap C'} = p_{C'}|_{C\cap C'}\;\;\text{ for }C,C'\in \cC.
$$ 
Restriction to the intersection means marginalization 
$$
p_C|_{C\cap C'}(s) = \sum_{r\,:\,r|_{C\cap C'}=s} p(r)
$$
where $r$ runs over \co{functions} $C\to \ZZ_d$ such that the restriction $r|_{C\cap C'}$ to the intersection coincides with $s$. We will write $\NS(\cC,\ZZ_d)$ for the set of nonsignaling distributions. This space is a (convex) polytope, and is usually referred to as the {\it nonsignaling polytope}. Sending a distribution $d\in D_R(\ZZ_d^M)$ to the nonsignaling distribution $(d|_C)_{C\in \cC}$ gives a function
$$
\Theta':D_R(\ZZ_d^M) \to \NS(\cC,\ZZ_d).
$$ 
   
\Def{\label{def:contextual-discrete}
{\rm
A nonsignaling distribution $p\in \NS(\cC,\ZZ_d)$ is called {\it contextual} (in the sense of \cite{abramsky2011sheaf}) if it does not lie in the image of $\Theta'$. Otherwise, it is called {\it noncontextual}.
}}    

Any discrete scenario can be realized as a simplicial scenario so that  nonsignaling distributions are in one-to-one correspondence with   distributions on 
the simplicial scenario. For this, given a discrete scenario $(M,\cC,\ZZ_d)$ we need to specify the space of measurements and the space of outcomes. 
The space $X_\cC$ of measurements, after we fix an ordering on $M$, is the following simplicial set: 
\begin{itemize}
\item The set $(X_\cC)_n$ of   $n$-simplices is given by   tuples $(m_0,m_1,\cdots,m_n)$ of measurements in $M$, where $m_0\leq m_1\leq \cdots\leq m_n$, such that each $m_i\in C$  for some context $C\in \cC$.  
\item The $i$-th face map deletes the $i$-th measurement
$$
d_i(m_0,m_1,\cdots,m_n) =( m_0,m_1,\cdots,m_{i-1},m_{i+1},\cdots,  m_n)
$$
 and the $j$-th degeneracy map copies the $j$-th measurement
 $$
s_j(m_0,m_1,\cdots,m_n) =( m_0,m_1,\cdots,m_{j-1},m_j,m_{j},m_{j+1},\cdots,  m_n).
$$
\end{itemize} 
As noted in \cite{abramsky2019comonadic} $(M,\cC)$ can be regarded as a simplicial complex. The simplicial set $X_\cC$ is the usual way of realizing an ordered simplicial complex as a simplicial set.
The space of outcomes, denoted by 
$\Delta_{\ZZ_d}$, is the following simplicial set:
\begin{itemize}
\item The set $(\Delta_{\ZZ_d})_n$ of $n$-simplices consists of  $(n+1)$-tuples $(a_0,a_1,\cdots,a_n)$ of outcomes $a_i\in \ZZ_d$.
\item The $i$-th face map acts by deleting the $i$-th outcome
$$
d_i(a_0,a_1,\cdots,a_n) = (a_0,a_1,\cdots,a_{i-1},a_{i+1},\cdots, a_n).
$$ 
\item The $j$-th degeneracy map copies the $j$-th outcome
$$
s_j(a_0,a_1,\cdots,a_n) = (a_0,a_1,\cdots,a_{j-1} ,a_j,a_j,a_{j+1},\cdots, a_n).
$$ 
\end{itemize}
It is straight-forward to verify that the simplicial identities 
in Eq.~(\ref{eq:simplicial-identities}) 
are satisfied.   
Next we consider the space $D_R \Delta_{\ZZ_d}$ of distributions.
An $n$-simplex $p$ is a distribution on $(\Delta_{\ZZ_d})_n=\ZZ_d^{n+1}$. The $i$-th face map  acts by
$$
d_ip(a_0,\cdots,a_{n-1}) = \sum_{a\in \ZZ_d} p(a_0,\cdots,a_{i-1},a,a_{i},\cdots  a_{n-1})
$$
and the $j$-th degeneracy map acts by
$$
s_jp(a_0,\cdots,a_{n+1}) = 
\left\lbrace
\begin{array}{ll}
p(a_0,\cdots,a_{j-1},a_j,a_{j+2},\cdots  a_{n+1}) & a_j=a_{j+1} \\
0 & \text{otherwise.}
\end{array}
\right. 
$$

\begin{thm}\label{thm:ComparisonToSheaf}
There is a commutative diagram
$$
\begin{tikzcd}
\clDist(X_\cC,\Delta_{\ZZ_d}) \arrow[r,"\Theta"] \arrow[d,"\cong"] & \siDist(X_\cC,\Delta_{\ZZ_d}) \arrow[d,"\cong"] \\
D_R(\ZZ_d^M) \arrow[r,"\Theta'"] & \NS(\cC,\ZZ_d)
\end{tikzcd}
$$
where the vertical maps are isomorphisms. In particular, a nonsignaling distribution $p$ is contextual in the sense of Definition \ref{def:contextual-discrete} if and only if it is contextual in the sense of Definition \ref{def:simp-contextuality} when regarded as a distribution on  $(X_\cC,\Delta_{\ZZ_d})$.
\end{thm}  
\begin{proof}
Let $U=\set{u_0,u_1,\cdots,u_k}$, where $u_0\leq u_1\leq \cdots\leq u_k$, be a subset of a context $C\in \cC$. 
Let $\Delta^U\cong \Delta^k$ denote the subspace generated by the simplex $\sigma_U=(u_0,u_1,\cdots,u_k)$. 
The inclusion map $i:\Delta^U \to \Delta^C$ induces a commutative diagram
\begin{equation}\label{diag:comm-diag-det-dist}
\begin{tikzcd}
\dDist(\Delta^C,\Delta_{\ZZ_d}) \arrow[r,"\cong"] \arrow[d,"i^*"] & \ZZ_d^C \arrow[d,"i^*"] \\
\dDist(\Delta^U,\Delta_{\ZZ_d}) \arrow[r,"\cong"] & \ZZ_d^U
\end{tikzcd}
\end{equation}
The top horizontal map sends an outcome map $r$ to the function $r_{\sigma_C}\in \ZZ_d^C$. The bottom horizontal map is also similarly defined. The right vertical map is the restriction of a function  $C\to \ZZ_d$ to the subset $U$. Commutativity of the diagram can be expressed as
$$
(r|_{\Delta^U})_{\sigma_U} = r_{\sigma_C}|_{U}.
$$
Observe that $X_\cC$ is the union of the subspaces $\Delta^C$ as $C$ runs over the contexts in $\cC$. 
This follows from the observation that an $n$-context $(m_0,m_1,\cdots,m_n)$, where $m_i\in C$, is also an $n$-context of $\Delta^C$. 
By Proposition \ref{pro:union-det-simp-distributions} deterministic distributions can be identified with tuples $(d_{\Delta^C})_{C\in \cC}$ of distributions compatible under restrictions to the intersections.
By Diag.~(\ref{diag:comm-diag-det-dist}) these are precisely the tuples $(r_{\sigma_C}\in \ZZ_d^C)_{C\in \cC}$ compatible under restriction to intersections. These are in bijective correspondence with functions $M\to \ZZ_d$. 
Therefore $\clDist(X_\cC,\Delta_{\ZZ_d})\cong D_R(\ZZ_d^M)$. 
Again by Proposition \ref{pro:union-det-simp-distributions} simplicial distributions are given by tuples $(p_{\Delta^C})_{C\in \cC}$ compatible under restriction to intersections. Applying $D_R$ to Diag.~(\ref{diag:comm-diag-det-dist}) allows us to identify such tuples with $(p_C\in D_R(\ZZ_d^C))_{C\in \cC}$ compatible under restriction to intersections. They are precisely the elements of the nonsignaling polytope, that is we have $\siDist(X_\cC,\Delta_{\ZZ_d})\cong \NS(\cC,\ZZ_d)$. Under the identifications $\Theta$-map becomes the $\Theta'$-map.
\end{proof}

\Ex{{\rm \label{ex:Bell-discrete}
Consider the CHSH scenario with
measurement set $M=\set{x_0\leq x_1\leq y_0\leq y_1}$ and outcome set $\ZZ_2$. 
The set $\cC$ of contexts of the scenario is given in Eq.~(\ref{eq:Bell-contexts}).
The corresponding measurement space $X_\cC$ is  $1$-dimensional, in fact the boundary of the square scenario $Q$ of Section \ref{sec:TopProofCont}, as depicted in   Fig.~(\ref{fig:Bell-discrete}). 
\begin{figure}[h!]
\begin{center}
\includegraphics[width=.2\linewidth]{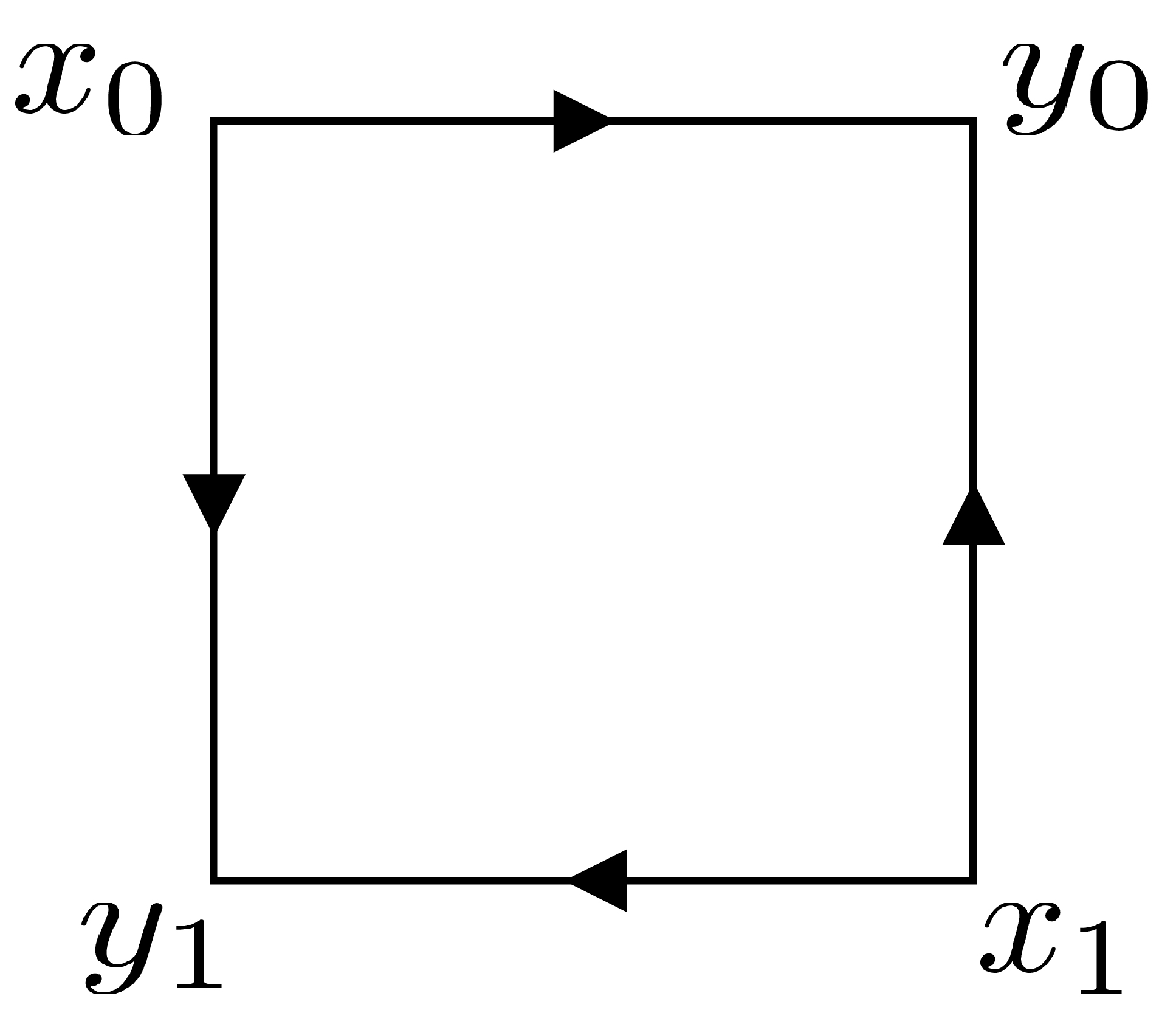} 
\end{center} 
\caption{
}
\label{fig:Bell-discrete}
\end{figure}  
}}

Under the identifications (vertical bijections) of Theorem \ref{thm:ComparisonToSheaf} the support of a distribution $p\in \siDist(X_\cC,\Delta_{\ZZ_d})$ as defined in Eq.~(\ref{eq:support-simp-dist}) coincides with the notion of support for discrete scenarios \cite{abramsky2011sheaf}. 

\Cor{\label{cor:strong-comparison-sheaf}
A nonsignaling distribution is strongly contextual in the sense of \cite{abramsky2011sheaf} if and only if it is strongly contextual in the sense of Definition \ref{def:StronglyContextual} when regarded as a simplicial distribution on $(X_\cC,\Delta_{\ZZ_d})$.}

\bibliography{bib}
\bibliographystyle{ieeetr}  
 

\end{document}